\documentclass[a4paper,12pt]{report}

\pdfoutput=1

\usepackage[paper=letterpaper,margin=1in]{geometry}
\usepackage[english]{babel}
\usepackage{amsmath,amssymb,amsfonts,graphicx,cite,setspace,bigstrut,longtable,array,url,breqn}
\usepackage{braket}
\usepackage{dsfont}
\usepackage[nottoc,notlot,notlof]{tocbibind}
\usepackage{titlesec}
\usepackage{tikz}
\usepackage{amsthm}
\usepackage{hyperref}

\usetikzlibrary{decorations.markings}

\numberwithin{equation}{section}

\newtheorem{thm}{Theorem}[section]
\newtheorem{defn}[thm]{Definition}
\newtheorem{propn}[thm]{Proposition}
\newtheorem{lemma}[thm]{Lemma}
\newtheorem{cor}[thm]{Corollary}
\newtheorem{ex}[thm]{Example}

\newcommand{\ddt}{\frac{\mathrm{d}}{\mathrm{d}t}}
\newcommand{\ddx}{\frac{\mathrm{d}}{\mathrm{d}x}}
\newcommand{\dt}{\mathrm{d}t\,}
\newcommand{\dtau}{\mathrm{d}\tau\,}
\newcommand{\dx}{\mathrm{d}x\,}
\newcommand{\dy}{\mathrm{d}y\,}
\newcommand{\dd}{\mathrm{d}}
\newcommand{\dpsi}{\mathrm{d}\psi}
\newcommand{\dddagger}{\dd^{\dagger}}

\newcommand{\psibar}{\bar{\psi}}
\newcommand{\epsbar}{\bar{\epsilon}}
\newcommand{\ddp}{\frac{\mathrm{d}p}{2\pi\hbar}}
\newcommand{\nn}{\nonumber}
\newcommand{\pathmeasure}{\mathcal{D}[x(t)]\,}
\newcommand{\pathmeasurey}{\mathcal{D}[y(t)]\,}
	
\newcommand{\Det}{\mathrm{Det}\,}

\begin{document}

\begin{titlepage}

~\\
\vskip 2cm

\begin{center}
{\Large \bf Path Integral Methods in Index Theorems}
\end{center}
\bigskip

\vspace{.5cm}
\centerline{
{\large Mark van Loon}
\footnote{mark.vanloon@merton.ox.ac.uk }
}
\vspace*{3.0ex}

\begin{center}
{\it
{
{
Merton College, University of Oxford, \\
Oxford, OX1 4JD, UK.\\
}
}}
\end{center}

\vspace*{1.8cm}
\centerline{\textbf{Abstract}} \bigskip
This paper provides a pedagogical introduction to the quantum mechanical path integral and its use in proving index theorems in geometry, specifically the Gauss-Bonnet-Chern theorem and Lefschetz fixed point theorem. It also touches on some other important concepts in mathematical physics, such as that of stationary phase, supersymmetry and localization. It is aimed at advanced undergraduates and beginning graduates, with no previous knowledge beyond undergraduate quantum mechanics assumed. The necessary mathematical background in differential geometry is reviewed, though a familiarity with this material is undoubtedly helpful.

\end{titlepage}

\chapter*{Preface}
As mentioned in the abstract, this paper provides a pedagogical introduction to the quantum mechanical path integral and its use in proving index theorems in geometry.

Several other works, such as \cite{Mirrorsymmetry, Nakahara, ooguri,  SiLi}, introduce some of the ideas in this paper, but tend to focus on other applications of these concepts. I hope that students will find this paper useful as it provides a single introduction to all these ideas without requiring advanced background knowledge.

The proofs of the Gauss-Bonnet-Chern theorem and Lefschetz Fixed Point theorem in chapter 3 are based on the proof outlines in \cite{Si Li} and \cite{ooguri}. A lot of the details in these proofs are worked out explicitly and some numerical factors that were stated incorrectly in \cite{SiLi} and \cite{ooguri} have been corrected. I hope these proofs are useful to students new to this material and am unaware of any other source that works out these proofs in detail using these methods.
\\
%The author is unaware of a single source that provides an accessible introduction to these concepts starting from just undergraduate quantum mechanics and hopes it will be helpful to advanced undergraduates interested in these ideas.

The paper is based on a dissertation submitted to The University of Oxford in partial fulfilment of the requirements for the degree of Master of Mathematics.

\chapter*{Acknowledgements}
I would like to extend my deepest gratitude to Prof. James Sparks, who supervised me for this dissertation. His suggested reading material, comments and ideas for improvement were invaluable and I would not have been able to put this dissertation together without him.

\tableofcontents

\chapter*{Introduction} \label{chap:intro}
Since its inception by Richard Feynman in the forties, applications of the path integral in physics abound and a $10,000$-word essay can hardly do justice to all of these. While commonly introduced as a gateway to perturbation theory in quantum field theory, or as a useful calculational tool in statistical mechanics, it is also an object that is intrinsically of interest: physically, as it provides new insight into quantum mechanical phenomena and relates quantum and classical physics, and mathematically, as it provides new proofs of index theorems.

The path integral, with the appearance of the action, and its interpretation as a ``sum over all paths", gives a nicer interpretation of what quantum mechanics is fundamentally about. Furthermore, it shows more clearly the correspondence with classical mechanics, and extends easily to quantum field theory, though we shall not consider that here.

It is also very well-suited to problems in supersymmetric quantum mechanics and systems defined on Riemannian manifolds. Specifically, we shall show how the path integral can be used to evaluate certain topological invariants of manifolds through the Gauss-Bonnet-Chern and Lefschetz fixed-point theorems.\\
\\
The overview of this dissertation is as follows. 

Chapter \ref{chap:pathintro} outlines the path integral in quantum mechanics. Section \ref{sec:pathintro} introduces the path integral and deduces some elementary results. Section \ref{sec:schrod} shows the equivalence with the Schr{\"o}dinger formulation and shows some interesting correspondences with classical mechanics. Section \ref{sec:mathcons} looks at some mathematical properties of the path integral. First we consider the stationary phase approximation, which is useful for systems with action $S \gg \hbar$ and provides a ``derivation" of classical mechanics. Secondly, we consider zeta-regularization to assign finite values to otherwise infinite quantities, an idea that is widely used in theoretical physics and other disciplines.

Chapter \ref{chap:mathprelims} gives the necessary mathematical background on Grassmann variables and differential geometry to discuss supersymmetry in chapter \ref{chap:SUSYQM}. The focus lies on differential forms, which are treated in section \ref{sec:diffforms}.

Chapter \ref{chap:SUSYQM} introduces supersymmetric quantum mechanics by analyzing some simple examples. The general structure of supersymmetry is reviewed in section \ref{sec:SUSYstructure}. Section \ref{sec:localization} looks at the property of localization, a method for evaluating exactly certain quantities in supersymmetric models. These ideas are then used in section \ref{sec:geomthms} to give ``physics proofs" of the Gauss-Bonnet-Chern and Lefschetz fixed-point theorems.

\chapter{Path integral approach to quantum mechanics} \label{chap:pathintro}

\section{Introduction to the path integral} \label{sec:pathintro}
In this section we define the propagator and path integral. We will assume a basic understanding of the Schr{\"o}dinger/Heisenberg picture of quantum mechanics, and will show how the Feynman `sum over all paths' emerges from it. Section \ref{sec:schrod} shows that the converse is also true, so that these two formulations are equivalent.

\subsection{Brief review of quantum mechanics}
We start with a brief review of quantum mechanics. For simplicity we will deal with quantum mechanics in $1+1$ dimensions, although the discussion generalises naturally to more spatial dimensions.

Physical states of a particle are described by a Hilbert space $\mathcal{H}$ of kets $\ket{\psi}$. Physical observables are self-adjoint linear operators on $\mathcal{H}$. Two particularly important operators are the position operator $\hat{x}$ and momentum operator $\hat{p}$, which satisfy the canonical commutation relation \cite{Hannabuss}
\begin{equation} 
\left[\hat{x},\hat{p} \right] = i\hbar
\end{equation}
with square brackets indicating a commutator: $[\hat{A},\hat{B}] = \hat{A}\hat{B} - \hat{B}\hat{A}$.

The position eigenstates $\ket{x}$ and momentum eigenstates $\ket{p}$, defined by $\hat{x}\ket{x} = x\ket{x}$ and $\hat{p}\ket{p} = p\ket{p}$, are interpreted as states in which the particle has definite position $x$ or momentum $p$, respectively. \\
\\
A particular realization of $\mathcal{H}$ is $\mathcal{H} = L^2(\mathbb{R})$, in which case the states are square-integrable functions $\psi(x)$ depending on position $x$. The position and momentum operators are $\hat{x} = x$ and $\hat{p} = -i\hbar\frac{\partial}{\partial x}$. We choose the norm of the eigenstates such that:
\begin{align}
\braket{x'|x} &= \delta(x-x') \nn\\
\braket{p'|p} &= 2\pi\hbar\cdot \delta(p'-p)\nn\\
\braket{x|p}  &= \exp(ipx/\hbar) \label{eq:planewave}
\end{align}
where $\delta(y)$ is the Dirac delta function. From equation \ref{eq:planewave} we recognize the momentum eigenstates as plane waves.

The position and momentum eigenstates both independently form a basis for $\mathcal{H}$. Combining this with our choice of normalization above, we get the very useful results:
\begin{equation} \label{eq:identityres}
 \int \dx \ket{x}\bra{x} = \mathds{1} =  \int \ddp \ket{p}\bra{p}
\end{equation}
where the integration is over $\mathbb{R}$.\\
\\
The evolution of states in $\mathcal{H}$ is governed by the Hamiltonian $H$ of the system, a Hermitian operator, through the Schr{\"o}dinger equation: \cite{Hannabuss}
\begin{equation}
i\hbar \ddt \ket{\psi(t)} = H \ket{\psi(t)}.
\end{equation}
This equation can be formally integrated to give evolution from initial time $t_i$ to final time $t_f$
\begin{equation}
\ket{\psi(t_f)} = U(t_f, t_i) \ket{\psi(t_i)}
\end{equation}
where $U(t_f,t_i) = U(t_f-t_i) =  \exp\left(-\frac{i(t_f-t_i)}{\hbar}H\right)$ is the time-evolution operator in the case of a time-independent Hamiltonian \cite{Hannabuss}. \\
\\
An important property of $U$ that immediately follows is that for any $t_i < t < t_f$:
\begin{equation} \label{eq:multiplU}
U(t_f,t_i) = U(t_f,t)U(t,t_i).
\end{equation}

The time-evolution operator gives the probability amplitude for transitioning from an initial state $\ket{\psi_i}$ to final state $\ket{\psi_f}$ as
\begin{equation} \label{eq:transitprob}
\braket{\psi_f|U(t_f,t_i)|\psi_i}.
\end{equation}

\subsection{Propagator and path integral} \label{sec:propdefn}
We can now define the propagator, which shall be the main object of interest to us.

\begin{defn} \label{def:propagator}
The \textbf{propagator} (also called \textbf{kernel}) $K(x_f, t_f; x_i, t_i)$ is the transition amplitude to go from position eigenstate $\ket{x_i}$ at time $t_i$ to position eigenstate $\ket{x_f}$ at time $t_f$: \cite{Mackenzie}
\begin{equation}\label{eq:propagatordef}
K(x_f, t_f; x_i, t_i) = \braket{x_f | U(t_f , t_i) | x_i}.
\end{equation}
\end{defn}

The propagator can be used to calculate the transition probability in \ref{eq:transitprob} by using a resolution of the identity given by \ref{eq:identityres}
\begin{align} \label{eq:transitprob2}
\braket{\psi_f|U(t_f,t_i)|\psi_i} &=\bra{\psi_f}\left(\int \dd x_f\, \ket{x_f}\bra{x_f}\right)U(t_f,t_i)\left(\int \dd x_i\, \ket{x_i}\bra{x_i}\right)\ket{\psi_i} \nonumber \\
&= \int \int \dd x_f \,\dd x_i \braket{\psi_f|x_f}\braket{x_f|U(t_f,t_i)|x_i}\braket{x_i|\psi_i} \nonumber \\
&= \int \int \dd x_f \,\dd x_i\, \psi_f^*(x_f) \psi_i(x_i) K(x_f,t_f;x_i,t_i)
\end{align}
where we used linearity of all operators to move the integrals to the front, and used the particular realization $\mathcal{H} \cong L^2({\mathbb{R}})$. From equation \ref{eq:transitprob2}, we conclude that the propagator uniquely determines all transition probabilities \cite{Blau}. \\
\\
As a simple example, let us evaluate the propagator for a free particle. 
\begin{ex}
\textbf{(Free particle propagator)} \\ The free particle has propagator \cite{Blau, hitoshipath}
\begin{equation}\label{eq:freeparticle}
K(x_f, t_f; x_i, t_i) = \sqrt{\frac{m}{2\pi i \hbar (t_f-t_i)}} \exp\left( \frac{im}{2\hbar}\frac{(x_f-x_i)^2}{t_f-t_i} \right).
\end{equation}

\begin{proof}
The free particle is described by a simple Hamiltonian: $\hat{H} = \frac{\hat{p}^2}{2m}$. Substituting this into definition \ref{eq:propagatordef}, we get propagator:
\begin{align}
K(x_f, t_f; x_i, t_i) &= \braket{x_f | \exp\left(-\frac{i(t_f-t_i)}{2m\hbar}\hat{p}^2\right) | x_i} \nonumber \\
&= \int \ddp \, \braket{x_f | \exp\left(-\frac{i(t_f-t_i)}{2m\hbar}\hat{p}^2\right) |p}\braket{p| x_i} 
\end{align}
where we used equation \ref{eq:identityres} to insert a resolution of the identity. Note that $\ket{p}$ is an eigenstate of $\hat{p}$ and hence of $\exp\left(-\frac{i(t_f-t_i)}{2m\hbar}\hat{p}^2\right)$, with eigenvalue $\exp\left(-\frac{i(t_f-t_i)}{2m\hbar}p^2\right)$. Further recognize the plane wave from equation \ref{eq:planewave}: $\braket{p|x_i} = \exp(-ipx_i/\hbar)$. Thus:
\begin{align}
K(x_f, t_f; x_i, t_i) &= \int \ddp \, \exp \left( -\frac{i(t_f-t_i)}{2m\hbar}p^2+\frac{i(x_f - x_i)}{\hbar}p\right) \nonumber \\
&= \sqrt{\frac{m}{2\pi i \hbar (t_f-t_i)}} \exp\left( \frac{im}{2\hbar} \frac{(x_f-x_i)^2}{t_f-t_i} \right).
\end{align}
Here we made use of the following identity:
\begin{equation} \label{eq:Fresnel}
\int \dx \exp\left(-\frac{1}{2}iax^2 + bx\right) = \left(\frac{2\pi}{ai}\right)^{\frac{1}{2}}\exp\left(-\frac{i}{2a}b^2\right).
\end{equation}
This follows by completing the square and using the Fresnel integral formula in appendix \ref{sec:fresnel}.

\end{proof}
\end{ex}

A noteworthy point is that the term in the exponential is exactly $\frac{i}{\hbar} S[x_c(t)]$, where $S[x_c(t)]$ is the action of the classical path. \cite{Mackenzie} As we shall see later, this is no coincidence, but a result of the Lagrangian describing the system being at most quadratic in the position.\\
\\
As the propagator represents propagation from an initial state to a final state, we would expect that we can express propagation from time $t_i$ to $t_f$ by propagation first from $t_i$ to intermediate time $t$ and then from $t$ to $t_f$ (where $t_i < t < t_f$). This result is known as the \emph{convolution property}: 

\begin{propn} \label{propn:convolution}
\textbf{(Convolution property)} For any $t_i<t<t_f$, the propagator satisfies: \cite{hitoshipath, Mackenzie}
\begin{equation}
K(x_f, t_f; x_i, t_i) = \int \dx \,\, K(x_f, t_f; x, t)K(x,t; x_i, t_i).
\end{equation}
\begin{proof}
We use definition \ref{def:propagator}, equation \ref{eq:multiplU} and a resolution of the identity from equation \ref{eq:identityres}
\begin{align}
K(x_f, t_f; x_i, t_i) &= \braket{x_f | U(t_f - t_i) | x_i} \nonumber \\
&= \braket{x_f|U(t_f-t) \mathds{1} U(t-t_i) |x_i} \nonumber \\
&= \int \dx \, \bra{x_f}U(t_f-t) \ket{x}\bra{x} U(t-t_i) \ket{x_i} \nonumber \\
&= \int \dx \, K(x_f, t_f; x, t)K(x,t; x_i, t_i)
\end{align}
where we used the fact that $\bra{x_f}U(t_f-t)$ is a linear operator to move the integral out of the inner product.
\end{proof}
\end{propn}

The free particle propagator found earlier can be checked to satisfy this equation. The proposition has an immediate corollary:

\begin{cor} \label{corr:smalltimepropagator}
Let $[t_i,t_f]$ be a time interval and let $N \in \mathbb{N}$. Define $\epsilon = \frac{1}{N} (t_f - t_i)$ and for $j=0,1,\dots,N$ let $t_j = t_i + j\epsilon $, so that the time interval is partitioned into $N-1$ time intervals of length $\epsilon$.

Then we have expression for the propagator:
\begin{align} \label{eq:smalltimepropagator}
K(x_f, t_f; x_i, t_i) = \int \dd x_1 &\dots \dd x_{N-1}\, K(x_f, t_f; x_{N-1}, t_{N-1}) K(x_{N-1}, t_{N-1}; x_{N-2}, t_{N-2}) \times\nn\\
&\times  \dots \times K(x_2,t_2; x_1,t_1) K(x_1,t_1; x_i,t_i)
\end{align}
\begin{proof}
This follows from repeated applications of proposition \ref{propn:convolution}.
\end{proof}
\end{cor}

The convolution property proves crucial in the definition of the path integral as it allows us to calculate the propagator by splitting the time interval $[t_i,t_f]$ into smaller intervals of length $\epsilon$ and taking the limit $\epsilon \rightarrow 0$. \\
\\
We now work out what such a propagator looks like for a Hamiltonian of the form
\begin{equation} \label{eq:simplehamiltonian}
\hat{H} = \hat{T} + \hat{V} = \frac{\hat{p}^2}{2m} + V(\hat{x}).
\end{equation}
This Hamiltonian describes a particle with mass $m$ moving in a potential given by $V(x)$.

To find the propagator, we take the partition of the time interval $[t_i,t_f]$ as defined in corollary \ref{corr:smalltimepropagator}. Note that the time-evolution operator satisfies for any $N \in \mathbb{N}$:
\begin{align} \label{eq:smalltimeevol}
U(t_f,t_i) = \exp\left(-\frac{i(t_f-t_i)}{\hbar}\hat{H}\right) &= \left(\exp\left(-\frac{i(t_f-t_i)}{N\hbar}\hat{H}\right)\right)^N \nonumber \\
&= \left(\exp\left(-\frac{i\epsilon}{\hbar}\hat{H}\right)\right)^N \nonumber\\
&= \left(\exp\left(-\frac{i\epsilon}{\hbar}(\hat{T} + \hat{V})\right)\right)^N
\end{align}
where again $\epsilon = \frac{t_f - t_i}{N}$.\\
\\
Due to the non-commutativity of $\hat{T}$ and $\hat{V}$, we cannot simply expand this exponential into a product of two exponentials, i.e. the equation 
\begin{equation} \label{eq:wrongexponential}
\exp\left(-\frac{i\epsilon}{\hbar}(\hat{T} + \hat{V})\right) = \exp\left(-\frac{i\epsilon}{\hbar}\hat{T}\right)\exp\left(-\frac{i\epsilon}{\hbar}\hat{V}\right)
\end{equation} 
does not hold in general.

However, this equation is approximately correct, the error being $\mathcal{O}(\epsilon^2)$ \cite{Blau}. See appendix \ref{sec:CBH} for the details. As we are interested in the limit $\epsilon \rightarrow 0$, we can discard these terms and approximate the small-time propagator by:
\begin{equation} \label{eq:smalltimeprop1}
\braket{x_{j+1}|U(t_{j+1},t_j)|x_{j}} \approx \braket{x_{j+1}|e^{-\frac{i\epsilon}{\hbar}\hat{T}}e^{-\frac{i\epsilon}{\hbar}\hat{V}}|x_{j}} = \braket{x_{j+1}|e^{-\frac{i\epsilon}{\hbar}\hat{T}}|x_{j}} e^{-\frac{i\epsilon}{\hbar}V(x_j)}
\end{equation}
where we used that $\ket{x_j}$ is an eigenstate of the operator $\hat{V}$ and hence of $e^{-\frac{i\epsilon}{\hbar}\hat{V}}$.\\
\\
Now insert an identity: $\mathds{1} = \int \ddp \, \ket{p}\bra{p}$, to get
\begin{align} \label{eq:smalltimeprop2}
\braket{x_{j+1}|e^{-\frac{i\epsilon}{\hbar}\hat{T}}|x_{j}} &= \int \ddp \, \braket{x_{j+1}|e^{-\frac{i\epsilon}{2m\hbar}\hat{p}^2}|p}\braket{p|x_{j}} \nonumber \\
&= \int \ddp \, \exp \left( -\frac{i\epsilon}{2m\hbar}p^2+\frac{i(x_{j+1} - x_j)}{\hbar}p\right) \nonumber \\
&= \left(\frac{m}{2i\pi\hbar\epsilon}\right)^{\frac{1}{2}} \exp\left(\frac{im}{2\hbar} (x_{j+1} - x_j)^2 \right)
\end{align}
where we used that $\ket{p}$ is an eigenstate of $\hat{T}$ and that $\braket{x|p} = e^{ipx/\hbar}$. Furthermore we used equation \ref{eq:Fresnel} to evaluate the Fresnel integral.\\
\\
We use these equations \ref{eq:smalltimeprop1} and \ref{eq:smalltimeprop2} in \ref{corr:smalltimepropagator} to finally obtain the following expression for the propagator:
\begin{align} \label{eq:propagatorlim}
  K(x_f,t_f;&x_i,t_i) = \lim_{N\rightarrow \infty} \left(\frac{mN}{2i\pi\hbar(t_f-t_i)}\right)^{\frac{N}{2}} \times \nn\\
&\times\int \dd x_1\dots\dd x_{N-1} \exp\left[\frac{i\epsilon}{\hbar} \sum_{j=0}^{N-1}\left( \frac{m}{2} \left(\frac{x_{j+1} - x_j}{\epsilon}\right)^2 - V(x_j) \right) \right] 
\end{align}
with $x_0 = x_i$ and $x_N = x_f$.\\
\\
This equation gives a correct formal expression for the propagator. One can use it to calculate propagators for any system, but the calculations are usually prohibitive and we resort to other ways of finding the propagator. For example, we can check our earlier expression for the free particle propagator, but this would take several pages and involves lots of non-trivial trigonometric identities (see \cite{hitoshipath} for details). A useful analogy is calculus: we rarely use the technical definitions of derivatives and integrals, instead resorting to theorems characterizing their properties, such as the product rule, chain rule and the fundamental theorem of calculus. \\
\\
Now note that in the exponent in equation \ref{eq:propagatorlim}, we have a term:
\begin{equation}
\epsilon \sum_{j=0}^{N-1} \left[ \frac{m}{2}\left(\frac{x_{j+1} - x_j}{\epsilon}\right)^2 - V(x_j)\right].
\end{equation}
We then take a limit $\epsilon \rightarrow 0$ over the whole integral. Assuming we can move this limit through the integration measure and the exponential (the latter is possible by continuity of $\exp$), we get an integral as a limit of a Riemann sum. We use that $\lim_{\epsilon \rightarrow 0} \left(\frac{x_{j+1} - x_j}{\epsilon}\right) = \dot{x_j}$ to get as exponent:
\begin{equation}
\frac{i}{\hbar} \epsilon \sum_{j=0}^{N-1} \left[ \frac{m}{2}\left(\frac{x_{j+1} - x_j}{\epsilon}\right)^2 - V(x_j)\right] \rightarrow \frac{i}{\hbar} \int_{t_i}^{t_f} \dt \, L(x(t),\dot{x}(t)) = \frac{i}{\hbar} S[x(t)]
\end{equation}
where we recognized the Lagrangian $L(t) = L(x(t),\dot{x}(t)) = \frac{m}{2} \dot{x}(t)^2 - V(x(t)) $ and used the definition of the action: $S = \int \dt L(t)$. Recall we already encountered an exponential of $\frac{i}{\hbar} S$ in our free particle example.\\
\\
We now define the path integral as the limit of integrals appearing in our expression for the propagator.

\begin{defn} \label{def:pathintegral}
Let $(t_i,x_i)$ and $(t_f,x_f)$ be two points in spacetime such that $t_i < t_f$. Let $t_0,t_1,\dots,t_N$ be a partition of $[t_i,t_f]$ into intervals of length $\epsilon = \frac{t_f-t_i}{N}$ and define $x_j = x(t_j)$.  We formally define:\cite{Blau}
\begin{align}
\mathcal{N} &= \lim_{N\rightarrow \infty} \left(\frac{mN}{2i\pi\hbar(t_f-t_i)}\right)^{\frac{N}{2}} \label{eq:Ndef} \\
\int_{x(t_i) = x_i}^{x(t_f) = x_f} \mathcal{D}[x(t)] &= \lim_{N\rightarrow \infty} \int \dd x_1 \dots \dd x_{N-1}. \label{eq:pathintegraldef}
\end{align} 

The \textbf{path integral} is formally defined to be: 
\begin{align}
& \int_{x(t_i) = x_i}^{x(t_f) = x_f} \mathcal{D}[x(t)] \exp\left(\frac{i}{\hbar} S[x(t)] \right) = \nonumber \\
&= \lim_{N\rightarrow \infty} \int \dd x_1 \dots \dd x_{N-1} \exp\left[\frac{i\epsilon}{\hbar} \sum_{j=0}^{N-1}\left( \frac{m}{2} \left(\frac{x_{j+1} - x_j}{\epsilon}\right)^2 - V(x_j) \right) \right] 
\end{align}
where we identify $x_0 = x_i$ and $x_N = x_f$.\\

The propagator then satisfies:
\begin{equation} \label{eq:proppathint}
\boxed{ K(x_f,t_f;x_i,t_i) = \mathcal{N} \int \mathcal{D}[x(t)] \exp\left(\frac{i}{\hbar} S[x(t)] \right) }
\end{equation}
with implicit limits of integration $x(t_i) = x_i, x(t_f) = x_f$.
\end{defn}

One should take this ``integral" as a formal construct, not as an integral in a strict mathematical sense. In fact, the path integral does not exist in a strict mathematical sense, as our ``measure" $\mathcal{D}[x(t)]$ is not a measure that can be imposed on the space of all paths \cite{Blau}.

Furthermore, our definition of $\mathcal{N}$ implies it is an infinite constant. The propagator is finite though, so equation \ref{eq:proppathint} implies that the path integral must be zero! However, this is not a problem as the product of $\mathcal{N}$ and $\mathcal{D}x$ is the relevant quantity, and we only ever see them together. In certain situations, such as when calculating scattering amplitudes \cite{Mackenzie}, one is concerned with ratios of path integrals, which actually are finite.\\
\\
For the physical point of view: we can think of this path integral as summing $\exp\left(\frac{i}{\hbar} S[x(t)]\right)$ over all paths $x(t)$ between $(t_i,x_i)$ and $(t_f,x_f)$. Equivalently all paths $x(t)$ are contributing to the probability amplitude of propagation, each path weighted by the phase $\exp\left(\frac{i}{\hbar} S[x(t)]\right)$.

This is the origin of the idea of a ``Feynman sum over all histories", the histories referring to different paths a particle can take. This is sometimes phrased as \emph{``a particle takes all possible paths between two points"}, though this can be slightly misleading. Firstly, while all paths contribute to the propagation amplitude, they are weighted by a phase depending on the action of the path. Secondly, quantum theory is inherently a theory of \emph{measurement} and if we only measure at times $t_i$ and $t_f$ - and not at any intermediate time $t$ - then the question ``Where was the particle at time $t$?" is the wrong question to ask within the quantum mechanical framework. 

\subsubsection{Free particle and normalization $\mathcal{N}$} \label{sec:freeparticle}
In this section we revisit the free particle, now evaluating the propagator via a path integral.

First we need to define the determinant of an infinite-dimensional matrix. 
\begin{defn} \label{defn:det}
Let $\mathcal{H}$ be a separable Hilbert space and $A: \mathcal{H} \rightarrow \mathcal{H}$ an elliptic, self-adjoint linear operator with a complete set of eigenvectors, with associated eigenvalues $\{ \lambda_n \}_{n= 1}^{\infty}$. Analogously to the finite-dimensional case, we define the determinant of $A$ as the product of its eigenvalues:
\begin{equation}
\Det A = \prod_{n= 1}^{\infty} \lambda_n.
\end{equation}
\end{defn}
Note that generally this determinant is infinite. However, this is not a problem, as we shall only be interested in ratios of determinants. 

In section \ref{sec:zeta}, we shall see another type of determinant, which is made finite by employing \emph{zeta-regularization}.

\begin{ex} \textbf{(Free particle revisited)} \label{ex:freeparticle2}
Recall our discussion of the free particle, in which we found the propagator:
\begin{equation} \label{eq:freeparticle2}
K(x_f, t_f; x_i, t_i) = \sqrt{\frac{m}{2\pi i \hbar (t_f-t_i)}} \exp\left( \frac{im}{2\hbar} \frac{(x_f-x_i)^2}{t_f-t_i} \right).
\end{equation}
We show that we also have expression:
\begin{equation} \label{eq:freeparticle3}
K(x_f, t_f; x_i, t_i) = \mathcal{N} \exp\left(\frac{im}{2\hbar} \frac{(x_f-x_i)^2}{t_f-t_i} \right)\, \Det \left(-\frac{m}{2\pi i \hbar}\partial_t^2\right)^{-1/2}.
\end{equation}
Equating these two, we get the important result:
\begin{equation} \label{eq:Nexpression}
\boxed{ \mathcal{N} = \sqrt{\frac{m}{2\pi i \hbar (t_f-t_i)}} \sqrt{\Det \left(-\frac{m}{2\pi i \hbar}\partial_t^2\right)} . }
\end{equation}
\end{ex}
\begin{proof}
First we find the classical path $x_c(t)$. This is easily found as 
\begin{equation}
x_c(t) = x_i + \frac{t - t_i}{t_f-t_i} (x_f - x_i)
\end{equation}
with associated classical action
\begin{equation}
S[x_c(t)] = \frac{m}{2}\frac{(x_f-x_i)^2}{t_f-t_i}.
\end{equation}
Expand paths around the classical path $x(t) = x_c(t) + y(t)$, so that $y(t)$ satisfies the boundary conditions 
\begin{equation} \label{eq:yboundary}
y(t_i) = 0 = y(t_f).
\end{equation}
We put this in the definition of the path integral:
\begin{align} 
K(x_f, t_f; x_i, t_i) &= \mathcal{N}\int \pathmeasure \exp\left(\frac{i}{\hbar} S[x(t)] \right)  \nn\\
&= \mathcal{N}\exp\left(\frac{i}{\hbar} S[x_c(t)] \right) \int \pathmeasurey \exp\left(\frac{i}{\hbar} S[y(t)] \right)  \nn\\
&= \mathcal{N}\exp\left(\frac{im}{2\hbar}\frac{(x_f-x_i)^2}{t_f-t_i} \right)\int \pathmeasurey \exp\left(\frac{im}{2\hbar}\int_{t_i}^{t_f} \dt \dot{y}^2(t) \right).
\end{align}
Now use integration by parts, noting that the boundary term vanishes by equation \ref{eq:yboundary}. Hence:
\begin{align} \label{eq:freeparticlecalc}
K(x_f, t_f; x_i, t_i) &=  \mathcal{N} \exp\left(\frac{im}{2\hbar} \frac{(x_f-x_i)^2}{t_f-t_i} \right)\int \pathmeasurey \exp\left(\frac{im}{2\hbar}\int_{t_i}^{t_f} \dt y(t) (-\partial_t^2) y(t) \right)\nn\\
&=\mathcal{N} \exp\left(\frac{im}{2\hbar}\frac{(x_f-x_i)^2}{t_f-t_i} \right)\, \Det \left(-\frac{m}{2\pi i \hbar}\partial_t^2\right)^{-1/2}
\end{align}
where we used the familiar Fresnel integral formula. This gives the stated result for $\mathcal{N}$.
\end{proof}

\subsection{Imaginary time propagator}
One may worry about convergence issues relating to the propagator as an integral of $\exp\left( i S[x(t)] /\hbar\right)$, which has unit modulus.

Given an initial time $t_i$, we have defined the propagator for any time $t_f > t_i$. Assuming the propagator is suitably analytic, we can extend its definition into the complex plane to get the imaginary time propagator in terms of $\tau = i t$. This is called a \emph{Wick rotation} and is how the path integral relates quantum (field) theory and statistical mechanics \cite{Mackenzie}.

\begin{defn}
We define the imaginary time propagator (or Euclidean propagator) from $x_i$ to $x_f$ to be
\begin{equation}
K_E(x_f, \tau_f; x_i, \tau_i) = \braket{x_f| \exp\left(-\frac{1}{\hbar}(\tau_f - \tau_i) H\right)|x_i}
\end{equation}
where $\tau_i=i t_i,\tau_f = i t_f \in i\mathbb{R}$ and $t_f > t_i$.

Assuming suitable conditions on $H$ so that this is analytic (except for a possible pole when $\tau_f = \tau_i$), it is related to the normal propagator by \cite{hitoshipath, Mackenzie}
\begin{equation}
K(x_f, t_f; x_i, t_i) = K_E(x_f, it_f; x_i, it_i).
\end{equation}
\end{defn}

The advantage is that for the common Hamiltonian $L = \frac{1}{2m} \hat{p}^2 + \hat{V}(x)$, the oscillatory $\exp(iS/\hbar)$ is replaced by a negative exponential: $\exp(iS/\hbar) \rightarrow \exp(-S_E/\hbar)$ where 
\begin{equation}
S_E = \int_{\tau_i}^{\tau_f} \dtau \left[ \frac{1}{2m} \left( \frac{\dd x}{\dtau} \right)^2 + V(x(\tau)) \right].
\end{equation}

As a result, quantities are generally better behaved when working in imaginary time, making it useful in calculations. We shall see this when we examine the harmonic oscillator in section \ref{sec:harmosc}.\\
\\
Note that the normalization constant $\mathcal{N}$ changes to 
\begin{align}
\mathcal{N} &\rightarrow  \lim_{N\rightarrow \infty} \left(\frac{mN}{2i\pi\hbar(\tau_f-\tau_i)}\right)^{\frac{N}{2}} \nn\\
&=\left[ \lim_{N\rightarrow \infty} \left(\frac{1}{i}\right)^{\frac{N}{2}} \right]\left[\lim_{N\rightarrow \infty} \left(\frac{mN}{2i\pi\hbar(t_f-t_i)}\right)^{\frac{N}{2}} \right] \nn\\
&= \mathcal{N} \prod_{j=1}^{\infty} \frac{1}{\sqrt{i}}. \label{eq:imNexpression}
\end{align}
For now we should understand this as a formal expression; in section \ref{sec:zeta} we discuss a way to regulate this.\\
\\
There is an important link between the Euclidean propagator and traces, such as those encountered in statistical mechanics. Consider an operator $\mathcal{O}$ and let $\beta >0$. Then
\begin{equation}
\mathrm{Tr} \left( \mathcal{O} e^{-\beta H} \right) = \int \dd x\, \braket{x|\mathcal{O}e^{-\beta H}|x}.
\end{equation}
From this we note the link with the Euclidean propagator by putting $\mathcal{O} = \mathds{1}$:
\begin{equation}
\braket{x|e^{-\beta H}|x} = K_E (x,\beta\hbar,x,0) = \mathcal{N}\int\mathcal{D}[x(\tau)] \exp \left( -\frac{1}{\hbar}S_E[x(\tau)] \right)
\end{equation}
where the integration is over a \emph{periodic} path: $x(0) = x(\beta\hbar) = x$ for fixed $x$. Varying all possible $x$ gives a path integral expression for the \textbf{partition function} \cite{Blau, hitoshipath, Mackenzie, Mirrorsymmetry}
\begin{equation}
\boxed{Z(\beta) \equiv \mathrm{Tr}\left(e^{-\beta H}\right) = \mathcal{N} \int_{x(0) = x(\beta\hbar)} \mathcal{D}[x(\tau)] \exp \left( -\frac{1}{\hbar}S_E[x(\tau)] \right).}
\end{equation}
We will use the path integral to calculate traces when we discuss the Witten index in chapter \ref{chap:SUSYQM}.

\subsection{Semi-classical approximation and harmonic oscillator}\label{sec:harmosc}
In this section, we analyse the harmonic oscillator, one of the few exactly solvable systems in the path integral formalism and one that is ubiquitous in physics. It is described by the Lagrangian density:
\begin{equation} \label{eq:harmosclagrangian}
\mathcal{L} = \frac{1}{2}m \left( \dot{x}^2(t) - \omega^2(t) x^2(t) \right).
\end{equation}

Often we are interested in the case of constant $\omega(t) = \omega_0$, which we consider separately.

First we show how it arises as an approximation to other systems through the semi-classical approximation. \\
\\
Consider a general system described by some action $S[x(t)] = \int_{t_i}^{t_f} L(t,x(t),\dot{x}(t))$. From Part A Classical Mechanics, this has a classical solution $x_c(t)$ extremizing the action:
\begin{equation}
\left. \frac{\delta S}{\delta x(t)}\right|_{x(t) = x_c(t)} = 0.
\end{equation}

We expand the action $S[x(t)]$ around the classical solution:
\begin{equation}
x(t) = x_c(t) + \sqrt{\hbar} y(t).
\end{equation}
The factor of $\sqrt{\hbar}$ is included to elucidate the dependence on powers of $\hbar$.

Now Taylor expand $S[x(t)]$ around $x_c$:
\begin{align}
S[x(t)] &= S[x_c] + \sqrt{\hbar} \left. \frac{\delta S}{\delta x(t)}\right|_{x(t) = x_c(t)} \delta y(t) +  \left. \frac{1}{2} \hbar \frac{\delta S}{\delta x(t) \delta x(t')}\right|_{x(t) = x_c(t)} y(t) y(t') + \mathcal{O}(\hbar^{3/2}) \nn\\
&= S[x_c] + \left. \frac{1}{2} \hbar \frac{\delta S}{\delta x(t) \delta x(t')}\right|_{x(t) = x_c(t)} y(t) y(t') + \mathcal{O}(\hbar^{3/2}) .
\end{align}
Ignoring terms of $\mathcal{O}(\hbar^{3/2})$, i.e. approximating
\begin{equation}
S[x(t)] \approx S[x_c] + \left. \frac{1}{2} \hbar \frac{\delta S}{\delta x(t) \delta x(t')}\right|_{x(t) = x_c(t)} y(t) y(t')
\end{equation}
is what we refer to as the \textbf{semi-classical approximation} \cite{Blau,Mackenzie}.\\
\\
For the common Lagrangian $L(t) = \frac{1}{2} m (\dot{x}(t))^2  - V(x(t))$, and for small $\hbar$, we get approximate action:
\begin{equation} \label{eq:semiclassical}
S[x_c(t) + y(t)] \approx S[x_c] + \hbar \int \dt \left( \frac{1}{2} m (\dot{y}(t))^2  - V''(x_c(t))y^2(t) \right)
\end{equation}
Most physical situations of interest take place around a minimum $x_c(t)$ of the potential, i.e. $V''(x_c(t)) > 0$. A comparison with equation \ref{eq:harmosclagrangian} reveals that the approximate action \ref{eq:semiclassical} is that of a harmonic oscillator with (angular) frequency $\omega(t) = \sqrt{V''(x_c(t))}$, thus showing the importance of evaluating this path integral.

\begin{ex} \label{ex:harmosc} \textbf{(Harmonic oscillator)} The harmonic oscillator with Lagrangian density
\begin{equation}
\mathcal{L} = \frac{1}{2}m \left( \dot{x}^2(t) - \omega^2(t) x^2(t) \right)
\end{equation}
has the propagator
\begin{equation} \label{eq:harmoscprop}
K(x_f, t_f; x_i, t_i) = \sqrt{\frac{m}{2\pi i \hbar (t_f-t_i)}} \sqrt{ \frac{ \Det \left(-\partial_{\tau}^2\right)}{\Det \left(-\partial_{\tau}^2 + \tilde{\omega}^2(\tau)) \right)}}\exp \left(\frac{i}{\hbar} S[x_c] \right)
\end{equation}
where $x_c$ is the classical path, $\tau = it$ and $\tilde{\omega}(\tau) = \omega(t)$.\\
\\
For a time-independent harmonic oscillator $(\omega(t) = \omega_0)$: \cite{Blau, hitoshipath}
\begin{equation} \label{eq:timeindharmoscprop}
K(x_f, t_f; x_i, t_i) = \sqrt{\frac{m}{2\pi i \hbar (t_f-t_i)}}\sqrt{\frac{\omega_0 (t_f-t_i)}{\sin(\omega_0(t_f-t_i))}} \exp \left(\frac{i}{\hbar} S[x_c] \right)
\end{equation}
\end{ex}
\begin{proof}
We find the Euclidean propagator with $\tau_f = i t_f, \tau_i = i t_i$ by expanding paths around the classical path: $\tilde{x}(\tau) = \tilde{x}_c(\tau) + \tilde{y}(\tau)$ where $\tilde{x}(\tau) = x(t)$ with boundary conditions $\tilde{y}(\tau_i) = 0 = \tilde{y}(\tau_f)$. Again we integrate by parts:
\begin{align}
K_E(x_f, \tau_f; x_i, \tau_i) &= \mathcal{N}_E \exp\left(-\frac{1}{\hbar} S_E[\tilde{x}_c] \right) \times \nn\\
&\times \int \mathcal{D}\tilde{y}[\tau] \exp\left(-\frac{m}{2\hbar}\int_{\tau_i}^{\tau_f} \dtau \tilde{y}(\tau) (-\partial_{\tau}^2 +\tilde{\omega}(\tau)^2 ) \tilde{y}(\tau) \right).
\end{align}
where $\tilde{y}(\tau) = y(t)$, $\tilde{\omega}(\tau) = \omega(t)$.

We substitute our expression 
\begin{equation}
\mathcal{N}_E = \sqrt{\frac{m}{2\pi i \hbar (\tau_f-\tau_i)}}  \sqrt{\Det\left(-\frac{m}{2\pi\hbar}\partial_{\tau}^2\right)}
\end{equation}
to get
\begin{equation}
K_E(x_f, \tau_f; x_i, \tau_i) = \sqrt{\frac{m}{2\pi i \hbar (\tau_f-\tau_i)}} \exp\left(-\frac{1}{\hbar} S_E[\tilde{x}_c] \right) \sqrt{\frac{\Det\left(-\frac{m}{2\pi\hbar}\partial_{\tau}^2\right)}{\Det \left(\frac{m}{2\pi\hbar}\left(-\partial_{\tau}^2 + \tilde{\omega}^2(\tau)\right) \right) }}
\end{equation}
Now we use $\exp\left(-\frac{1}{\hbar} S_E[\tilde{x}_c] \right) = \exp \left(\frac{i}{\hbar} S[x_c] \right)$ and the relation \\$K(x_f,t_f;x_i,t_i) = K_E(x_f,\tau_f;x_i,\tau_i)$. Furthermore we cancel the constants $\frac{m}{2\pi\hbar}$ from inside the determinants, as they both yield the same multiplicative constant $\prod_{j=1}^{\infty} \frac{m}{2\pi\hbar}$, to get
\begin{equation}
K(x_f, t_f; x_i, t_i) = \sqrt{\frac{m}{2\pi i \hbar (t_f-t_i)}} \sqrt{ \frac{ \Det \left(-\partial_{\tau}^2\right)}{\Det \left(-\partial_{\tau}^2 + \tilde{\omega}^2(\tau)) \right)}}\exp \left(\frac{i}{\hbar} S[x_c] \right)
\end{equation}
which proves equation \ref{eq:harmoscprop}.

Note the operators $-\partial_{\tau}^2$ and $-\partial_{\tau}^2+\tilde{\omega}(\tau)^2$ are positive-definite, so these determinants are well-defined. \\
\\
To get the time-independent solution \ref{eq:timeindharmoscprop}, we evaluate these determinants explicitly. 

Consider the operator $-\partial_{\tau}^2$ with boundary conditions $\tilde{y}(\tau_i) = 0 = \tilde{y}(\tau_f)$. It has eigenfunctions $\sin(\sqrt{\lambda_n} \tau)$ and $\cos(\sqrt{\lambda_n} \tau)$ with eigenvalues $\lambda_n$. From standard results of Fourier analysis, we know the sines (together with $y(\tau) = 1$) form an orthonormal basis for the Hilbert space of functions $\tilde{y}(\tau)$ on $[\tau_i,\tau_f]$ with boundary conditions $\tilde{y}(t_i) = 0 = \tilde{y}(t_f)$.  For notational clarity, define $\Delta \tau = \tau_f - \tau_i = i(t_f - t_i)$. The boundary conditions impose that
\begin{equation}
\sqrt{\lambda_n} = \frac{\pi n}{\Delta \tau}, \qquad n \in \mathbb{N}.
\end{equation}
Hence we see that
\begin{equation}
\Det (-\partial_{\tau}^2) = \prod_{n=1}^{\infty} \frac{\pi^2n^2}{(\Delta \tau)^2}.
\end{equation}
By similar considerations, we get
\begin{equation}
\Det (-\partial_{\tau}^2 + \omega_0^2) = \prod_{n=1}^{\infty} \left( \omega_0^2 + \frac{\pi^2n^2}{(\Delta \tau)^2}  \right).
\end{equation}
Therefore their ratio satisfies:
\begin{align} \label{eq:detratio}
\frac{ \Det (-\partial_{\tau}^2+\omega_0^2) }{\Det (-\partial_{\tau}^2)} &= \prod_{n=1}^{\infty} \left( \frac{\frac{\pi^2n^2}{(\Delta \tau)^2} + \omega_0^2}{\frac{\pi^2n^2}{(\Delta \tau)^2}} \right) \nn\\
&= \prod_{n=1}^{\infty} \left( 1 + \frac{\omega_0^2 (\Delta \tau)^2}{\pi^2n^2} \right) \nn\\
&= \prod_{n=1}^{\infty} \left( 1 - \frac{\omega_0^2 (t_f-t_i)^2}{\pi^2n^2} \right) \nn\\
&= \frac{\sin(\omega_0(t_f-t_i))}{\omega_0(t_f-t_i)}
\end{align}
where we used the relation $\Delta \tau = i(t_f-t_i)$ and the infinite product representation:
\begin{equation}
\frac{\sin(z)}{z} = \prod_{n=1}^{\infty} \left( 1 - \frac{z^2}{\pi^2n^2} \right) .
\end{equation}
Substituting equation \ref{eq:detratio} into equation \ref{eq:harmoscprop} gives the claimed result.
\end{proof}

\section{Schr{\"o}dinger equation from the path integral} \label{sec:schrod}
Having derived the path integral picture of quantum mechanics from the Schr{\"o}dinger formulation, we will now show the correspondence goes both ways. This shows that quantum mechanics can actually be defined in terms of the path integral, rather than by imposing the initially rather mysterious Schr{\"o}dinger equation. \\
\\
We follow the derivations in \cite{Blau} and \cite{hitoshipath}, the central idea being variations of the paths.\\
\\
We assume definition \ref{def:pathintegral} of the path integral and propagator. Consider a variation of the path:
\begin{equation} \label{eq:xvar1}
x(t) \rightarrow x(t) + \delta x(t).
\end{equation}
Such a variation leaves the path integral unchanged by invariance of the integration measure: $\mathcal{D}[x(t) + \delta x(t)] = \mathcal{D}[x(t)]$.\\
\\
We shall need a lemma from Part A Calculus of Variations:
\begin{lemma} \label{lemma:actionvar}
For a variation of the path $x(t)$ as in equation \ref{eq:xvar1}, the variation of $S$ is:
\begin{equation} \label{eq:actionvar}
\delta S[x(t)]  = \left[ \frac{\partial L}{\partial \dot{x}} \delta x(t) \right]^{t_f}_{t_i} + \int_{t_i}^{t_f} \dt\left( \frac{\partial L}{\partial x} - \ddt \frac{\partial L}{\partial \dot{x}} \right) \delta{x} .
\end{equation}
\end{lemma}

\begin{proof}
We use the definition $S = \int \dt \,L(t, x(t),\dot{x}(t))$ and the fact that $L$ does not depend explicitly on time:
\begin{equation} 
\delta S[x(t)] = \int \dt \, \delta\left( L(x(t),\dot{x}(t))\right) = \int_{t_i}^{t_f}\dt \left( \frac{\partial L}{\partial x}\delta x + \frac{\partial L}{\partial \dot{x}} \delta \dot{x} \right) .
\end{equation}
The result now follows by integrating by parts.
\end{proof}

Some important results follow.

\begin{propn}
The following relation, which is a path integral version of Ehrenfest's theorem, holds: \cite{hitoshipath}
\begin{equation} \label{eq:ehrenfest}
\int \pathmeasure \left( \frac{\partial L}{\partial x} - \ddt \frac{\partial L}{\partial \dot{x}} \right) \exp\left(\frac{i}{\hbar} S[x(t)] \right) = 0.
\end{equation}
Furthermore  
\begin{equation} \label{eq:momentum}
\frac{\partial}{\partial x_f} \int \pathmeasure \exp\left(\frac{i}{\hbar} S[x(t)] \right) = \frac{i}{\hbar} p_f\int \pathmeasure  \exp\left(\frac{i}{\hbar} S[x(t)] \right)
\end{equation}
where $p_f$ is defined by $p_f = p(t_f) = \frac{\partial L}{\partial \dot{x}}(t_f)$.
\end{propn}
\begin{proof}
Let us consider variations keeping the endpoints $x_i$ and $x_f$ fixed: 
\begin{equation} \label{eq:xvar}
\delta x(t_i) = 0 = \delta x(t_f).
\end{equation}

Using lemma \ref{lemma:actionvar}, we see that
\begin{equation} \label{eq:actionvar1}
\delta S[x(t)] = \int_{t_i}^{t_f} \dt\left( \frac{\partial L}{\partial x} - \ddt \frac{\partial L}{\partial \dot{x}} \right) \delta{x} .
\end{equation}

Then from our definition of the propagator:
\begin{align}
\delta K(x_f,t_f;x_i,t_i) &= \mathcal{N} \int \pathmeasure \delta \left( \exp\left(\frac{i}{\hbar} S[x(t)] \right) \right)\nn \\
&= \mathcal{N} \int \pathmeasure \frac{i}{\hbar}\delta S[x(t)] \exp\left(\frac{i}{\hbar} S[x(t)] \right).
\end{align}

This quantity vanishes, since the propagator only depends on the beginning and end of the path and these are fixed by equation \ref{eq:xvar}. So by \ref{eq:actionvar1}:
\begin{equation}
\int \pathmeasure \int_{t_i}^{t_f} \dt\left( \frac{\partial L}{\partial x} - \ddt \frac{\partial L}{\partial \dot{x}} \right) \delta{x}(t)  \exp\left(\frac{i}{\hbar} S[x(t)] \right) = 0.
\end{equation}

As this must hold for all variations $\delta x(t)$, we get equation \ref{eq:ehrenfest}.\\
\\
To get equation \ref{eq:momentum}, we again perform a variation of $x$, but only keep $x_i$ fixed:
\begin{equation}
\delta x(t_i) = 0, \qquad \delta x(t_f) \neq 0.
\end{equation}

Most of our previous argument goes through; however by lemma \ref{lemma:actionvar}:
\begin{equation}
\delta S[x(t)] = p(t_f)\delta x(t_f) + \int_{t_i}^{t_f} \dt\left( \frac{\partial L}{\partial x} - \ddt \frac{\partial L}{\partial \dot{x}} \right) \delta{x} .
\end{equation}

We can ignore the second term, as by equation \ref{eq:ehrenfest} this vanishes in the path integral. Therefore
\begin{align}
\delta \int \pathmeasure \exp\left(\frac{i}{\hbar} S[x(t)] \right) &= \int \pathmeasure \frac{i}{\hbar}p(t_f) \delta x(t_f) \exp\left(\frac{i}{\hbar} S[x(t)] \right) \nn \\
&= \frac{i}{\hbar} p(t_f) \delta x(t_f) \int \pathmeasure \exp\left(\frac{i}{\hbar} S[x(t)] \right)
\end{align}
from which we deduce equation \ref{eq:momentum}.\\
\end{proof}

We see that this is what motivates the definition $p = -i\hbar \frac{\partial}{\partial x}$ in quantum mechanics. Note the path integral and the propagator only satisfy this equation for $x(t_f)$, not for $x(t_i)$.\\
\\
Given that we got interesting results by taking a non-zero variation of $x_f$, we should consider non-zero variations of $t_f$. This is exactly where Schr{\"o}dinger's equation comes from.

First we define the wavefunction in terms of the propagator:

\begin{defn} \textbf{(Wavefunction)}
Given the propagator $K(x_f,t_f;x_i,t_i)$ and an initial wavefunction $\psi(x_i,t_i) = f(x_i)$, define the wavefunction $\psi(x_f,t_f)$ at time $t_f > t_i$ by: \cite{hitoshipath}
\begin{equation}
\psi(x_f,t_f) = \int \dd x_i\, K(x_f,t_f;x_i,t_i) f(x_i)
\end{equation}
\end{defn}

We now prove Schr{\"o}dinger's equation.

\begin{thm} \label{thm:schrod} \textbf{(Schr{\"o}dinger's equation)}\\
The propagator $K(x_f,t_f;x_i,t_i)$ satisfies the following differential equation: \cite{hitoshipath}
\begin{equation} \label{eq:schrodingerkernel}
\left(i\hbar\frac{\partial}{\partial t_f} - H(t_f)\right) K(x_f,t_f;x_i,t_i) = 0
\end{equation}
Furthermore, the wavefunction $\psi(x,t)$ satisfies the same equation:
\begin{equation}
\boxed{ \left(i\hbar\frac{\partial}{\partial t} - H(t)\right) \psi(x,t) = 0}
\end{equation}
\end{thm}

\begin{proof}
We consider a variation 
\begin{equation}
\delta t_i = 0, \qquad \delta t_f \neq 0.
\end{equation}

We need to be careful with our limits, as we still need the same endpoint, i.e. $\delta x(t_f) = 0$. So we have to change $x(t_f)$ to $x(t_f) - \dot{x}(t_f) \delta t_f$ \cite{hitoshiclassical}. Similarly to lemma \ref{lemma:actionvar}, we get variation of the action:
\begin{equation} \label{eq:actionvar2}
\delta S = L(t_f) \delta t_f - \dot{x} \frac{\partial L}{\partial \dot{x}} (t_f) \delta t_f = -H(t_f) \delta t_f
\end{equation}
where we recognized the Hamiltonian $H = \dot{x} \frac{\partial L}{\partial \dot{x}} - L$ and used equation \ref{eq:ehrenfest} to ignore a term that vanishes in the path integral. \\
\\
Using this on our path integral, we find that 
\begin{equation}
\frac{\partial}{\partial t_f}\int \pathmeasure \exp\left(\frac{i}{\hbar} S[x(t)] \right) = - \int \pathmeasure \frac{i}{\hbar}H(t_f) \exp\left(\frac{i}{\hbar} S[x(t)] \right)
\end{equation}
from which equation \ref{eq:schrodingerkernel} follows.\\
\\
Using our definition of the wavefunction, and that $\psi(x_i,t_i) = f(x_i)$ is independent of time, we get the final result. \\
\end{proof}

In the proof above, we found that the following relation holds in the path integral:
\begin{equation}
\frac{\partial S}{\partial t_f} + H(t_f) = 0.
\end{equation}
This is the \textbf{Hamilton-Jacobi} relation familiar from classical mechanics \cite{hitoshiclassical}. Again we see an equation from classical mechanics that is not true exactly in quantum mechanics, but holds inside the path integral.\\ 
\\
Theorem \ref{thm:schrod} shows that Schr{\"o}dinger's equation follows from the path integral definition of quantum mechanics and thus the two formulations are equivalent. In doing so, we have also seen a correspondence between relations in classical and quantum mechanics.

\section{Mathematical considerations: stationary phase and zeta-regularization} \label{sec:mathcons}

In this section we explore two important mathematical aspects of the path integral: the \emph{stationary phase approximation} and \emph{zeta-regularization}. 

\subsection{Stationary phase approximation}
The stationary phase approximation is an approximation to path integrals with $S \gg \hbar$. It further provides a ``derivation" of the principle of least action in classical mechanics. It also has connections with localization in supersymmetry, which we examine in section \ref{sec:localization}.

\subsubsection{Single-variable stationary phase}
We shall be concerned with the behaviour as $\hbar \rightarrow 0$ of the integral:
\begin{equation} \label{eq:Idefn}
I(\hbar) := \int \dx\, \exp\left(\frac{i}{\hbar} f(x)\right)
\end{equation}
for some suitably differentiable real function $f(x)$. Note that the integrand satisfies $|\exp\left(i f(x)/\hbar\right)| = 1$, so it is not Lebesgue integrable over $\mathbb{R}$. However, it exists as an improper Riemann integral\cite{Parissis}:  $\lim_{R\rightarrow\infty} \int_{-R}^R \dx\, \exp\left(if(x)/\hbar\right)$, which is how we consider it.

\begin{propn} \label{propn:statphase1}
Suppose $f(x)$ is a $C^{\infty} (\mathbb{R})$ function with a single non-degenerate stationary point $x = x_0$, i.e. $f'(x_0) = 0, f''(x_0) \neq 0$. Then as $\hbar \rightarrow 0$:
\begin{equation}
I(\hbar) = \int \dx\, \exp\left(\frac{i}{\hbar} f(x)\right)  = \sqrt{\frac{2\pi i h}{|f''(x_0)|}} + \mathcal{O} (\hbar^{3/2}).
\end{equation}
\end{propn}

\begin{proof}
We change variable: $x = x_0 + \sqrt{\hbar} y$ and Taylor expand $f$ in $y$: \cite{Blau}
\begin{align}
f(x_0 + \sqrt{\hbar}y) &= f(x_0) + \sqrt{\hbar} f'(x_0) y + \frac{1}{2!}\hbar f''(x_0)y^2 + \frac{1}{3!}\hbar^{3/2} f^{(3)}(x_0)y^3 +  \dots \nn \\
&=  f(x_0) + \frac{1}{2!}\hbar f''(x_0)y^2 + r(y)
\end{align}
where
\begin{equation}
r(y) = \hbar \sum_{n=3}^{\infty} \frac{\hbar^{n/2}}{n!}f^{(n)}(x_0)y^n.
\end{equation}

Important to note here is that $r(y)/\hbar$ is a power series in strictly positive powers of $\sqrt{\hbar}$ \cite{Blau}. Thus a Taylor expansion of the exponential $\exp(ir(y)/\hbar)$ in the integral yields:
\begin{equation}
 I(\hbar) = \sqrt{\hbar} \int \dy\, \exp\left(i \frac{1}{2} f''(x_0) y^2\right)\left(1 + \mathcal{O}(\hbar^{1/2})\right).
\end{equation}

In fact, all non-integer powers of $\hbar$ disappear as the integral of any odd power $y^n$ in this expansion vanishes (the product of $y^n$ and the exponential is an odd function):
\begin{align}
I(\hbar) &=  \sqrt{\hbar}  \int \dy\, \exp\left(i \frac{1}{2} f''(x_0) y^2\right)\left(1 + \hbar(\dots) + \hbar^2(\dots) + \dots\right)\nn \\
&= \sqrt{\frac{2\pi i h}{|f''(x_0)|}} + \mathcal{O} (\hbar^{3/2})
\end{align}
where we used our familiar Fresnel integral formula.

Hence in the limit $\hbar \rightarrow 0$, the integral satisfies:
\begin{equation}
I(\hbar) \approx \sqrt{\frac{2\pi i h}{|f''(x_0)|}}.
\end{equation}
as claimed.
\end{proof}

We took the limits of integration to be over all of $\mathbb{R}$, but in fact the leading $\mathcal{O}(\sqrt{\hbar})$ contribution is given by integration over any small neighbourhood around $x_0$. To see this, we need Van der Corput's lemma: \cite{Parissis}

\begin{lemma} \label{lemma:corput} \textbf{(Van der Corput's lemma)}
Suppose $f$ is $C^1$ function on $[a,b]$ such that $|f'(x)| \geq \gamma>0$ for all $x \in [a,b]$ and that $f'(x)$ is monotonic. Then
\begin{equation}
\left| \int_a^b \dx\, \exp\left(\frac{i}{\hbar} f(x)\right) \right| \leq C \hbar
\end{equation}
where $C$ is some constant not depending on $\hbar$, $a$ or $b$.
\end{lemma}

\begin{proof}
We follow the proof in \cite{Parissis}. 

Integrate by parts:
\begin{align}
&\int_a^b \dx\, \exp\left(\frac{i}{\hbar} f(x)\right) = \int_a^b \dx\, \frac{\hbar}{i} \frac{1}{f'(x)} \ddx \exp\left(\frac{i}{\hbar} f(x)\right) \nonumber\\
&= -i\hbar \left[ \frac{\exp\left(\frac{i}{\hbar} f(b)\right)}{f'(b)} - \frac{\exp\left(\frac{i}{\hbar} f(a)\right)}{f'(a)} - \int_a^b \dx\, \exp\left(\frac{i}{\hbar} f(x)\right) \ddx\left(\frac{1}{f'(x)}\right) \right].
\end{align}
Now use the triangle inequality:
\begin{align}
&\frac{1}{\hbar} \left| \int_a^b \dx\, \exp\left(\frac{i}{\hbar} f(x)\right) \right| = \nn\\
& \leq  \left| \frac{1}{f'(b)}\right| + \left| \frac{1}{f'(a)}\right| + \left| \int_a^b \dx\, \exp\left(\frac{i}{\hbar} f(x)\right) \ddx\left(\frac{1}{f'(x)}\right)\right| \nn\\
& =  \left| \frac{1}{f'(b)}\right| + \left| \frac{1}{f'(a)}\right| + \int_a^b \dx\,\left|  \ddx\left(\frac{1}{f'(x)}\right)\right|.
\end{align}
As $f'(x)$ is monotonic, $\ddx\left(\frac{1}{f'(x)}\right)$ is of fixed sign and hence we can move the modulus out of the integral again to get the bound:
\begin{align}
\left| \int_a^b \dx\, \exp\left(\frac{i}{\hbar} f(x)\right) \right| &\leq \hbar \left( \left| \frac{1}{f'(b)}\right| + \left| \frac{1}{f'(a)}\right| + \left| \int_a^b \dx\, \ddx\left(\frac{1}{f'(x)}\right)\right| \right) \nn\\
&\leq\hbar \left( \left| \frac{1}{f'(b)}\right| + \left| \frac{1}{f'(a)}\right| + \left|  \frac{1}{f'(b)} -  \frac{1}{f'(a)}\right| \right) \nn\\
& \leq C\hbar
\end{align}
where 
\begin{equation} 
C = \frac{4}{\gamma} \geq 4 \max \left\{\left| \frac{1}{f'(b)}\right|, \left| \frac{1}{f'(a)}\right| \right\}
\end{equation}
which is indeed independent of $a$, $b$ and $\hbar$.
\end{proof}

Now let $U$ be any small interval around the critical point $x_0$ of $f$, and consider the complement $\mathbb{R}\setminus U$, which is the union of two intervals. Then by the above lemma
\begin{equation} \label{eq:corput}
\boxed{\int_{\mathbb{R}\setminus U} \dx\, \exp\left(\frac{i}{\hbar} f(x)\right) = \mathcal{O}(\hbar).}
\end{equation}
Thus it follows that
\begin{align} \label{eq:statphaselocal}
\int_U \dx \exp\left(\frac{i}{\hbar} f(x)\right) &= \int_{\mathbb{R}} \dx \exp\left(\frac{i}{\hbar} f(x)\right) - \int_{\mathbb{R}\setminus U} \dx \exp\left(\frac{i}{\hbar} f(x)\right) \nn\\
&=  \sqrt{\frac{2\pi i h}{|f''(x_0)|}} + \mathcal{O}(\hbar^{3/2}) - \mathcal{O}(\hbar) \nn\\
&=  \sqrt{\frac{2\pi i h}{|f''(x_0)|}} + \mathcal{O}(\hbar)
\end{align}
where we used proposition \ref{propn:statphase1} and equation \ref{eq:corput}. Thus the leading $\mathcal{O}(\sqrt{\hbar})$ contribution comes just from the stationary point.\\
\\
We can now prove the stationary phase approximation in generality.

\begin{thm} \textbf{(Stationary phase)}
Suppose $f(x)$ is a $C^{\infty} (\mathbb{R})$ function with non-degenerate isolated stationary points $x_1,x_2,\dots,x_n$. Then as $\hbar \rightarrow 0$, we have relation:
\begin{equation} \label{eq:stationaryphase}
\int \dx \exp\left(\frac{i}{\hbar} f(x)\right) = \sum_{i=1}^n\sqrt{\frac{2\pi i h}{|f''(x_i)|}} + \mathcal{O} (\hbar)
\end{equation}
\end{thm}
\begin{proof}
Choose small non-overlapping intervals $U_i$ such that $x_i \in U_i$ for all $i$. Then $\mathbb{R}\setminus \bigcup_i U_i$ is a union of intervals and we can apply equation \ref{eq:corput} to see that their contribution is $\mathcal{O}(\hbar)$. Thus:
\begin{align}
\int \dx \exp\left(\frac{i}{\hbar} f(x)\right) &= \sum_{i=1}^n \int_{U_i} \dx \exp\left(\frac{i}{\hbar} f(x)\right) + \int_{\mathbb{R}\setminus \bigcup_i U_i} \dx \exp\left(\frac{i}{\hbar} f(x)\right) \nn\\
&= \sum_{i=1}^n\sqrt{\frac{2\pi i h}{|f''(x_i)|}} + \mathcal{O} (\hbar)
\end{align}
where we used equation \ref{eq:statphaselocal}. 
\end{proof}

If we ignore all terms but the leading $\mathcal{O}(\sqrt{\hbar})$ contribution, we get the \emph{stationary phase approximation}: \cite{Mackenzie}
\begin{equation}
\boxed{ \int \dx \exp\left(\frac{i}{\hbar} f(x)\right) \approx \sum_{\substack{x_i \text{ stationary} \\ \text{ points of }f}}\sqrt{\frac{2\pi i h}{|f''(x_i)|}} }
\end{equation}

\subsubsection{Principle of least action}
We return to the path integral and consider the regime of classical mechanics, i.e. the limit $\hbar \rightarrow 0$, or more precisely: the situation $|S| \gg \hbar$. By analogy with the stationary phase approximation for a function $f: \mathbb{R} \rightarrow \mathbb{R}$, the path integral
\begin{equation}
\lim_{\hbar\rightarrow 0 } \int \pathmeasure \exp\left(\frac{i}{\hbar} S[x(t)]\right) 
\end{equation}
is completely determined by paths extremizing the action, i.e. paths $x_c(t)$ satisfying \cite{Mackenzie}
\begin{equation} \label{eq:leastaction}
\boxed{ \left. \frac{\delta S[x(t)]}{\delta x(t)} \right|_{x(t) = x_c(t)} = 0}
\end{equation}
which is exactly the classical equation of motion! This shows the \emph{correspondence principle}: quantum mechanics reproduces classical mechanics in the appropriate limit: $\hbar \rightarrow 0$.

\subsection{Zeta-regularization} \label{sec:zeta}
In section \ref{sec:freeparticle} we defined the determinant of suitable infinite-dimensional operators as the product of its eigenvalues, generally giving an infinite answer.

Here we shall provide a common, but powerful alternative: the \textbf{zeta-regularized determinant}. The main idea is to use analytic continuation of certain meromorphic functions to make the infinite eigenvalue product finite.

\begin{defn}\textbf{(Zeta-regularized determinant)} \label{defn:zetadet}
As in definition \ref{defn:det}, let $\mathcal{H}$ be a separable Hilbert space and $A: \mathcal{H} \rightarrow \mathcal{H}$ an elliptic, self-adjoint linear operator with a complete set of eigenvectors and associated eigenvalues $\{ \lambda_n \}_{n= 1}^{\infty}$. Define its spectral zeta function $\zeta_A (s)$ by \cite{Barone, Blau}
\begin{equation} \label{eq:spectralzeta}
\zeta_A (s) = \sum_{n=1}^{\infty} \frac{1}{\lambda_n^s}
\end{equation}
for large enough $\mathrm{Re} (s)$ such that this converges. Use analytic continuation to extend this function to the complex plane. By a standard result in functional analysis this function is meromorphic and differentiable at $0$. 
Then define the (zeta-regularized) determinant of $A$ as:
\begin{equation}
\det A = \exp (- \zeta_A'(0) ).
\end{equation}
\end{defn}

To see what motivates this definition, consider differentiating equation \ref{eq:spectralzeta} term-by-term to get the formal expression:
\begin{equation} \label{eq:zetaderiv}
\zeta_A'(s) = \sum_{n=1}^{\infty} -\frac{\log \lambda_n }{\lambda_n^s}
\end{equation}
so that (again formally)
\begin{equation}
\exp ( -\zeta_A'(0) ) = \exp \left( \sum_{n=1}^{\infty} \log \lambda_n \right) = \prod_{n= 1}^{\infty} \lambda_n .
\end{equation}

Of course the difference here is that equation \ref{eq:zetaderiv} only rigorously holds for large enough $\text{Re} (s)$. Therefore the identity $\zeta_A'(0) = - \sum_{n=1}^{\infty} \log \lambda_n$ is merely a formal expression resulting from our interpretation of $\zeta_A'(0)$ as
\begin{equation}
\zeta_A'(0) = \sum_{n=1}^{\infty} \frac{\mathrm{d}}{\mathrm{d}s} \left. \left( \frac{1}{\lambda_n^s} \right) \right|_{s=0}.
\end{equation}
Similar ideas applied to the Riemann zeta-function gives identities such as $\sum_{n=1}^{\infty} n = -1/12$. As absurd as this may seem, these ideas will be useful to us: we are interested in ratios of determinants and these will be the same for the regularized and non-regularized versions.\\
\\
We shall consider the regularized determinants for the free particle and harmonic oscillator and compare the results for the propagator as given in sections \ref{sec:freeparticle} and \ref{sec:harmosc}. Naturally the free particle result will be the same, as we can simply redefine the normalization $\mathcal{N}$. The harmonic oscillator however, will be a non-trivial case and we shall get the same final result.

\begin{ex}\textbf{(Free particle determinant)}
Consider the operator $A = - \partial_{\tau}^2$ on the space of functions $y:[\tau_i,\tau_f] \rightarrow \mathbb{R}$ satisfying $y(\tau_i) = 0 = y(\tau_f)$. Its regularized determinant satisfies \cite{Blau}
\begin{equation}
\det A = 2(\tau_f - \tau_i).
\end{equation}
\end{ex}
\begin{proof}
Define $\Delta\tau = \tau_f - \tau_i$.  Recall from example \ref{ex:harmosc} that the operator $A$ has eigenvalues $\lambda_n = \frac{\pi^2 n^2}{(\Delta\tau)^2}$. We calculate the spectral zeta-function:
\begin{align}
\zeta_A(s) &= \sum_{n=1}^{\infty} \left( \frac{\pi n}{\Delta\tau} \right)^{-2s} \nn\\
&= \left( \frac{\Delta\tau}{\pi}\right)^{2s} \sum_{n=1}^{\infty} n^{-2s} \nn\\
&= \left( \frac{\Delta\tau}{\pi}\right)^{2s} \zeta(2s)
\end{align}
where $\zeta(s) = \sum_{n=1}^{\infty} \frac{1}{n^s}$ is the standard Riemann zeta-function. \\
\\
Thus the derivative satisfies
\begin{align}
\zeta'_A(0) &= \left. 2\left(\frac{\Delta\tau}{\pi}\right)^{2s} \left[ \log \left(\frac{\Delta\tau}{\pi} \right)\zeta(2s) + \zeta'(2s) \right] \right|_{s=0} \nn\\
&= 2 \left[ - \frac{1}{2}\log \left(\frac{\Delta\tau}{\pi} \right) - \frac{1}{2} \log(2\pi) \right] \nn\\
&= -\log\left(2\Delta\tau\right)
\end{align}
where we made use of the well-known identities: $\zeta(0) = -\frac{1}{2}, \, \zeta'(0) = -\frac{1}{2} \log(2\pi)$.\\
\\
Using this in the definition of the regularized determinant, we find:
\begin{equation}
\det A = \exp \left( \log(2\Delta\tau) \right) = 2(\tau_f-\tau_i).
\end{equation}
\end{proof}

The example above describes the general strategy when evaluating spectral zeta-functions: try to write it in terms of well-known functions for which you know relevant values and use them to evaluate the derivative at $0$. We were very fortunate in this example that we were able to express it in a particularly simple form; generally we will not be so lucky.\\
\\
Now let us analyse a more complicated system: the harmonic oscillator.

\begin{ex} \textbf{(Harmonic oscillator determinant)} Consider the operator $A_{\omega} =  -\partial_{\tau}^2 + \omega^2 $  (with $\omega$ a constant) on the space of functions $y:[\tau_i,\tau_f] \rightarrow \mathbb{R}$ satisfying $y(\tau_i) = 0 = y(\tau_f)$. Its regularized determinant satisfies \cite{Barone}
\begin{equation}
\det A_{\omega} = 2\frac{\sinh(\omega(\tau_f - \tau_i))}{\omega} = 2i\frac{\sin(\omega(t_f - t_i))}{\omega} .
\end{equation}
Note this reduces to the free particle determinant in the limit $\omega\rightarrow 0$.
\end{ex}
\begin{proof}
We follow the proof in \cite{Barone}. We will use without proof their expression for the spectral zeta-function as a sum of elementary functions.\\
\\
Note $A_{\omega}$ is positive-definite, so its determinant is well-defined. As before, define $\Delta \tau = \tau_f - \tau_i$.\\
\\
Recall that the eigenvalues of $-\partial_{\tau}^2$ are $\lambda_n = \frac{n^2 \pi^2}{(\tau_f-\tau_i)^2}$. Thus the spectral zeta-function of $\tilde{A}_{\omega}$ is 
\begin{equation}
\zeta_{\tilde{A}_{\omega}}(s) = \left( \frac{\Delta\tau}{\pi}\right)^{2s} \sum_{n=1}^{\infty} (n^2+\nu^2) ^{-s} 
\end{equation}
where we define $\nu = \omega \Delta\tau/\pi$. This is the so-called Epstein zeta-function which gives us the expression
\begin{equation} \label{eq:spectralharmosc1}
\zeta_{\tilde{A}_{\omega}}(s) = -\frac{1}{2} \left(\frac{\Delta\tau}{\pi\nu} \right)^{2s} + \frac{F(s)}{\Gamma(s)}
\end{equation}
where $\Gamma(s)$ is the familiar gamma function and $F(s)$ is a function expressed in terms of gamma functions and a modified Bessel function of the second kind. See \cite{Barone} for details. Crucial here is that $F(s)$ is regular at $s=0$ and that $F(0) = - \pi\nu + \sum_{n=1}^{\infty} \frac{\exp(-2\pi n\nu)}{n}$.

Taking the derivative of equation \ref{eq:spectralharmosc1} yields:
\begin{equation}
\zeta'_{\tilde{A}_{\omega}}(0) = -\log \left( \frac{\Delta\tau}{\pi\nu} \right) + \lim_{s\rightarrow 0} \left( -\frac{\Gamma'(s)}{\Gamma^2(s)} F(s) + \frac{F'(s)}{\Gamma(s)} \right).
\end{equation}
Now use the expression $\Gamma(s) \approx 1/s$ as $s \rightarrow 0$ to find that:
\begin{equation} \label{eq:spectralharmosc2}
\zeta'_{\tilde{A}_{\omega}}(0) = -\log \left( \frac{\Delta\tau}{\pi\nu} \right) + F(0) = -\log \left( \frac{\Delta\tau}{\pi\nu} \right) - \pi\nu + \sum_{n=1}^{\infty} \frac{\exp(-2\pi n\nu)}{n}
\end{equation}
where we eliminated a term by using regularity of $F$ at $0$ and our expression for $\Gamma(s)$.

To evaluate the infinite sum we take a derivative:
\begin{align} 
-\frac{1}{2\pi} \frac{\partial}{\partial\nu} \sum_{n=1}^{\infty} \frac{\exp(-2\pi n\nu)}{n} =  \sum_{n=1}^{\infty} \exp(-2\pi n\nu) = \frac{e^{-2\pi\nu}}{1-e^{-2\pi\nu}}
\end{align} 
Thus, up to an additive constant $C$, we find our sum by integrating:
\begin{align}\label{eq:harmoscinfsum}
\sum_{n=1}^{\infty} \frac{\exp(-2\pi n\nu)}{n} = -2\pi \int \mathrm{d}\nu\, \frac{e^{-2\pi\nu}}{1-e^{-2\pi\nu}} &= C -  \log\left(1-e^{-2\pi\nu}\right) \nn\\
&= C -\log\left(e^{-\pi\nu} (e^{\pi\nu}-e^{-\pi\nu})\right) \nn\\
&= C +\pi\nu - \log(2\sinh(\pi\nu)).
\end{align}
Using now that $\lim_{\nu\rightarrow \infty} \sum_{n=1}^{\infty} \frac{\exp(-2\pi n\nu)}{n} = 0$ and that $\lim_{\nu\rightarrow \infty} \log\left(1-e^{-2\pi\nu}\right) = 0$, we find that $C=0$.\\
\\
Upon substituting equation \ref{eq:harmoscinfsum} into \ref{eq:spectralharmosc2}, we find:
\begin{equation}
- \zeta'_{\tilde{A}_{\omega}}(0) = \log \left( \frac{\Delta\tau}{\pi\nu} \right) + \log(2\sinh(\pi\nu)) = \log \left( \frac{2 \sinh(\omega\Delta\tau)}{\omega} \right).
\end{equation} 
Now use $\Delta\tau = i(t_f-t_i)$, and that $\sinh(ix) = \sin(x)$, to find that
\begin{equation}
- \zeta'_{\tilde{A}_{\omega}}(0) = - \log \left( \frac{2 \sin(\omega(t_f-t_i))}{\omega} \right)
\end{equation} 
as claimed.
\end{proof}
Therefore the ratio
\begin{equation}
\frac{\det (-\partial_{\tau}^2)}{\det (-\partial_{\tau}^2 + \omega^2)} = \frac{\omega(\tau_f-\tau_i)}{\sinh(\omega(\tau_f-\tau_i))} = \frac{\omega(t_f-t_i)}{\sin(\omega(t_f-t_i))}
\end{equation}
is the same as for the non-regularized determinants. When we substitute this into the determinant expression for the harmonic oscillator propagator (equation \ref{eq:harmoscprop}), we see that the physically relevant quantity is indeed unchanged.

\chapter{Mathematical preliminaries for supersymmetry}\label{chap:mathprelims}
In the first chapter, we discussed the path integral in quantum mechanics. In the next chapter, we shall consider supersymmetric quantum mechanics and shall generalise this to take place on arbitrary Riemannian manifolds. To give the necessary background, and to set the notation, we shall discuss Grassmann (fermionic) variables and differential geometry.

\section{Grassmann variables}
For the purposes of describing fermions in supersymmetry, we will use so-called \emph{Grassmann} variables.

\begin{defn}\textbf{(Grassmann variables)}
Grassmann variables are an associative, anticommutative algebra, with the following properties (for Grassmann variables $\psi^a, \psi^b, \psi$ and real variable $X$): \cite{Mirrorsymmetry, ooguri}
\begin{itemize}
	\item Anticommutativity: $\psi^a\psi^b = -\psi^b \psi^a$.
	\item Commutativity with real numbers: $\psi X = X\psi$.
	\item Integration: 
	\begin{equation} \label{eq:grassmannint}
	\int \dpsi = 0, \qquad \int \psi\,\dpsi = 1.
	\end{equation}
\end{itemize}
For multiple Grassmann variables $\psi_1, \dots, \psi_n$ we use the convention:
\begin{equation} \label{eq:multgrassint}
\int \psi_1 \dots \psi_n \dd \psi_1 \dots \dd \psi_n = 1.
\end{equation}
\end{defn}

The correspondence between (anti-)commutation relations above with (anti-)commutation relations of creation operators are why these are sometimes called \emph{fermionic} variables and real variables are called \emph{bosonic} variables.

Note that anti-commutativity implies that $\psi^2 = 0$ for any Grassmann variable. Hence the most general analytic function of a single Grassmann variable is $f(\psi) = a + b\psi$ for $a,b\in\mathbb{R}$.

Finally, we consider Grassmann integrals.

\begin{propn}
Let $f(\psi_1,\dots,\psi_n)$ be an analytic function of Grassmann variables $\psi^1, \dots, \psi^n$ with power series expansion
\begin{equation}
f(\psi_1,\dots,\psi_n) = \sum c_{j_1,\dots,j_n} (\psi^1)^{j_1} \dots (\psi^n)^{j_n}.
\end{equation}

Then
\begin{equation}
\int \dpsi^1 \dots \dpsi^n f(\psi_1,\dots,\psi_n) = c_{1,\dots,1}.
\end{equation}

\end{propn}
\begin{proof}
As $(\psi^k)^{j_k} = 0$ for $j_k>1$, the sum is finite. Thus we can exchange sum and integral:
\begin{align}
\int \dpsi^1 \dots \dpsi^n f(\psi_1,\dots,\psi_n) &= \int \dpsi^1 \dots \dpsi^n \sum c_{j_1,\dots,j_n} (\psi^1)^{j_1} \dots (\psi^n)^{j_n} \nn\\
&= \sum_{j_k \in\{0,1\}} c_{j_1,\dots,j_n} \int \dpsi^1 \dots \dpsi^n   (\psi^1)^{j_1} \dots (\psi^n)^{j_n}.
\end{align}
Now use equation \ref{eq:grassmannint} to see that this final integral vanishes unless $j_k = 1$ for all $k$, when it equals $1$ by equation \ref{eq:multgrassint}. Hence:
\begin{equation}
\int \dpsi^1 \dots \dpsi^n f(\psi_1,\dots,\psi_n) = c_{1,\dots,1}.
\end{equation}
\end{proof}

\section{Differential geometry} \label{sec:diffgeom}
In this section we give an overview of basic notions of differential geometry that will be relevant when discussing SUSY QM on manifolds. We will take a ``physicist's approach", in which we state most results without proof, but will provide examples to explain the ideas.

There are many excellent books discussing the topic, with slightly different approaches. We shall roughly follow \cite{diffgeom}, \cite[Chapter~1]{Mirrorsymmetry} and \cite[Chapter~$5-7$]{Nakahara}.

\subsection{Manifolds}
Let us first define a smooth manifold. One should think of this as a space that looks locally Euclidean. 

\begin{defn} \textbf{(Manifold)}
A topological space $M$ is a smooth n-dimensional (real) manifold if \cite{Nakahara}
\begin{itemize}
\item There is a set of pairs $\{(U_{\alpha}, \phi_{\alpha}) \}$ where $\{ U_{\alpha} \}$ is an open covering of $M$ and each $\phi_{\alpha}$ is a homeomorphism $\phi_{\alpha} : U_{\alpha} \rightarrow V_{\alpha}$ onto an open subset $V_{\alpha}$ of $\mathbb{R}^n$.
\item If for any $\alpha$ and $\beta$:  $ U_{\alpha} \cap U_{\beta} \neq \emptyset$, then the transition function $\phi_{\alpha} \circ \phi^{-1}_{\beta} : V_{\beta} \rightarrow V_{\alpha}$ is smooth, i.e. infinitely differentiable.
\end{itemize}
\end{defn}
We call a pair $(U_{\alpha}, \phi_{\alpha} )$ a \emph{chart} and the collection $\{(U_{\alpha}, \phi_{\alpha} )\}$ an \emph{atlas}. We call $\phi_{\alpha}$ a coordinate function or \emph{coordinates}. The function $\phi_{\alpha}$ is represented as a real $n$-vector: $( x^1, \dots, x^n) \in \mathbb{R}^n$. By a slight abuse of notation we also call these $x^i$ coordinates.\\
\\
While we have defined the manifold by referring to a specific atlas, there are many different possible atlases and we think of the manifold as existing independently of the choice of atlas. As a useful analogy, one might consider vector spaces existing independently of a choice of basis, even though they can be defined in terms of them. \\
\\
In the following, by ``manifold" we shall mean a smooth real manifold, unless explicitly stated otherwise.

\begin{ex}\label{ex:sphere}
The unit $n$-sphere $S^2$ defined by $\left\{ (x_1,\dots,x_{n+1}) \in \mathbb{R}^n : \sum_{i=1}^{n+1} x_i^2 = 1 \right\}$ with its induced topology is an $n$-dimensional manifold.
\end{ex}
\begin{proof}
Stereographic projection from two poles yields two charts that form an atlas.
\end{proof}

Given two manifolds $M$ and $N$, we can define their \emph{product manifold}.
\begin{defn}\textbf{(Product manifold)}
Let $M$ be an m-dimensional manifold with atlas $\{( U_{\alpha}, \phi_{\alpha}) \} $ and $N$ an n-dimensional one with atlas $\{( U'_{\beta}, \phi'_{\beta}) \} $. Define the product manifold $M\times N$ to be the topological space $M\times N$ with the product topology and the atlas $\{ ((U_{\alpha}\times U'_{\beta}), (\phi_{\alpha},\phi'_{\beta}))\}$. 
\end{defn}
\vspace*{1pt}
\begin{ex}
The torus $T^2$ is the product manifold $S^1 \times S^1$.
\end{ex}

We now define a \emph{fibre bundle}. Intuitively, this is a manifold $B$, the base space, over which at each point $x\in B$, there is another manifold $F_x$, called the fibre at $x$. As an analogy, consider a hairbrush, where the handle forms the base space and the bristles form the fibres. 

The important point is that locally the bundle looks like a product manifold $B \times F$.

\begin{defn} \textbf{(Fibre bundle)}
A smooth fibre bundle is a $4$-tuple $(E,B,\pi,F)$ where $E,B,F$ are smooth manifolds and $\pi: E \rightarrow B$ is a continuous surjection such that for any point $x \in B$, there is a neighbourhood $U \subseteq B$ and a homeomorphism $\phi: U\times F  \rightarrow \pi^{-1}(U)$ satisfying:
\begin{equation}
(\pi \circ \phi)(x,f) = x 
\end{equation}
for all $x \in U$ and $f \in F$.

A \textbf{section} $f$ of a fibre bundle is a continuous map $f: B \rightarrow E$ satisfying $\pi(f(x)) = x$ for all $x \in B$ \cite{Mirrorsymmetry}. This locally looks like $f:x \mapsto (x,g(x)) $ for some function $g: U \rightarrow F$, thus generalising the notion of a graph.
\end{defn}

We are interested in \textbf{vector bundles}, where the manifold $F$ is a real $n$-dimensional vector space and the map $ v \mapsto \phi(x,v) $ is an isomorphism between $F$ and $\mathbb{R}^n$.\\
\\
We should think of a vector bundle as follows: at every point $x \in B$ there is a vector space $F_x$, which are isomorphic to each other, but not the same. Hence we cannot, for example, add vectors in different fibres. In our hairbrush: all the bristles are equivalent (homeomorphic), but not equal.  

An important example defined below is the \emph{tangent bundle}. Intuitively, the tangent space is given by derivatives of curves, and the tangent bundle is the collection of all tangent spaces.

\begin{defn} \textbf{(Tangent and cotangent bundle)}
Consider an $n$-dimensional manifold $B$, a point $x \in B$ and local coordinates $\{x^{\mu}\}$. Define an equivalence relation $\sim$ on the set of curves $\{ \gamma_i: (-\epsilon,\epsilon)\rightarrow B \mid \gamma_i(0) =x\}$ by $\gamma_i \sim \gamma_j$ if 
\begin{equation}
\left.  \frac{\mathrm{d} x^{\mu}(\gamma_i (t))}{\mathrm{d}t} \right|_{t=0} = \left. \frac{\mathrm{d} x^{\mu}(\gamma_j (t))}{\mathrm{d}t} \right|_{t=0} .
\end{equation}
We identify a \textbf{tangent vector} $X$ with an equivalence class of such curves. In coordinates we can express the vector $X$ as $X = X^{\mu}\partial_{\mu}$, where $\partial_{\mu} = \frac{\partial}{\partial x^{\mu}}$ and $X^{\mu} = \left.  \frac{\mathrm{d} x^{\mu}(\gamma_i (t))}{\mathrm{d}t} \right|_{t=0}$ and where we used the summation convention \cite{Nakahara}.

The space of all tangent vectors $X$ at $x$ forms the \textbf{tangent space} at x, denoted $T_x B$, and the collection of all tangent spaces at different points on the manifold $B$ forms the \textbf{tangent bundle}:
\begin{equation}
TB = \bigcup_{x\in B} T_x B.
\end{equation}

As a finite-dimensional vector space, $T_x B$ has a dual space $T^*_x B$ called the \textbf{cotangent space} of linear maps $f: T_x B \rightarrow \mathbb{R}$. We call elements in $T_x^* B$ \textbf{1-forms}. The collection of cotangent spaces forms the \textbf{cotangent bundle}:
\begin{equation}
T^* B = \bigcup_{x\in B} T^*_x B.
\end{equation}
\end{defn}

From the definition of the tangent space $T_x B$, and given local coordinates $\{ x^{\mu}\}$, we note that the vectors $\{ \partial_{\mu} \}$ form a basis for $T_x B$. Then $T^*_x B$ has an associated dual basis $\{ \dx^{\mu} \}$, satisfying $\left< \dx^{\mu}, \partial_{\nu} \right> = \delta^{\mu}_{\nu}$ where we define the inner product as $\left< \phi,V \right> = \phi(V) = V(\phi)$ for $V \in T_x B$ and $\phi \in T_x^* B$.\\
\\
Given tangent and cotangent spaces, we can uniquely define tensor spaces of tensors $Q$ of type $(p,q)$, which are multilinear maps 
\begin{equation}
Q: \underbrace{T^*_x B \times \dots \times T^*_x B}_{p\,\text{copies}} \times \underbrace{T_x B \times \dots \times T_x B}_{q\,\text{copies}}  \rightarrow \mathbb{R}.
\end{equation}

\subsection{Riemannian manifolds}
The manifolds that we will discuss in supersymmetry are Riemannian manifolds, in which each tangent space has an inner product.

\begin{defn}\textbf{(Riemannian manifold)}
A Riemannian manifold is a pair $(B,g)$, where $B$ is a manifold and $g = g(x)$ is a smooth function $g(x): T_x B \times T_x B \rightarrow \mathbb{R}$ defining a (positive-definite) inner product \cite{Mirrorsymmetry}.
\end{defn}

We think of $g$ as a smooth $(0,2)$-type tensor field $g_{\mu\nu} = g_{\mu\nu}(x)$. It is invertible with inverse $g^{\mu\nu}$ satisfying $g^{\mu\rho}g_{\rho\nu} = \delta^{\mu}_{\nu}$.

The metric defines lengths of curves on the manifold: let $x(t)$ be a curve in $B$ and $\{ x^{\mu} \}$ be local coordinates. Then the length of the curve is \cite{diffgeom}
\begin{equation}
L[x(t)] = \int \dt \sqrt{g_{\mu\nu} \frac{\dd x^{\mu}}{\dt} \frac{\dd x^{\nu}}{\dt}}.
\end{equation}
As an inner product, it also gives a notion of angles between curves.\\
\\
The metric itself is unable to relate nearby fibres on a manifold. This is where the idea of a connection comes into play; however we can only define ``constancy" on curves, not globally.

\begin{defn}\textbf{(Connection)}
Let $E$ be a fibre bundle with base space $B$. Let $\Gamma(E)$ be the set of smooth sections of $E$. A connection $\nabla$ is a linear map \cite{Mirrorsymmetry}
\begin{equation}
\nabla: \Gamma(E) \rightarrow \Omega^1 \otimes \Gamma(E)
\end{equation}
(where $\Omega^1$ is the set of all sections of 1-forms) satisfying the Leibniz rule:
\begin{equation}
\nabla(\sigma\otimes f) = \nabla \sigma \otimes f + \sigma \otimes \dd f
\end{equation}
for any smooth section $\sigma$ and smooth function $f$.
\end{defn}

In local coordinates $\{ x^{\mu} \}$, the connection acts as 
\begin{equation}
\nabla_{\mu} f = \partial_{\mu} f
\end{equation}
for any function $f$. Further
\begin{equation}
\nabla_{\mu} V^{\nu} = \partial_{\mu} V^{\nu} + \Gamma^{\nu}_{\,\mu\lambda} V^{\lambda}
\end{equation}
for any vector field $V^{\nu}$, where we call the $\Gamma^{\nu}_{\,\mu\lambda}$ connection coefficients or Christoffel symbols. The Leibniz rule extends this to arbitrary tensors, so the connection is completely specified by $\Gamma^{\nu}_{\,\mu\lambda}$.

A Riemannian manifold has a special connection: the Levi-Civita connection, which is the one we shall be concerned with.

\begin{thm}\textbf{(Levi-Civita connection)}
Any Riemannian manifold $(B,g)$ admits a unique metric-compatible connection $(\nabla g = 0)$ that is torsion free $(\Gamma^{\nu}_{\,\mu\lambda} = \Gamma^{\nu}_{\,\lambda\mu})$, called the Levi-Civita connection. In coordinates:
\begin{equation}
\Gamma^{\rho}_{\mu\nu} = \frac{1}{2} g^{\rho\sigma} \left( \partial_{\mu} g_{\nu\sigma} + \partial_{\nu} g_{\mu\sigma} - \partial_{\sigma} g_{\mu\nu} \right).
\end{equation}
\end{thm}

We are now ready for the final ingredient in our discussion of Riemannian manifolds: curvature. 

\begin{defn}\textbf{(Curvature)}
Define the curvature tensor $R^{\rho}_{\,\sigma\mu\nu}$ as a failure of the connection to commute:
\begin{equation}
\left( \nabla_{\mu}\nabla_{\nu} - \nabla_{\nu}\nabla_{\mu} \right) V^{\rho} =  R^{\rho}_{\,\sigma\mu\nu} V^{\sigma}
\end{equation}
for any vector $V^{\sigma}$. Then in coordinates:
\begin{equation} \label{eq:curvcoords}
R^{\rho}_{\,\sigma\mu\nu} = \partial_{\mu} \Gamma^{\rho}_{\,\nu\sigma} - \partial_{\nu} \Gamma^{\rho}_{\,\mu\sigma} + \Gamma^{\rho}_{\,\mu\lambda}\Gamma^{\lambda}_{\,\nu\sigma} - \Gamma^{\rho}_{\,\nu\lambda}\Gamma^{\lambda}_{\,\mu\sigma}.
\end{equation}
\end{defn}

This measures locally how much the space is not ``flat", or how it locally differs from Euclidean space. Note curvature is an intrinsic property of the manifold and independent of any embedding.

\begin{propn}
The Riemann tensor $R_{\mu\nu\sigma\rho}$ has the following symmetries:
\begin{align}
R_{\mu\nu\sigma\rho} &= - R_{\nu\mu\sigma\rho} \nn\\
R_{\mu\nu\sigma\rho} &= - R_{\mu\nu\rho\sigma} \nn\\
R_{\mu\nu\sigma\rho} + R_{\mu\sigma\rho\nu} &+ R_{\mu\rho\nu\sigma} = 0
\end{align}
\end{propn}
\begin{proof}
This follows from the coordinate expression \ref{eq:curvcoords}. 
\end{proof}

\begin{propn}
For any point $x_0\in B$ there exist coordinates $\{ x^{\mu} \}$ around $x_0$ such that
\begin{equation}
\partial_{\lambda} g_{\mu\nu} (x_0) = 0, \qquad g_{\mu\nu}(x_0) = \delta_{\mu\nu} (x_0)
\end{equation} 
in these coordinates. Then
\begin{equation}
\Gamma^{\rho}_{\,\mu\nu} (x_0) = 0 \nn\\
\end{equation}
and further
\begin{equation}
R^{\rho}_{\,\sigma\mu\nu} (x_0) = \partial_{\mu} \Gamma^{\rho}_{\,\nu\sigma} (x_0) - \partial_{\nu} \Gamma^{\rho}_{\,\mu\sigma} (x_0).
\end{equation}
Call these \textbf{Riemann normal coordinates}.
\end{propn}

\subsection{Differential forms} \label{sec:diffforms}
We can use tensors to define differential forms, which shall be of major importance when discussing supersymmetry on manifolds.
\begin{defn} \textbf{(Differential form)} A differential form of order $r$, (or an $r$\emph{-form}), is a totally antisymmetric tensor of type $(0,r)$ \cite{Nakahara}.

The space of all $r$-forms at $x\in B$ is denoted by $\Lambda^r_x B$. The \textbf{exterior algebra} $\Lambda^*_x B$ is the direct sum of these:
\begin{equation} \label{eq:extalgebra}
\Lambda^*_x B = \bigoplus_{r\in\mathbb{Z}_{\geq 0}} \Lambda^r_x B.
\end{equation}

From this we can form the \textbf{exterior bundle} $\Lambda^*B$:
\begin{equation}
\Lambda^* B = \bigcup_{x\in B} \Lambda^*_x B.
\end{equation}

Define $\Omega^r(B)$ to be the space of smooth sections of $\Lambda^r B$, where we identify $\Omega^0(B)$ as the space of smooth functions on $B$.
\end{defn}

We define the \textbf{wedge product} (or \emph{exterior product}) on the exterior algebra as follows: for a $q$-form $\omega$ and an $r$-form $\xi$, the wedge product $\omega \wedge \xi$ is a totally antisymmetric $(q+r)$-form given by: \cite{Nakahara}
\begin{equation}
(\omega \wedge \xi)(V_1, \dots, V_{q+r}) = \frac{1}{q!r!}\sum_{\sigma \in \text{Sym}(q+r)} \text{sgn}(\sigma) \omega(V_{\sigma(1)},\dots,V_{\sigma(q)})\xi(V_{\sigma(q+1)},\dots V_{\sigma(q+r)})
\end{equation}
where the $V_i$ are vectors, $\mathrm{Sym}(q+r)$ denotes the permutation group and \\$\mathrm{sgn}: \mathrm{Sym}(q+r)\rightarrow \{+1,-1\}$ the sign-function on permutations.

We shall need the following proposition, which we do not prove here.
\begin{propn}
Let $V$ be an $n$-dimensional vector space with basis $\{v_i \}_{1\leq i \leq n}$. Then the set $\{ v_{\mu_1}\wedge\dots\wedge v_{\mu_r}\}_{\mu_1 < \mu_2 < \dots < \mu_r}$ is a basis for $\Lambda^r V$. 

Hence
\begin{equation}
\text{dim } \Lambda^r V = \binom{n}{r}.
\end{equation} 
Specifically, $\text{dim } \Lambda^n V = 1$ and $\text{dim } \Lambda^r V = 0$ if $r>n$. 
\end{propn}

From this proposition, we gather that $\Lambda^*_x B$ is a graded algebra, the grading being provided by the order. Also the direct sum in \ref{eq:extalgebra} is finite:
\begin{equation}
\Lambda^*_x B = \bigoplus_{r=0}^n \Lambda^r_x B.
\end{equation}

Furthermore, as $\dim \Lambda^r V = \dim \Lambda^{n-r}V$, these spaces are isomorphic. We shall see later that for Riemannian manifolds, there is a canonical isomorphism given by the \emph{Hodge star}.

Let us look at an example of differential forms on a vector space $V$.
\begin{ex} \label{ex:R3forms}
Let $V = \mathbb{R}^3$ and let $\omega_r \in \Omega^r(\mathbb{R}^3)$. Then they are of the following form
\begin{enumerate}
\item $\omega_0 = f(x,y,z)$,
\item $\omega_1 = \omega_x\dx + \omega_y \dy + \omega_z \mathrm{d}z$,
\item $\omega_2 = \omega_{xy}\dx\wedge\dy + \omega_{yz} \dd y \wedge\dd z +\omega_{zx} \dd z \wedge\dd x $,
\item $\omega_3 = \omega_{xyz}\dx\wedge\dy\wedge\mathrm{d}z$,
\end{enumerate}
where $f, \omega_x,\omega_y,\omega_z,\omega_{xy},\dots$ are all smooth functions on $\mathbb{R}^3$.

Later we shall see that we can identify 0-forms and 3-forms with functions and 1-forms and 2-forms with vectors.
\end{ex}

Given a map $F: B \rightarrow B$, there is a natural induced map on differential forms.

\begin{defn} \label{defn:pullback}
Given a map $F: B \rightarrow B$, define the \textbf{pullback} $F^*: \Omega^p(B) \rightarrow \Omega^p(B)$ by \cite{diffgeom}
\begin{align}
F^*(f) &= f \circ F \nn\\
F^* (\dd f) &= \dd(f\circ F)
\end{align}
for $f \in \Omega^0(B)$ and extend to $p$-forms via:
\begin{equation}
F^* (\alpha\wedge\beta) =F^* (\alpha) \wedge\beta + \alpha\wedge F^*(\beta).
\end{equation}
\end{defn}
Then in coordinates
\begin{align}
F^* &\left(\sum_{i_1,\dots,i_p} a_{i_1, \dots, i_p}(x) \dd x^{i_1}\wedge\dots\wedge \dd x^{i_p}\right) \nn\\
&= \sum_{i_1,\dots,i_p} a_{i_1, \dots, i_p} (F(x))  \dd (x^{i_1} \circ F )\wedge\dots\wedge \dd(x^{i_p} \circ F ) \nn\\
&= \sum_{i_1,\dots,i_p} a_{i_1, \dots, i_p} (F(x))  \dd (F^{i_1} (x) )\wedge\dots\wedge \dd(F^{i_p} (x) ).
\end{align}

\begin{defn}
Let $(B,g)$ be a Riemannian manifold. Define an inner product on the spaces $\Lambda^r T_x B$ on decomposable $r$-forms by
\begin{equation}
\left< v_1 \wedge v_2 \wedge \dots \wedge v_r , w_1 \wedge w_2 \wedge \dots \wedge w_r \right> = \det \left( \left< v_i , w_k \right> \right)
\end{equation}
where the inner product on the RHS is defined by the metric $g^{\mu\nu}$. Extending this linearly to all $r$-forms gives the full inner product.
\end{defn}

Now we can finally discuss the object of central importance to us: the \textbf{exterior derivative}, which shall be identified with an important operator in supersymmetry.
\begin{defn}
The \textbf{exterior derivative} is a map $\dd_r : \Omega^r(B) \rightarrow \Omega^{r+1}(B)$ defined such that on an $r$-form 
\begin{equation}
\omega = \frac{1}{r!}\omega_{\mu_1\dots\mu_r} \dx^{\mu_1} \wedge\dots\wedge\dx^{\mu_r}
\end{equation}
it acts as \cite{Nakahara}
\begin{equation} \label{eq:ddefn}
\dd_r\omega = \frac{1}{r!}\left(\partial_\nu \omega_{\mu_1\dots\mu_r}\right) \dx^{\nu}\wedge \dx^{\mu_1} \wedge\dots\wedge\dx^{\mu_r}
\end{equation}
Now define $\dd : \Omega(B) \rightarrow \Omega(B)$ by $\left. \dd \right|_{\Omega^r(B)} = \dd_r$.

We call an $r$-form $\alpha$ \textbf{closed} if $\dd \alpha = 0$ and \textbf{exact} if there exists an $(r-1)$-form $\beta$ satisfying $\alpha = \dd \beta$.
\end{defn}

We state here some of its important properties.

\begin{propn} \label{propn:exteriorder}
The exterior derivative is the unique linear map $\dd: \Omega(B) \rightarrow \Omega(B)$ such that: \cite{diffgeom}
\begin{enumerate}
	\item If $f \in \Omega^0(B)$ then $\dd f$ is the derivative of f.
	\item It squares to zero:
		\begin{equation} \label{eq:dsquared}
		\boxed{ \dd^2 = 0.}
		\end{equation}
	\item If $\alpha \in \Omega^p(B)$ then
		\begin{equation} 
		\dd(\alpha\wedge\beta) = \dd\alpha\wedge\beta + (-1)^p \alpha\wedge\dd\beta.
		\end{equation}
\end{enumerate}
\end{propn}
\begin{proof}
Chase calculations using the definition. Note commutativity of partial derivatives is crucial in proving equation \ref{eq:dsquared}.

We do not consider uniqueness here.
\end{proof}

Let's return to $\mathbb{R}^3$ and see what the operator $\dd$ looks like.

\begin{ex} \label{ex:graddivcurl}
Consider the differential forms $\omega_r \in \Omega^r(\mathbb{R}^3)$, expanded as in example \ref{ex:R3forms}.
\begin{enumerate}
\item For a function $f$:
	\begin{equation}
	\dd f = \partial_x f \dd x + \partial_y f \dd y + \partial_z f \dd z
	\end{equation}
	which we identify with grad $f$.
\item For a 1-form $\omega_1$:
	\begin{align}
	\dd \omega_1 &= \dd\left( \omega_x\dx + \omega_y \dy + \omega_z \dd z \right) \nn\\
	&= \dd\omega_x\wedge \dx + \dd \omega_y \wedge \dy + \dd \omega_z \wedge \dd z 
	\end{align}
	where we used proposition \ref{propn:exteriorder}. Thus:
	\begin{align}
		\dd \omega_1 &=\left(\partial_x \omega_x \dd x + \partial_y \omega_x \dd y + \partial_z \omega_x \dd z\right)\wedge \dd x + \dots\nn\\
		&= \left(\partial_z \omega_x - \partial_x \omega_z\right)\dd z \wedge \dd x + \left(\partial_y \omega_x - \partial_x \omega_y	\right)\dd x \wedge \dd y \nn\\
		&\quad+ \left(\partial_y \omega_z - \partial_z \omega_y\right)\dd y \wedge \dd z
	\end{align}
	which we identify with the curl of a vector.
\item Similarly for a 2-form $\omega_2$:
	\begin{equation}
	\dd \omega_2 = \left( \partial_x \omega_{yz} + \partial_y \omega_{zx} + \partial_z \omega_{xy} \right) \dd x\wedge \dd y \wedge \dd z
	\end{equation}
which we identify with the divergence of a vector.
\end{enumerate}
\end{ex}
Thus the familiar relations curl grad $= 0$ and div curl $=0$ are both consequences of $\dd^2 = 0$.

\subsubsection{De Rham cohomology}
We saw that $\dd$ satisfies the property $\dd^2 = 0$, so that all exact forms are closed. A natural question to consider, is if there are any closed forms that are not exact and if so, what the space parameterising them looks like. This is precisely what \emph{De Rham cohomology} captures.

\begin{defn}\textbf{(De Rham cohomology)}
Let $B$ be an $n$-dimensional manifold. The $p$-th De Rham cohomology group is defined as the quotient vector space \cite{diffgeom, Mirrorsymmetry}
\begin{equation}
H^{p}_{DR}(B) = \frac{\mathrm{Ker}\, d_p: \Omega^p(B) \rightarrow \Omega^{p+1}(B)}{\mathrm{Im}\, d_{p-1}: \Omega^{p-1}(B) \rightarrow \Omega^p(B)}
\end{equation} 
with operation
\begin{equation}
[ \omega_1] \wedge [\omega_2] = [ \omega_1 \wedge \omega_2]
\end{equation}
where $[\alpha]$ denotes the equivalence class of $p$-form $\alpha$.

We define the $p$-th \textbf{Betti number} of $B$ to be the dimension of the $p$-th De Rham cohomology group of $B$:
\begin{equation}
b_p(B) = \mathrm{dim}\, H^{p}_{DR}(B).
\end{equation}
Define the \textbf{Euler number} $\chi(B)$ of $B$ as the alternating sum of Betti numbers:
\begin{equation}
\chi(B) = \sum_{p=0}^{\infty} (-1)^p b_p(B) = b_0(B) - b_1(b) + \dots + (-1)^n b_n(B)
\end{equation}
which is a finite sum as $\Omega^r(B) = 0$ for all $r > n$.
\end{defn}
For $2$-manifolds (surfaces), this definition Euler number is equivalent to that defined through subdivisions (tilings) of a surface as $\chi(B) = V - E + F$, where $V$ is the number of vertices of the tiling, $E$ the number of edges and $F$ the number of faces.\\
\\
We begin with a simple result:
\begin{propn}
The $0$-th De Rham cohomology group of a manifold $B$ is equal to $\mathbb{R}^k$, with $k$ the number of connected components of $B$. \cite{diffgeom}
\end{propn}
\begin{proof}
From the definition, $f \in H^{0}_{DR}(B)$ iff $\dd f = 0$, iff $f$ is constant on each connected component of $B$. Each constant is a real number, giving parameter space $\mathbb{R}$. The result follows as $f$ is allowed to vary between connected components.
\end{proof}

We now give, without proof, the De Rham cohomology groups of some simple manifolds: \cite{diffgeom}
\begin{itemize}
\item	For $\mathbb{R}^n$ we have the \emph{Poincar{\'e} lemma}: $H^{p}_{DR}(\mathbb{R}^n) = 0$ if $p>0$. 

From the appropriate identifications with grad, div and curl, this shows the well-known results in $\mathbb{R}^3$ that 
	\begin{itemize}
		\item if $\nabla \wedge \phi = 0$, then $\phi = \nabla f$ for some function $f$.
		\item if $\nabla \cdot \phi = 0$, then $\phi = \nabla \wedge \psi$ for some vector field $\psi$.
	\end{itemize}
\item For the circle $S^1$: $H^0_{DR}(S^1) = H^1_{DR}(S^1) = \mathbb{R}$.
\item For the 2-torus $T^2$: $H^0_{DR}(T^2) = \mathbb{R}$, $H^1_{DR}(T^2) = \mathbb{R}^2$ and $H^2_{DR}(T^2) = \mathbb{R}$.

We check that $\chi(T) = 1 - 2 + 1 = 0$, as expected.

\item For the $n$-sphere $S^n$: $H^p_{DR}(S^n) = \mathbb{R}$ if $p=0$ or $p=n$ and is $0$ otherwise.

Again we check that $\chi(S^2) = 1 - 0 + 1 = 2$.
\end{itemize}

In fact, De Rham cohomology is \emph{homotopy invariant} \cite{diffgeom}, which means that, for example, the Poincar{\'e} lemma can be extended to any \emph{contractible} manifold, i.e. one homotopic to a point.

Specifically for the $n$-dimensional disk $D^n$, we have that $H^0_{DR}(D^n ) = \mathbb{R}$ and $H^p_{DR}(D^n) = 0$ for $p>0$. We shall use this in section \ref{sec:geomthms} to prove \emph{Brouwer's fixed-point theorem}.\\
\\
Finally, note that the pullback $F^*$ of a map $F:M \rightarrow M$ induces a map on the De Rham cohomologies via
\begin{equation}
F^*\left[\alpha\right] = \left[ F^* \alpha \right]
\end{equation}
with $\alpha \in \Omega^p(M)$ and $[\,.\,]$ denoting its equivalence class.

\subsubsection{Orientability, integration and Hodge dual}
We can use differential forms to define orientations.
\begin{defn}\textbf{(Orientability)}
Let $B$ be an $n$-dimensional manifold. Then $B$ is orientable if there exists an everywhere non-zero form $\omega\in\Omega^n(B)$ (called a \textbf{volume form}) \cite{diffgeom}.
\end{defn}

We call two orientations equivalent if they are related by a strictly positive function: $\tilde{\omega}(x) = h(x) \omega(x)$, where $h(x) > 0$ everywhere. Thus a connected orientable manifold only has two inequivalent orientations corresponding to the two possible signs.

\begin{defn}
On Riemannian manifolds, there is a particular volume form $\dd V$ of interest:
\begin{equation}
\dd V = \sqrt{\det(g)} \dd x^1 \wedge \dots \wedge \dd x^n
\end{equation}
which is independent of the chosen coordinate system.

Then define the \textbf{Hodge dual} $*$ by \cite{Mirrorsymmetry}
\begin{equation}
\theta \wedge *\psi = \left< \theta, \psi \right> \dd V
\end{equation}
for any $r$-forms $\theta, \psi$.
The Hodge dual is invertible and defines a canonical isomorphism between $\Omega^r(B)$ and $\Omega^{n-r}(B)$. 
\end{defn}

\begin{ex}
It is Hodge duality that allows us in $\mathbb{R}^3$ to identify $0$-forms with $3$-forms and $1$-forms with $2$-forms, and to identify grad, curl and div in example \ref{ex:graddivcurl}. 

It also gives a way to define the cross product: let $v$ and $w$ be two 1-forms representing vectors in $\mathbb{R}^3$. Then their cross product is
\begin{equation}
v \times w = *(v\wedge w) \in \Omega^1(\mathbb{R}^3).
\end{equation}
\end{ex} 

A crucial feature of differential forms is that we can integrate over them. Note that under coordinate transformations, their antisymmetry gives exactly the desired Jacobian determinant. We give the technical definition as in \cite{Nakahara}.
\begin{defn} \textbf{(Integration)}
Let $B$ be a compact, orientable $n$-manifold and $\omega$ a volume form. Let $\{U_i\}$ be an open covering of $B$ such that every point $x\in B$ is only in finitely many $U_i$. Further, let $\{\epsilon_i(x)\}$ be a partition of unity subordinate to $\{U_i\}$, i.e. a family of functions $\epsilon_i: B \rightarrow \mathbb{R}$ satisfying
\begin{itemize}
\item $0\leq\epsilon_i(x) \leq 1$ for all $x\in B$;
\item $\epsilon_i(x) = 0$ if $x \notin U_i$;
\item $\sum_i \epsilon_i(x) = 1$. (This is well defined as $x$ is only in finitely many $U_i$.)
\end{itemize}
Define $\omega_i(x) = \epsilon_i(x) \omega(x) $, so that $\omega(x) = \sum_i \omega_i(x)$ by the last property.

Let $x^i$ be coordinates with coordinate function $\phi$.
Then we define the integral of a volume form $\omega  = f \dd x^1 \wedge \dots \wedge \dd x^n$ on $U_i$ by
\begin{equation}
\int_{U_i} \omega = \int_{\phi(U_i)}  \dd x^1 \dd x^2 \dots \dd x^n\,f(\phi^{-1}(x))
\end{equation}
where the right hand side is just a repeated real integral. This turns out to be independent of the choice of coordinates.

Then define the integral of $\omega$ over $B$ as
\begin{equation}
\int_B \omega = \sum_i \int_{U_i} \omega_i.
\end{equation}
\end{defn}

Use the integral to define a global inner product on $\Omega^r (B)$ by
\begin{equation}
(\theta,\psi) = \int \theta \wedge *\psi
\end{equation}
for $\theta,\psi \in \Omega^r (B)$.

Since $\dd$ is an operator on differential forms, it has an adjoint $\dddagger:\Omega(B) \rightarrow \Omega(B) $ defined by $(\theta,\dd\phi) = (\dddagger\theta,\phi)$ for $r$-form $\theta$ and $(r-1)$-form $\phi$.

\begin{defn}
Define the \textbf{Laplacian} operator on differential forms by 
\begin{equation}
\Delta = \dd\dddagger + \dddagger\dd.
\end{equation}
We call a differential form $h$ \textbf{harmonic} if $\Delta h = 0$ and denote the space of all harmonic $r$-forms by $\mathcal{H}^r(B)$.
\end{defn}

\begin{propn} \label{propn:harm}
Let $B$ be a compact manifold without boundary. An $r$-form $h$ on $B$ is harmonic iff it is closed $(\dd h = 0)$ and co-closed $(\dddagger h = 0)$.
\end{propn}
\begin{proof}
$\Delta h = 0 $ iff $(\phi, \Delta h) = 0$ for any $r$-form $\phi$. Pick $\phi = h$ to get
\begin{equation}
(h,(\dd\dddagger + \dddagger\dd) h ) = (\dd h, \dd h) + (\dddagger h, \dddagger h)
\end{equation}
and this is $0$ iff $\dd h =0$ and $\dddagger h = 0$.
\end{proof}

We shall now state a crucial relation.
\begin{thm} \textbf{(Hodge decomposition)}
Let $\theta$ be a differential form. Then it has a unique decomposition \cite{Nakahara}
\begin{equation}
\theta = h + \dd \alpha + \dddagger \beta
\end{equation}
where $h$ is harmonic.

Thus 
\begin{equation} \label{eq:harmderham} 
\boxed{ \mathcal{H}^r(B) \cong H^r_{DR}(B) } .
\end{equation}
\end{thm}
\begin{proof}
We do not prove existence here. Uniqueness follows from applying $\dd$ and $\dddagger$ to the above equation and using proposition \ref{propn:harm}. 

For equation \ref{eq:harmderham}, we use that $\text{ker } \dddagger\dd = \text{ker } \dd$ and $\text{ker } \dd\dddagger = \text{ker } \dddagger$ (which follows from positive-definiteness of the inner product $(\, ,\,)$), to see that $\dd\dddagger \beta \neq 0$ if $\dddagger \beta \neq 0$ and $\dddagger\dd \alpha \neq 0$ if $\dd \alpha \neq 0$. Then $\text{ker }\dd$ corresponds to forms of the form $\theta = h + \dd \alpha$, giving the result.
\end{proof}

\chapter{Supersymmetric quantum mechanics}\label{chap:SUSYQM}
In this chapter we finally introduce supersymmetry. We first give an introduction in flat space ($\mathbb{R}^n$), and then introduce supersymmetry on Riemannian manifolds. Finally, we use results from the previous chapter to prove the \emph{Gauss-Bonnet-Chern} and \emph{Lefschetz fixed-point} theorems.

Citations from \cite{Mirrorsymmetry} refer to chapters $9.1-9.3$ and $10.1-10.4$. Citations from \cite{Nakahara} refer to chapter $12.9$.

\section{Introduction to supersymmetric quantum mechanics} \label{sec:susyexamples}
Informally speaking, a supersymmetric quantum mechanical model is one in which the action $S$ depends on both bosonic and fermionic (Grassmann) variables, with a symmetry relating the two that leaves $S$ invariant.

We first analyse two examples of SUSY QM in flat space that elucidate its most important features.

\subsubsection{Example 1: real variables}
Let's start with a simple model, with real bosonic variables $x_k$, fermionic variables $\psi_k$ and Lagrangian (using summation convention):
\begin{equation}
L = \frac{1}{2}\dot{x_j}^2 + \frac{1}{2}i\psi_j\dot{\psi_j}
\end{equation}
where $(\,\dot{}\,)$ denotes differentiation with respect to time $t$.

\begin{propn}
This system is invariant under the following transformation: \cite{Nakahara}
\begin{equation} \label{eq:supsymvar1}
\delta x_j = i\epsilon \psi_j, \qquad \delta \psi_j = -\epsilon\dot{x}_j
\end{equation}
where $\epsilon$ is a real infinitesimal Grassmann constant. This is called a \textbf{supersymmetry transformation} as it relates bosons and fermions.
\end{propn}
\begin{proof}
To check invariance, we calculate
\begin{align}
\delta L &= \dot{x}_j \delta \dot{x}_j + \frac{i}{2} \left( \delta \psi_j \dot{\psi}_j + \psi_j\delta\dot{\psi}_j \right) \nn\\
&= \dot{x}_j \ddt \delta x_j + \frac{i}{2} \left( \delta \psi_j \dot{\psi}_j + \psi_j\ddt\delta\psi_j \right) 
\end{align}
where we used commutativity of $\ddt$ and $\delta$.  Thus
\begin{align}\label{eq:Lvar1}
\delta L &= i\epsilon \dot{x}_j \dot{\psi}_j + \frac{i}{2} \left( -\epsilon\dot{x}_j\dot{\psi}_j - \psi_j\epsilon\ddot{x}_j  \right) \nn\\
&= i\epsilon \left( \dot{x}_j \dot{\psi}_j  - \frac{1}{2}\dot{x}_j \dot{\psi}_j + \frac{1}{2}\epsilon \ddot{x_j}\psi_j \right) \nn\\
&= \frac{i\epsilon}{2} \ddt\left(\dot{x}_j\psi_j \right)
\end{align}
where we used anti-commutativity between $\epsilon$ and $\psi_j$.

Thus $S = \int \dt L$ is invariant.
\end{proof}

If we had allowed $\epsilon$ to be time-dependent, we would have had
\begin{equation}
\delta S = \int \dt i \dot{\epsilon} \dot{x_j}\psi_j
\end{equation}
from which we define the \textbf{supercharge} $Q$ for this transformation:
\begin{equation}
Q = i \dot{x_j}\psi_j.
\end{equation}
Then in equation \ref{eq:Lvar1}
\begin{equation} \label{eq:Lvar}
\delta L = \frac{1}{2}\epsilon \frac{\mathrm{d}Q}{\dt}.
\end{equation}

Consider the change in $Q$ under our supersymmetry variation \ref{eq:supsymvar1}:
\begin{align} \label{eq:Qvar}
\delta Q &= i \dot{x_j}\delta \psi_j + i \left(\delta \dot{x_j}\right)\psi_j  \nn\\
&= -i\epsilon \dot{x}_j^2 - i^2 \epsilon \psi_j \dot{\psi}_j \nn\\
&= -2i\epsilon\left(\frac{1}{2} \dot{x}_j^2 + \frac{i}{2}\psi_j\dot{\psi_j}\right) \nn\\
&= -2i\epsilon L
\end{align}
where we used anti-commutativity of $\epsilon$, $\psi_j$ and $\dot{\psi}_j$.

So the variation of $Q$ produces the Lagrangian, which is a general feature of SUSY QM. Comparing equation \ref{eq:supsymvar1} with \ref{eq:Lvar} and \ref{eq:Qvar}, we see that the roles of bosonic and fermionic quantities have interchanged.

\subsubsection{Example 2: complex variables}
Let's consider a slightly more complicated model, with one bosonic variable $x$, two (complex) fermionic variables: $\psi$ and its complex conjugate $\bar{\psi} = \psi^{\dagger}$, and with Lagrangian:
\begin{equation}
L = \frac{1}{2}\dot{x}^2 + \frac{1}{2}i(\bar{\psi}\dot{\psi} - \dot{\bar{\psi}}\psi) - \frac{1}{2} h'(x)^2 - h''(x)\bar{\psi}\psi.
\end{equation}
The quantity $h(x)$ is called the \emph{superpotential}. 

\begin{propn} \label{propn:susyinv}
This system is invariant under the following SUSY transformation \cite{Mirrorsymmetry}
\begin{align} \label{eq:supsymvar2}
\delta x &= \epsilon \bar{\psi} - \bar{\epsilon}\psi \nn\\
\delta \psi &= \epsilon (i\dot{x} +h'(x)) \nn\\
\delta \bar{\psi} &= \bar{\epsilon}(-i\dot{x} + h'(x))
\end{align}
where $\epsilon$ is an infinitesimal complex Grassmann constant and $\epsbar$ its complex conjugate.
\end{propn}
\begin{proof}
We check:
\begin{align}
\delta L &= \dot{x}\delta\dot{x} + \frac{i}{2}\left(\delta\bar{\psi}\dot{\psi} + \bar{\psi}\frac{\dd (\delta\psi)}{\dt} - \frac{\dd (\delta \bar{\psi})}{\dt}\psi - \dot{\bar{\psi}}\delta\psi \right) \nn\\
&\qquad -h'\delta h' - \delta h'' \bar{\psi}\psi - h''\delta \bar{\psi} \psi - h'' \bar{\psi}\delta\psi.
\end{align}
Now use equation \ref{eq:supsymvar2} and the chain rule $\delta h' = h'' \delta{x};\,\,\delta h'' = h'''\delta x$ to see that $\delta h'' \bar{\psi}\psi = 0$ as $\psi^2 = 0 = \bar{\psi}^2$. Plugging everything in:
\begin{align}
\delta L &=\dot{x}(\epsilon\bar{\psi} -\bar{\epsilon}\psi) + \frac{i}{2}\left[ \bar{\epsilon}(-i\dot{x} + h')\dot{\psi}+\bar{\psi}\epsilon\left(i\ddot{x} + \ddt h'\right) - \bar{\epsilon}\left(-i\ddot{x} + \ddt h'\right)\psi \right. \nn\\
& \qquad \left.  -\dot{\bar{\psi}}\epsilon(i\dot{x}+h')\right] -h'h''(\epsilon\bar{\psi} - \bar{\epsilon}\psi) - h''\bar{\epsilon} (-i\dot{x} + h')\psi - h''\psibar\epsilon(i\dot{x}+h').
\end{align}
Now use the anti-commutation relations of the fermionic variables $\psi,\psibar,\epsilon,\epsbar$ and the chain rule: $\frac{\dd (h'(x))}{\dt} = h''(x) \dot{x}$ to get: 
\begin{align}
\delta L &=\epsilon\bar{\psi}\dot{x} -\bar{\epsilon}\psi\dot{x} + \frac{1}{2}\left[ \bar{\epsilon}\dot{x}\dot{\psi} + i\bar{\epsilon}h'\dot{\psi}+\epsilon \psibar\ddot{x} - i\epsilon\psibar\dot{x}h'' - \bar{\epsilon}\ddot{x}\psi - i \epsbar\dot{x} h''\psi - \epsilon \dot{\psibar} \dot{x}+i\epsilon\dot{\psibar} h'\right] \nn\\
& \qquad  -h'h'' \epsilon\bar{\psi} +h'h'' \bar{\epsilon}\psi +i h''\dot{x} \bar{\epsilon}\psi  - h''h'\bar{\epsilon} \psi + i h''\dot{x}\epsilon\psibar+ h''\epsilon\psibar h'.
\end{align}
A lot of terms cancel; furthermore we can group terms:
\begin{align}
\delta L &= \frac{\epsilon}{2}\left(\dot{x} \dot{\psibar} + \ddot{x}\psibar\right) - \frac{\epsbar}{2}\left(\dot{x} \dot{\psi} + \ddot{x}\psi\right) + \frac{i}{2}\epsilon \left(h'\dot{\psibar} + \dot{x}h''\psibar \right) + \frac{i}{2}\epsbar \left(h'\dot{\psi} + \dot{x}h''\psi \right)\nn\\
&= \ddt\left( \frac{1}{2}\epsilon \psibar\left( \dot{x}+ih'\right) + \frac{1}{2}\epsbar\psi\left(- \dot{x}+ ih'\right) \right)
\end{align}
which is a total derivative, thus not changing $S = \int L \dt$.
\end{proof}

So far, we have assumed $\epsilon, \epsbar$ are time-invariant; if we allow them to be time-dependent, then \cite{Mirrorsymmetry}
\begin{equation}
\delta S = \int \dt (-i\dot{\epsilon} Q - i\dot{\epsbar} \bar{Q})
\end{equation}
where $Q, \bar{Q}$ are the supercharges:
\begin{align}
Q &= \psibar (i\dot{x} + h'(x)),\nn\\
\bar{Q} &= \psi (-i\dot{x} +h'(x)).
\end{align}

The conjugate momenta for $x$ and $\psi$ are given by $p = \partial L/\partial \dot{x} = \dot{x}$ and $\pi_{\psi} = \partial L/\partial \dot{\psi} = i\psibar$. We then perform a Legendre transform on the Lagrangian to find the Hamiltonian:
\begin{align}
H &= p^2 + \frac{1}{2} \left(  \pi_{\psi} \dot{\psi} + \pi_{\bar{\psi}} \dot{\psibar} \right) - L \nn\\
&= \frac{1}{2} p^2 + \frac{1}{2} h'(x)^2 + h''(x)\psibar\psi.
\end{align}

Now let us quantize the system. We impose commutation relations for bosons and anti-commutation relations for fermions: \cite{Mirrorsymmetry}
\begin{align}
\left[ x, p \right] &= i \\
\left\{ \psi, \pi_{\psi} \right\} &= i
\end{align}
so that 
\begin{align} \label{eq:commrels}
\left[ x, p \right] &= i \\
\left\{ \psi, \psibar \right\} &= 1.
\end{align}
All other (anti-)commutators vanish. 

In quantizing the Hamiltonian, there is an operator ordering ambiguity; we choose  \cite{Mirrorsymmetry}
\begin{equation}
H = \frac{1}{2} p^2 + \frac{1}{2} h'(x)^2 + \frac{1}{2}h''(x)(\psibar\psi - \psi\psibar).
\end{equation}

Now define the vacuum state $\ket{0}$ as annihilated by $\psi$:
\begin{equation}
\psi \ket{0} = 0
\end{equation}
and define fermionic states $\left(\psibar\right)^n\ket{0}$ by using the ``raising operator" $\psibar$. Since $\psibar^2=0$, this is a 2-dimensional space spanned by 
\begin{equation}
\{ \ket{0} , \psibar\ket{0} \}.
\end{equation}

Thus the total Hilbert space of states is \cite{Mirrorsymmetry}
\begin{equation}
\mathcal{H} = \mathcal{H}^B \oplus \mathcal{H}^F
\end{equation}
where 
\begin{align}
\mathcal{H}^B &= L^2(\mathbb{R},\mathbb{C})\ket{0}\nn\\
\mathcal{H}^F &= L^2(\mathbb{R},\mathbb{C})\psibar\ket{0}
\end{align}
are the bosonic and fermionic spaces, respectively.\\
\\
Consider the fermion number operator $F$:
\begin{equation}
F = \psibar\psi.
\end{equation}

\begin{propn}
$F$ commutes with $H$.
\end{propn}
\begin{proof}
We calculate:
\begin{align}
2[F,H] &= \left[ \psibar\psi, p^2 + h'(x)^2 + h''(x)\left(\psibar\psi-\psi\psibar\right)\right] \nn\\
&= h''(x) \left[ \psibar\psi, \psibar\psi \right] - h''(x) \left[ \psibar\psi, \psi\psibar\right] \nn\\
&= -h''(x) \left( \psibar\psi \psi\psibar - \psi \psibar\psibar\psi \right) \nn\\
&= 0 
\end{align}
where we used the (anti-)commutation relations \ref{eq:commrels} and the identities $\psi^2 = 0 = \psibar^2$.
\end{proof}

By Heisenberg's equation of motion, $F$ is preserved. In fact, $F\ket{0} = \psibar\psi\ket{0} = 0$ and $F\psibar\ket{0} = \psibar\psi\psibar\ket{0} = \psibar\ket{0}$ (using $\{\psibar,\psi\} = 1$). So we see that $F$ takes the value $0$ on $\mathcal{H}^B$ and $1$ on $\mathcal{H}^F$. Hence we say the operator $(-1)^F$ provides a $\mathbb{Z}_2$ \emph{grading} on $\mathcal{H}$.\\
\\
Under quantization, the supercharges $Q, Q^{\dagger}$ are promoted to operators
\begin{align}
Q &= \psibar \left(ip + h'(x) \right) \\
Q^{\dagger} &= \psi \left(-ip + h'(x) \right).
\end{align}
Note that $Q^2 = 0 = \left(Q^{\dagger}\right)^2$ since $\psi^2= 0 = \psibar^2$.

They satisfy some important properties.

\begin{propn}
$Q$ and $Q^{\dagger}$ map $\mathcal{H}^B$ to $\mathcal{H}^F$ and vice versa.
\end{propn}
\begin{proof}
We show this for $Q$; the proof for $Q^{\dagger}$ is similar.

Consider states $\ket{\phi_B} \in \mathcal{H}^B$ and $\ket{\phi_F} \in \mathcal{H}^F$, i.e. $\ket{\phi_B}= f_B(x) \ket{0}$ and $\ket{\phi_F} = f_F(x)\psibar\ket{0}$. Then
\begin{equation}
Q\ket{\phi_B} = \psibar (ip + h'(x)) f_b(x) \ket{0} =  (ip + h'(x)) f_b(x) \psibar \ket{0} \in \mathcal{H}^F
\end{equation}
and
\begin{equation}
Q\ket{\phi_F} = \psibar (ip + h'(x)) f_F(x) \psibar \ket{0} =  (ip + h'(x)) f_F(x) \psibar^2 \ket{0} = 0  \in \mathcal{H}^B.
\end{equation}
\end{proof}

The following property will be crucial to us.
\begin{propn}
The anti-commutator of $Q$ and $Q^{\dagger}$ gives the Hamiltonian:
\begin{equation}
\boxed{ \left\{Q,Q^{\dagger} \right\} = 2H }.
\end{equation}
\end{propn}
\begin{proof}
We omit the proof here as this is just an expansion of anti-commutators using the canonical relations \ref{eq:commrels}. Details can be found in \cite{Mirrorsymmetry}.
\end{proof}

\section{General structure of supersymmetric quantum mechanics} \label{sec:SUSYstructure}
The examples in the previous section highlighted the structure of SUSY QM; in this section we shall provide a general definition of SUSY QM. We follow \cite{SiLi} and \cite{WittenSUSY}. 

\begin{defn} \textbf{(SUSY QM)}
Consider a quantum mechanical system consisting of a Hilbert space $\mathcal{H}$ and Hamiltonian $H$. It is \textbf{supersymmetrically quantum mechanical} (SQM) of type $N$ if \cite{SiLi}
\begin{enumerate}
\item $\mathcal{H}$ is $\mathbb{Z}_2$ graded by an operator $(-1)^F$:
	\begin{equation}
	\mathcal{H} = \mathcal{H}^B \oplus \mathcal{H}^F
	\end{equation}
	where 
	\begin{align}
	(-1)^F \ket{\phi} &=\ket{\phi} \qquad\quad\text{if } \phi \in \mathcal{H}^B\nn\\
	(-1)^F \ket{\phi} &=-\ket{\phi} \qquad \text{if } \phi \in \mathcal{H}^F.
	\end{align}
	We call $\mathcal{H}^B$ and $\mathcal{H}^F$ the bosonic and fermionic spaces respectively.
\item There are $N$ \emph{supercharges} $Q_I$ that anti-commute with $(-1)^F$:
	\begin{equation}
	\left\{ Q_I, (-1)^F \right\} = 0 = \left\{ Q^{\dagger}_I, (-1)^F \right\}.
	\end{equation}
	Therefore $Q_I$ and $Q^{\dagger}_I$ map bosons to fermions and vice versa:
	\begin{align}
	Q_I,Q_I^{\dagger}&: \mathcal{H}^B \rightarrow \mathcal{H}^F\nn\\
	Q_I,Q_I^{\dagger}&: \mathcal{H}^F \rightarrow \mathcal{H}^B.
	\end{align}
\item The supercharges satisfy the superalgebra condition:
	\begin{align}
	\left\{ Q_I, Q_J \right\} &= 0\\
	\left\{ Q_I, Q^{\dagger}_J \right\} &= 2\delta_{IJ} H.
	\end{align}
\end{enumerate}
We shall restrict ourselves to SQM models with a single supercharge $(N=1)$.
\end{defn}

Some important properties follow from the superalgebra condition.

\begin{cor} \label{cor:zeroenergy}
$H$ is a non-negative operator and 
\begin{equation}
H\ket{\phi} = 0 \quad \Leftrightarrow \quad Q\ket{\phi} = 0 = Q^{\dagger}\ket{\phi}.
\end{equation}
Furthermore $Q$ and $Q^{\dagger}$ commute with $H$:
\begin{equation}\label{eq:Qcommute}
\left[ Q, H \right] = 0 = \left[ Q^{\dagger}, H \right].
\end{equation}
\end{cor}
\begin{proof}
We note that if $Q\ket{\phi} = 0 = Q^{\dagger}\ket{\phi}$, then $H\ket{\phi} = 0$ is trivial.

For the other implication: suppose $H\ket{\phi} = 0$. Then $\braket{\phi|H|\phi}=0$. We expand:
\begin{equation}
0= \braket{\phi|2H|\phi} = \braket{\phi|QQ^{\dagger}|\phi} + \braket{\phi|Q^{\dagger}Q|\phi} = \braket{Q^{\dagger}\phi|Q^{\dagger}\phi} + \braket{Q\phi|Q\phi}
\end{equation}
which implies $Q\ket{\phi} = 0 = Q^{\dagger}\ket{\phi}$ by non-negativity of the inner product.\\
\\
As for equation \ref{eq:Qcommute}: we simply expand
\begin{equation}
2 \left[ Q, H \right] = \left[ Q, QQ^{\dagger} + Q^{\dagger}Q \right] = Q^2Q^{\dagger} +Q Q^{\dagger}Q - QQ^{\dagger}Q - Q^{\dagger}Q^2 = 0
\end{equation}
as $Q^2=0$. Similarly for $Q^{\dagger}$.
\end{proof}

Assuming the Hamiltonian has a countable spectrum, it gives us a $\mathbb{Z}_{\geq 0}$ grading on our Hilbert space, which can be restricted to the bosonic and fermionic spaces:
\begin{equation}
\mathcal{H} = \bigoplus_{n\in\mathbb{Z}_{\geq 0}} \mathcal{H}_n,\qquad \mathcal{H}^B = \bigoplus_{n\in\mathbb{Z}_{\geq 0}} \mathcal{H}^B_n,\qquad
\mathcal{H}^F = \bigoplus_{n\in\mathbb{Z}_{\geq 0}} \mathcal{H}^F_n
\end{equation}
where $\mathcal{H}_n$ is the $n$-th energy level, and $\mathcal{H}^B_n$ and $\mathcal{H}^F_n$ are its restrictions to $\mathcal{H}^B$ and $\mathcal{H}^F$ respectively.\\
\\
As $Q,Q^{\dagger}$ commute with $H$, they preserve the energy levels:
\begin{align}
Q,Q^{\dagger}&: \mathcal{H}^B_n \rightarrow \mathcal{H}^F_n\nn\\
Q,Q^{\dagger}&: \mathcal{H}^F_n \rightarrow \mathcal{H}^B_n.
\end{align}

\begin{propn}
For $n>0$: \cite{SiLi}
\begin{equation}
\boxed{\mathcal{H}^B_n \cong \mathcal{H}^F_n}.
\end{equation}
\end{propn}
\begin{proof}
For $n>0$, define $Q_n := \left(Q+Q^{\dagger}\right)/\sqrt{2E_n}$, which maps $\mathcal{H}^B_n$ to $\mathcal{H}^F_n$ and vice versa. The relation $\left\{ Q, Q^{\dagger}\right\} = 2H$ implies that $Q_n^2 = \frac{1}{2E_n} 2 E_n = I$ when restricted to the $n$-th energy level. Thus for $n>0$
\begin{align}
\left. Q_n \right| _{\mathcal{H}^B_n} &: \mathcal{H}^B_n \rightarrow \mathcal{H}^F_n\nn\\
\left. Q_n \right| _{\mathcal{H}^F_n}&: \mathcal{H}^F_n \rightarrow \mathcal{H}^B_n
\end{align}
are both invertible operators providing the required isomorphism.
\end{proof}

This means that bosonic and fermionic states at non-zero energies are paired. However this pairing generally fails to hold for the zero-energy supersymmetric ground states, and we define the \textbf{Witten index} to be the difference between the number of bosonic and fermionic supersymmetric ground states.

\begin{defn}
We define the \textbf{Witten index} to be $\dim \mathcal{H}^B_0 - \dim \mathcal{H}^F_0$.

By the isomorphism above, it satisfies: \cite{Mirrorsymmetry, SiLi}
\begin{equation}
\dim \mathcal{H}^B_0 - \dim \mathcal{H}^F_0 = \mathrm{Tr}\left( (-1)^F \right) = \mathrm{Tr}\left( (-1)^F e^{-\beta H} \right) 
\end{equation}
for any $\beta>0$.
\end{defn}

Since $Q^2 = 0$, it is natural to consider the cohomology of $Q$:
\begin{align}
H^B(Q) &= \frac{\mathrm{Ker}\, Q: \mathcal{H}^B \rightarrow \mathcal{H}^F}{\mathrm{Im}\, Q:\mathcal{H}^F \rightarrow \mathcal{H}^B} \nn\\
H^F(Q) &= \frac{\mathrm{Ker}\, Q: \mathcal{H}^F \rightarrow \mathcal{H}^B}{\mathrm{Im}\, Q:\mathcal{H}^B \rightarrow \mathcal{H}^F}.
\end{align}
At any excited level, $QQ^{\dagger} + Q^{\dagger}Q = 2E_n$, so the cohomology is trivial. (Explicitly: if $\ket{\phi} \in \mathcal{H}_n$ satisfies $Q\ket{\phi} = 0$, then $\ket{\phi} = Q\left( \frac{1}{2E_n} Q^{\dagger} \ket{\phi} \right) \in \mathrm{Im}\, Q$.)

However, the cohomology is non-trivial at zero energy, and by corollary \ref{cor:zeroenergy} we see that 
\begin{equation}
H^B(Q) \cong \mathcal{H}^B_0, \qquad H^F(Q) \cong \mathcal{H}^F_0
\end{equation}
so that the Witten index is given by
\begin{equation}
\mathrm{Tr}\left( (-1)^F e^{-\beta H} \right) = \dim H^B(Q) - \dim H^F(Q).
\end{equation}
It has a representation as a path integral: \cite{Mirrorsymmetry, SiLi}
\begin{equation}
\mathrm{Tr}\left( (-1)^F e^{-\beta H} \right) = \int_{PBC} \mathcal{D}\phi \mathcal{D}\psibar \mathcal{D}\psi \exp\left( -S_E(\phi,\psibar,\psi) \right)
\end{equation}
where we have absorbed the (infinite) normalization constant into the path measure and where $PBC$ denotes periodic boundary conditions:
\begin{equation}
\phi(0) = \phi(\beta), \psi(0) = \psi(\beta), \psibar(0) = \psibar(\beta).
\end{equation}
The condition $\phi(0) = \phi(\beta)$ comes from the fact that we are evaluating a trace in a Euclidean time path integral. The conditions $\psi(0) = \psi(\beta), \psibar(0) = \psibar(\beta)$ is a result from the fact that $(-1)^F$ is a fermionic operator and that the trace is cyclical. \cite{Mirrorsymmetry}

\section{Localization} \label{sec:localization}
In this section we examine localization, an important feature of SUSY QM.

First let's revisit the complex field example in section \ref{sec:susyexamples}. We get rid of the time variable to get action
\begin{equation}
S = -\frac{1}{2}(\partial h(x))^2 - \partial^2 h(x) \psi_1\psi_2 = S_0(X) - S_1(X)\psi_1\psi_2.
\end{equation}
This system is invariant under the transformation
\begin{align} 
\delta x &= \epsilon^1 \psi_1 + \epsilon^2\psi_2 \nn\\
\delta \psi_1 &= \epsilon^2 \partial h\nn\\
\delta \psi_2 &= -\epsilon^1\partial h
\end{align}
for infinitesimal Grassmann constants $\epsilon^1, \epsilon^2$.

Because there is no time variable, the path integral exists rigorously to give partition function \cite{Mirrorsymmetry}
\begin{equation}
Z := \frac{1}{\sqrt{2\pi}}\int \dd X \dd\psi_1\dd\psi_2  \exp\left(-S_0(X) + S_1(X)\psi_1\psi_2 \right).
\end{equation}

Now suppose that $\partial h \neq 0$ everywhere. We pick the supersymmetry transformation $\epsilon^1 = \epsilon^2 = -\psi^1 / \partial h$ to eliminate the $\psi_1$ variable:
\begin{equation} \label{eq:actioninvar}
S(X,\psi_1,\psi_2) = S(X', 0, \psi'_2) = S(X'), \qquad X = X' + g(X')\psi_1\psi_2
\end{equation}
where $g(X') = 1/ \partial h(X')$.

Then we evaluate:
\begin{align}
Z &= \frac{1}{\sqrt{2\pi}} \int \dd X \dd\psi_1 \dd\psi_2\, e^{-S(X,\psi_1,\psi_2)} \nn\\
&= \frac{1}{\sqrt{2\pi}} \int \dd X' \dd\psi_1 \dd\psi_2\, e^{-S(X')} \frac{\dd X}{\dd X'} \nn\\
&= \frac{1}{\sqrt{2\pi}} \int \dd X' \dd\psi_1 \dd\psi_2\, e^{-S(X')}( 1 + \partial g(X')\psi_1\psi_2)
\end{align}
where we used equation \ref{eq:actioninvar}. The first term does not survive the Grassmann integration and the final term is a total derivative, so that
\begin{equation}
Z = \frac{1}{\sqrt{2\pi}} \int \dd X'\, \partial g(X')e^{-S(X')} = 0.
\end{equation}

Now if $\partial h = 0 $ for some locus of points $L$, we can consider an $\epsilon$-small neighbourhood $L_{\epsilon}$ and its complement $L_{\epsilon}'$. By our previous argument, the path integral over $L_{\epsilon}'$ vanishes. Thus we see that the path integral is completely determined by an infinitesimal neighbourhood of the fixed points.

This is an example of the general \emph{localization principle}.

\begin{thm} \label{thm:localization} \textbf{(Localization principle)}
Consider a supersymmetric model with supersymmetry group $F$ leaving $S$ invariant. Then the path integral of an $F$-invariant operator $\mathcal{O}$ is completely determined by the loci where the fermionic supersymmetry transformation is zero \cite{Mirrorsymmetry, WittenLocalization}.
\end{thm}
\begin{proof}
We give a heuristic proof that explains the main idea. We follow Witten's argument in \cite[Section~$5$]{WittenLocalization}.

Let $\mathcal{E}$ be the function space we are integrating over. Suppose $F$ has no fixed points; then we can consider the quotient $\mathcal{E}/F$, which is a smooth space. As $\mathcal{O}$ and $S$ are $F$-invariant, the integral equals
\begin{equation}
\int_{\mathcal{E}}  e^{iS}\mathcal{O} = \mathrm{vol}(F) \int_{\mathcal{E}/F}  e^{iS}\mathcal{O} 
\end{equation}
where vol$(F)$ is the volume of the group $F$, which is $0$ for a fermionic group as
\begin{equation}
\int \dd \theta = 0
\end{equation}
for a fermionic variable $\theta$. Thus if $F$ is fermionic:
\begin{equation}
\int_{\mathcal{E}}e^{iS} \mathcal{O}  = 0.
\end{equation}

Now suppose $F$ has some fixed point locus $\mathcal{E}_0$. Let $\mathcal{C}_{\epsilon}$ be an $\epsilon$-small neighbourhood of $\mathcal{E}_0$ and $\mathcal{C}'_{\epsilon}$ its complement: $\mathcal{E} = \mathcal{C}_{\epsilon} \cup \mathcal{C}'_{\epsilon}$. The path integral splits into one over $\mathcal{C}_{\epsilon}$ and one over $\mathcal{C}'_{\epsilon}$. By our previous argument:
\begin{equation}
\int_{\mathcal{C}'_{\epsilon}}e^{iS} \mathcal{O} = 0.
\end{equation}
Therefore
\begin{equation}
\int_{\mathcal{E}}e^{iS} \mathcal{O}  = \int_{\mathcal{C}'_{\epsilon}}e^{iS} \mathcal{O}.
\end{equation}
Now let $\epsilon \rightarrow 0$ to get the result.
\end{proof}

We should compare this with the stationary phase approximation, where we found that for $S \gg \hbar$ the dominant contribution to the path integral comes from the classical path(s). The localization principle is of a much stronger form though, stating that the path integral is \emph{completely determined} by certain configurations. In other words: the extra structure of supersymmetry allows us to calculate more quantities exactly.

\section{Supersymmetry on Riemannian manifolds and geometrical theorems} \label{sec:geomthms}
In this section we look at an SQM model on Riemannian manifolds. We shall see how notions in supersymmetry are related to notions in geometry and shall use the path integral with supersymmetry to prove two geometrical theorems: the \emph{Gauss-Bonnet-Chern} and the \emph{Lefschetz fixed-point} theorems, which have far-reaching implications beyond physics. The proofs we give are not standard ones; in fact, these theorems were proven before supersymmetry was invented! However, once the SUSY machinery is in place, the proofs are remarkably simple, only requiring some long but elementary calculations. Furthermore, the method presented here can be extended to give a proof of the Atiyah-Singer index theorem (of which our theorems are special cases), for which ``standard" proofs not involving supersymmetry are not accessible to physicists.\\
\\
Consider a compact, oriented, Riemannian manifold $(M,g)$ of dimension $n$. We consider the SQM model with Lagrangian \cite{Mirrorsymmetry, ooguri, SiLi} \footnote{Note the sign before $R_{IJKL}$ is opposite to that in \cite{SiLi} and \cite{ooguri} as a result of a different sign convention for $R_{IJKL}$. Symmetries of $R_{IJKL}$ show our Lagrangian matches that in \cite{Mirrorsymmetry}.}
\begin{equation} \label{eq:manifoldlagrangian}
L = \frac{1}{2} g_{IJ} \dot{\phi}^I \dot{\phi}^J + \frac{i}{2} g_{IJ}\left( \psibar^I D_t \psi^J - D_t \psibar^I \psi^J \right) +\frac{1}{4} R_{IJKL} \psi^I \psi^J \psibar^K \psibar^L
\end{equation}
where $\phi^I$ are $n$ bosonic fields, $\psi^I$ and $\psibar^I$ are $n$ fermionic fields and
\begin{equation}
D_t \psi^I = \partial_t \psi^I + \Gamma^I_{\, JK} \dot{\phi}^J \psi^K
\end{equation}
where the $\Gamma^I_{\, JK}$ are Christoffel symbols associated to the Levi-Civita connection.

\begin{propn}
The model above is invariant under the supersymmetry \cite{Mirrorsymmetry, ooguri,SiLi}
	\begin{align}
	\delta \phi^I &= \epsilon \psibar^I - \epsbar \psi^I \nn\\
	\delta \psi^I &= \epsilon \left(i\dot{\phi}^I - \Gamma^I_{\, JK} \psibar^J \psi^K\right)\nn\\
	\delta \psibar^I &= \epsbar \left(-i\dot{\phi}^I - \Gamma^I_{\, JK} \psibar^J \psi^K\right)
	\end{align}
\end{propn}
\begin{proof}
The proof is similar to that of proposition \ref{propn:susyinv} and is most easily carried out using Riemann normal coordinates; we omit it here.
\end{proof}

The supercharges are \cite{Mirrorsymmetry, SiLi}
\begin{align}
Q &= i g_{IJ}\psibar^I \dot{\phi}^J = i \psibar^I P_I \nn\\
Q^{\dagger} &= - i g_{IJ}\psi^I \dot{\phi}^J = - i \psi^I P_I
\end{align}
where $P_I = g_{IJ} \dot{\phi}^J$ is the momentum conjugate to $\phi^I$.

The fermion number operator is 
\begin{align}
F = g_{IJ} \psibar^I \psi^J.
\end{align}

We quantize the system by imposing canonical (anti-)commutation relations
\begin{align}
\left[ \phi^I, P_J \right] &= \delta^I_J \nn\\
\left\{ \psi^I,\psibar^J \right\} &= g^{IJ}
\end{align}
with all other (anti-)commutators vanishing. \\
\\
The Hilbert space can be realized as the space of differential forms $\mathbb{C} \otimes \Omega(M)$ with the inner product \cite{Mirrorsymmetry}
\begin{equation}
(\omega_1, \omega_2 ) = \int_M \bar{\omega}_1 \wedge * \omega_2.
\end{equation}

With this realization the observables are
\begin{align} \label{eq:identification}
\phi^I &= x^I \times \nn\\
P_I &= -i \nabla_I \nn\\
\psibar^I &= \dd x^I \times \nn\\
\psi^I &= g^{IJ}\, i_{\partial_J}
\end{align}
where $\nabla$ is the Levi-Civita connection and $i_V$ denotes contraction of a differential form with vector field $V$.

Furthermore we have the correspondence:
\begin{eqnarray}
\ket{0} & \leftrightarrow & 1 \nn\\
\psibar^I\ket{0} &\leftrightarrow & \dd x^I \nn\\
\psibar^I\psibar^J\ket{0} &\leftrightarrow & \dd x^I\wedge \dd x^J \nn\\
 & \vdots & \nn\\
\psibar^I\dots\psibar^n\ket{0} &\leftrightarrow & \dd x^I\wedge \dots \wedge \dd x^n.
\end{eqnarray}

Most importantly to us, the supercharges and Hamiltonian are
\begin{eqnarray} \label{eq:Qcorr}
Q & \leftrightarrow &  \dd x^I \wedge \nabla_I = \dd \nn\\
Q^{\dagger}& \leftrightarrow &  \dddagger\nn\\
H =\frac{1}{2} \left\{Q,Q^{\dagger}\right\} & \leftrightarrow & \frac{1}{2}\Delta = \frac{1}{2} \left(\dd\dddagger + \dddagger\dd\right).
\end{eqnarray}

Thus the supersymmetric ground states correspond to harmonic forms. Furthermore, the grading by the fermion number operator $F$ corresponds to grading by form degree
\begin{equation}
\mathcal{H}_0 = \mathcal{H}(M,g) = \bigoplus_{p=0}^n \mathcal{H}^p(M,g).
\end{equation}

In section \ref{sec:SUSYstructure}, we saw that the Witten index can be found from the $Q$-cohomology:
\begin{equation}
\mathrm{Tr } (-1)^F = \sum_{p=0}^n (-1)^p \dim \mathcal{H}^p(M,g).
\end{equation}
Equation \ref{eq:Qcorr} implies that the $Q$-cohomology corresponds to the De-Rham cohomology, so that:
\begin{equation}
\boxed{ \mathrm{Tr } (-1)^F = \sum_{p=0}^n (-1)^p \dim \mathcal{H}^p(M,g) = \sum_{p=0}^n (-1)^p \dim H^p_{DR}(M,g) = \chi(M).} 
\end{equation}
Hence the Witten index is equal to the Euler number of the manifold!

This will be the starting point in proving the theorems in the next section. 

\subsection{Gauss-Bonnet-Chern theorem}
First we consider the Gauss-Bonnet-Chern theorem, a generalization of the Gauss-Bonnet theorem. It is primarily interesting as it relates a local quantity of a manifold, the curvature, to a global topological invariant, the Euler number. We follow the proof outline in \cite{ooguri, SiLi}, filling in many details.

\begin{thm}\textbf{(Gauss-Bonnet-Chern theorem)}
Consider a compact, oriented, Riemannian manifold $(M,g)$ of dimension $n$. Then if $n$ is odd:
\begin{equation} \label{eq:oddgaussbonnet}
\chi(M) = 0
\end{equation}
and if $n=2m$ is even: \cite{allenweil, ooguri, SiLi}  \footnote{The sign convention for $R_{IJKL}$ means that in \cite{SiLi} and \cite{ooguri}, the formula gains a prefactor $(-1)^m$.\\ Furthermore, they are missing a factor of $2^m$. \cite{SiLi} references Chern's original paper ``On the curvatura integra in a Riemannian manifold.", \emph{Ann. Math. 46 , 674(1942)}, which contains a version of the theorem using curvature 2-forms. In translating to an integral over $\dd V$, they forget a factor of $2^m$ coming from the Hodge dual. Our formula matches that in \cite{allenweil} and is seen to be correct by verification for $m=1$.}
\begin{equation} \label{eq:evengaussbonnet}
\chi(M) = \frac{1}{2^{3m}m!\pi^m} \int_M \dd V \epsilon^{I_1 J_1\dots I_m J_m} \epsilon^{K_1 L_1\dots K_m L_m} R_{I_1 J_1 K_1 L_1} \dots R_{I_m J_m K_m L_m}.
\end{equation}
In the case $n=2$, this reduces to the more elementary result commonly referred to as the Gauss-Bonnet theorem:
\begin{equation}
2\pi \chi(M) = \int_M K \dd A
\end{equation}
where $K = R/2 = \frac{1}{2} g^{IK}g^{JL} R_{IJKL}$ is the Gaussian curvature.
\end{thm}
\begin{proof}
Consider the Lagrangian:
\begin{equation}
L = \frac{1}{2} g_{IJ} \dot{\phi}^I \dot{\phi}^J + i g_{IJ}\psibar^I D_t \psi^J + \frac{1}{4} R_{IJKL} \psi^I \psi^J \psibar^K \psibar^L
\end{equation}
which differs from that in equation \ref{eq:manifoldlagrangian} by a total derivative: $\frac{i}{2} D_t \left( \psibar^I \psi^J \right)$ \cite{SiLi}. Hence it is invariant under the same supersymmetry transformation:
\begin{align} \label{eq:gaussbonnettrans}
	\delta \phi^I &= \epsilon \psibar^I - \epsbar \psi^I \nn\\
	\delta \psi^I &= \epsilon \left(i\dot{\phi}^I - \Gamma^I_{\, JK} \psibar^J \psi^K\right)\nn\\
	\delta \psibar^I &= \epsbar \left(-i\dot{\phi}^I - \Gamma^I_{\, JK} \psibar^J \psi^K\right).
\end{align}

The supercharges etc. are also unaffected, so that the Witten index equals the Euler number:
\begin{equation}
\mathrm{Tr } (-1)^F = \mathrm{Tr } \left( (-1)^F e^{-\beta H} \right) = \chi(M).
\end{equation}

We shall evaluate this via a path integral 
\begin{equation}
 \mathrm{Tr } \left( (-1)^F e^{-\beta H} \right) = \int_{PBC} \mathcal{D}\phi \mathcal{D}\psibar \mathcal{D}\psi \exp\left( -S_E(\phi,\psibar,\psi) \right).
\end{equation}

The Euclidean action with periodic boundary conditions is
\begin{equation}
S_E = \int_0^\beta \dt \left( \frac{1}{2} g_{IJ} \dot{\phi}^{I}\dot{\phi}^{J} + g_{IJ}\psibar^I D_t \psi^J -\frac{1}{4} R_{IJKL} \psi^I \psi^J \psibar^K \psibar^L \right).
\end{equation}

By the localization principle, the path integral localizes to the configurations for which the RHS of the fermionic part of transformation \ref{eq:gaussbonnettrans} vanishes. These are exactly the constant modes \cite{SiLi}.

Alternatively we could rescale $t = \beta \tau$ and $\psi \rightarrow \beta^{-1/4} \psi$ to get \cite{ooguri}
\begin{equation}
S_E = \int_0^1 \dd \tau \left(  \frac{1}{2\beta} g_{IJ} \frac{\dd\phi^{I}}{\dd\tau}\frac{\dd\phi^{j}}{\dd\tau} + \frac{1}{\sqrt{\beta}}g_{IJ}\psibar^I D_{\tau} \psi^J -\frac{1}{4} R_{IJKL} \psi^I \psi^J \psibar^K \psibar^L \right).
\end{equation}
Now use independence of the Witten index from $\beta$ to take the limit $\beta\rightarrow 0$ and see the path integral localizes to constant modes.\\
\\
Because of the periodic boundary conditions, we can do a Fourier expansion of the variables around these constant modes:
\begin{align}
\phi^I &= x_0^I + \sqrt{\beta} \sum_{k\neq0} a_k^I \exp\left(\frac{2\pi i k }{\beta}t\right)\nn\\
\psi^I &= \beta^{1/4}\psi_0^I +  \sum_{k\neq0} \psi_k^I \exp\left(\frac{2\pi i k}{\beta}t\right)\nn\\
\psibar^I &= \beta^{1/4}\psibar_0^I +  \sum_{k\neq0} \psibar_k^I \exp\left(\frac{2\pi i k}{\beta}t\right).
\end{align}
where the factors of $\beta$ are included to ensure independence of the path measure from $\beta$.

The path measure then becomes 
\begin{align}
\mathcal{D}\phi &\rightarrow \frac{\dd V}{(2\pi)^{n/2}} \prod_{k\neq0} \frac{\dd^n a_k^I}{(2\pi)^{n/2}} \nn\\
\mathcal{D}\psibar &\rightarrow \dd^n \psibar_0^I \prod_{k\neq0} \dd^n \psibar_k^I \nn\\
\mathcal{D}\psi &\rightarrow \dd^n \psi_0^I \prod_{k\neq0} \dd^n \psi_k^I
\end{align}
Note this is different for the bosonic and fermionic variables, because the fermionic variables do not pick up a $\sqrt{2\pi}$ under a Fourier transform due to the rules of Grassmann integration.\\
\\
As the path integral is invariant under coordinate transformations, we can work in Riemann normal coordinates centered around $x_0^I$ to see that\cite{SiLi}
\begin{align}
S_E &= \int_0^{\beta} \dt \left[ \frac{1}{\beta} \left( -\sum_{k\neq0} \frac{1}{2}|a_k^I|^2 (2\pi k i)^2 + \sum_{k\neq0} 2\pi k i \psibar_k^I\psi^I_k + \right.\right. \nn\\
&\qquad\qquad\left.\left. + \frac{1}{4} R_{IJKL}(x_0^I) \psi_0^I \psi_0^J \psibar_0^K \psibar_0^L \right) + \mathcal{O}(1) \right] \nn\\
&= \sum_{k\neq0} \left( 2\pi^2 k^2 |a_k^I|^2  + 2\pi k i \psibar_k^I\psi_k^I \right) + \frac{1}{4} R_{IJKL}(x_0^I)\psi_0^I \psi_0^J \psibar_0^K \psibar_0^L + \mathcal{O}(\beta).
\end{align}
There are no $(\bar{a}_k, a_j)$ or $(\bar{\psi}_k, \psi_j)$ cross-terms as these are multiples of $e^{2\pi(j-k)t/\beta}$, which do not survive the $t$-integral.

In the limit $\beta \rightarrow 0$, the path integral is then
\begin{align} \label{eq:wittenpath}
\int_{PBC} \mathcal{D}\phi \mathcal{D}\psibar \mathcal{D}\psi &\exp\left( -S_E(\phi,\psibar,\psi) \right) = \int \prod_{k\neq0} \frac{\dd^n a_k^I}{(2\pi)^{n/2}} \exp \left( -\sum_{k\neq0}  2\pi^2 k^2 |a_k^I|^2 \right) \nn\\
& \times \int \prod_{k\neq0} \dd^n \psibar_k^I \dd^n \psi_k^I \exp \left( -\sum_{k\neq0} 2\pi k i \psibar_k^I\psi_k^I \right) \nn\\
& \times \int \frac{\dd V}{(2\pi)^{n/2}} \dd^n \psibar_0^I  \dd^n\psi_0^I \exp\left(\frac{1}{4} R_{IJKL}(x_0^I) \psi_0^I \psi_0^J \psibar_0^K \psibar_0^L \right).
\end{align}
Consider the integration over the non-zero modes:
\begin{align}
&\int \prod_{k\neq0} \frac{\dd^n a_k^I}{(2\pi)^{n/2}} \exp \left( -\sum_{k\neq0}  2\pi^2 k^2 |a_k^I|^2 \right) \int \prod_{k\neq0} \dd^n \psibar_k^I \dd^n \psi_k^I \exp \left( -\sum_{k\neq0} 2\pi k i \psibar_k^I\psi_k^I \right) \nn\\
&=  \prod_{k\neq0} \int\frac{\dd^n a_k^I}{(2\pi)^{n/2}} \exp \left( -2\pi^2 k^2 |a_k^I|^2 \right) \times \prod_{k\neq0}  \int\dd^n \psibar_k^I \dd^n \psi_k^I \exp \left( - 2\pi k i \psibar_k^I\psi_k^I \right) \nn\\
\end{align}
We show this integral equals $1$ in one dimension; the $n$-dimensional product is simply the $n$-th power of this and is still $1$.\\
\\
Applying the standard Gaussian integrals from appendix \ref{sec:fresnel}, we get
\begin{align} \label{eq:infprods}
 \prod_{k\neq0} \int\frac{\dd a_k}{\sqrt{2\pi}} \exp &\left( -2\pi^2 k^2 |a_k|^2 \right) \times \prod_{k\neq0}  \int\dd \psibar_k \dd \psi_k \exp \left( - 2\pi k i \psibar_k\psi_k \right) \nn\\
&= \prod_{k\neq0} \left(\frac{1}{\sqrt{2\pi}} \sqrt{\frac{\pi}{2\pi^2k^2}} \right)\prod_{k\neq0} (-2\pi k i) \nn\\
&= \prod_{k\neq0} \left( -i \frac{k}{|k|} \right) =\left(\prod_{j\neq0} (-1)\right)\left( \prod_{k \neq 0 } i \right)\left(\prod_{m \neq 0} \mathrm{sgn}(m)\right).
\end{align}

We now use zeta-regularization, so that
\begin{equation}
\prod_{k \geq 1} b = \frac{1}{\sqrt{b}}
\end{equation}
for a constant $b$. Then
\begin{equation}
\prod_{k \neq 0} b = \left( \prod_{k \geq 1} b \right)^2 = \frac{1}{b}.
\end{equation}
Thus
\begin{align}
&\prod_{j\neq0} (-1) = -1 \nn\\
&\prod_{k \neq 0 } i = \frac{1}{i} \nn\\
&\prod_{m \neq 0} \mathrm{sgn}(m) = \prod_{m \leq -1} (-1) = \frac{1}{\sqrt{-1}} = \frac{1}{i}.
\end{align}
Putting this in equation \ref{eq:infprods} gives
\begin{equation}
\prod_{k\neq0} \int\frac{\dd a_k}{\sqrt{2\pi}} \exp \left( -2\pi^2 k^2 |a_k|^2 \right) \times \prod_{k\neq0}  \int\dd \psibar_k \dd \psi_k \exp \left( - 2\pi k i \psibar_k\psi_k \right) = \frac{-1}{i^2} = 1
\end{equation}
as claimed.\\
\\
Thus we see in equation \ref{eq:wittenpath} that:
\begin{align} 
& \int_{PBC} \mathcal{D}\phi \mathcal{D}\psibar \mathcal{D}\psi \exp\left( -S_E(\phi,\psibar,\psi) \right) \nn\\
& = (2\pi)^{-n/2} \int_M \dd V \int \dd^n\psibar_0^I\dd^n\psi_0^I \exp \left(\frac{1}{4} R_{IJKL}(x_0^I) \psi_0^I \psi_0^J \psibar_0^K \psibar_0^L \right).
\end{align}
By the rules of Grassmann integration, only terms of the form $\psi_0^1 \dots \psi_0^n \psibar_0^1 \dots \psibar_0^n$ in the Taylor expansion will contribute. \\
\\
There are two cases:
\begin{itemize}
\item If $n$ is odd, there is no such term, so that
		\begin{equation}
		\int_{PBC} \mathcal{D}\phi \mathcal{D}\psibar \mathcal{D}\psi \exp\left( -S_E(\phi,\psibar,\psi) \right) = 0.
		\end{equation}
		Therefore
		\begin{equation}
		\chi(M) = \mathrm{Tr }(-1)^F =\int_{PBC} \mathcal{D}\phi \mathcal{D}\psibar \mathcal{D}\psi \exp\left( -S_E(\phi,\psibar,\psi) \right) = 0
		\end{equation}
		which proves equation \ref{eq:oddgaussbonnet}.
\item If $n = 2m$ is even, the term of power $m$ in the exponential gives the only non-zero contribution. From expanding the exponential, it has a prefactor $\frac{1}{m!} \left(\frac{1}{4} \right)^m = \frac{1}{2^{2m}m!}$, and the terms are of the form
		\begin{align}
		\int \dd \psi_0^1 \dots \dd \psi_0^n \dd \psibar_0^1& \dots \dd \psibar_0^n\,  R_{I_1 J_1 K_1 L_1} \dots R_{I_m J_m K_m L_m}\nn\\
		& \times \psi_0^{I_1}\psi_0^{J_1}\psibar_0^{K_1}\psibar_0^{K_1} \dots \psi_0^{I_m}\psi_0^{J_m}\psibar_0^{K_m}\psibar_0^{K_m}.
		\end{align}
		By the rules of Grassmann integration, the ordering matters and the integral picks up a factor of $\mathrm{sgn}(\tau)$, where $\tau$ is the permutation $\tau = \tau_{IJ}\tau_{KL}$, where e.g.
		\begin{equation}
		\tau_{IJ} = \begin{pmatrix} 1 & 2 & 3 & 4 & \dots & 2m -1 & 2m \\ I_1 & J_1 & I_2 & J_2& \dots & I_m & J_m \end{pmatrix} .
		\end{equation}
		In tensor form, this is represented by the Levi-Civita tensor: 
		$\epsilon^{I_1J_1\dots I_mJ_m}\epsilon^{K_1L_1\dots K_mL_m} $.\\
		\\
		Putting this all together:
		\begin{align}
			\chi(M) & = \mathrm{Tr }(-1)^F =\int_{PBC} \mathcal{D}\phi \mathcal{D}\psibar \mathcal{D}\psi \exp\left( -S_E(\phi,\psibar,\psi) \right) \nn\\
			&= (2\pi)^{-n/2} \int_M \dd V \int\dd^n\psibar_0^I \dd^n\psi_0^I \exp \left(\frac{1}{4} R_{IJKL}(x_0^I) \psi_0^I \psi_0^J \psibar_0^K \psibar_0^L \right) \nn\\
			&= (2\pi)^{-m} \int_M \dd V \frac{1}{2^{2m}m!} \epsilon^{I_1J_1\dots I_mJ_m}\epsilon^{K_1L_1\dots K_mL_m}R_{I_1 J_1 K_1 L_1} \dots R_{I_m J_m K_m L_m} \nn\\
			&= \frac{1}{2^{3m}m!\pi^m} \int_M \dd V \epsilon^{I_1J_1\dots I_mJ_m}\epsilon^{K_1L_1\dots K_mL_m} R_{I_1 J_1 K_1 L_1} \dots R_{I_m J_m K_m L_m} 
		\end{align}
		which proves equation \ref{eq:evengaussbonnet}. Note that even though we used Riemann normal coordinates, this is a tensor identity and hence holds in all coordinates.
\end{itemize}
Thus we have shown the Gauss-Bonnet-Chern theorem.\\
\\
For the $2$-dimensional case: set $m=1$ in the formula above to get
\begin{equation} \label{eq:surfcurv}
2\pi \chi(M) = \frac{1}{4} \int_M \dd A\, \epsilon^{IJ}\epsilon^{KL} R_{IJKL}.
\end{equation}
Now use $\epsilon^{00} = 0 = \epsilon^{11}$ and $\epsilon^{01} = 1, \epsilon^{10} = -1$ with the symmetries $R_{IJKL} = -R_{JIKL} = - R_{IJLK}$ to get:
\begin{equation} \label{eq:surfcurv1}
\epsilon^{IJ}\epsilon^{KL} R_{IJKL} = R_{0101} + R_{1010} - R_{0110} - R_{1001} = 4 R_{0101}.
\end{equation}
Also, in Riemann normal coordinates $g^{IJ}(x_0) = \delta^{IJ}$, so that the scalar curvature $R$ satisfies
\begin{equation}\label{eq:surfcurv2}
R = g^{IK} g^{JL} R_{IJKL} = \delta^{IK} \delta^{JL} R_{IJKL} = R_{0000} + R_{0101} + R_{1010} + R_{1111} = 2R_{0101}.
\end{equation}
Combining equations \ref{eq:surfcurv1} and \ref{eq:surfcurv2} yields $K = R/2 = \frac{1}{4}\epsilon^{IJ}\epsilon^{KL} R_{IJKL}$. Thus equation \ref{eq:surfcurv} implies that
\begin{equation}
2\pi \chi(M) = \frac{1}{4} \int_M \epsilon^{IJ}\epsilon^{KL} R_{IJKL}\dd A  = \int K \dd A.
\end{equation}
as claimed.
\end{proof}

\subsection{Lefschetz fixed-point theorem}
Now we prove the Lefschetz fixed-point theorem, which relates the index of fixed points of a smooth map $f:M\rightarrow M$ (a local quantity) to a global quantity $\Lambda_f$. The proof is loosely based on the proof in \cite{SiLi}. However, it contains some major errors that we believe have been corrected here. A large part of the proof of Gauss-Bonnet carries over.

\begin{thm}\textbf{(Lefschetz fixed-point theorem)}
Let $f:M\rightarrow M$ be a smooth map from a compact, oriented, Riemannian manifold $M$ to itself with a finite number of (necessarily isolated) fixed points. Define 
\begin{equation}
\Lambda_f = \sum_{q\geq0} (-1)^q \mathrm{Tr} \left( f^*_q \right)
\end{equation}
where $f^*_q = f^* | H^q_{DR}(M)$ is the restriction of the pullback $f^*$ to the $q$-th De Rham cohomology $H^q_{DR}(M)$.

Further define for a fixed point $p$ of $f$ the index $i(f,p)$:
\begin{equation}
i(f,p) = \mathrm{sgn} \left(\det \left( D_p f - \mathds{1} \right)\right).
\end{equation}

Then
\begin{equation}
\Lambda_f = \sum_{\mathrm{fixed}\,\mathrm{points}\, p} i(f,p).
\end{equation}
\end{thm}

\begin{proof}
We shall consider the SUSY Lagrangian
\begin{equation}
L = \frac{1}{2} g_{IJ} \dot{\phi}^I \dot{\phi}^J - i g_{IJ}D_t \psibar^I \psi^J +\frac{1}{4} R_{IJKL} \psi^I \psi^J \psibar^K \psibar^L
\end{equation}
which differs from that in equation \ref{eq:manifoldlagrangian} by a total derivative: $-\frac{i}{2} D_t \left( \psibar^I \psi^J \right)$, so is invariant under the same SUSY transformations.

As with the Witten index, there is a path integral expression for $\Lambda_f$:
\begin{equation}\label{eq:lefschetzpath}
\Lambda_f = \mathrm{Tr} \left( (-1)^F e^{-\beta H} f^* \right) = \int_{BC} \mathcal{D}\phi\mathcal{D}\psibar\mathcal{D}\psi\, e^{-S_E} 
\end{equation}
where the boundary conditions $BC$ are to be determined. Again it is exactly the supersymmetry that ensures this is independent of $\beta$.

To find the boundary conditions, recall from definition \ref{defn:pullback} how $f^*_q$ acts. Consider the $q$-form
\begin{equation}
\alpha_q = \sum_{i_1, \dots, i_q} a_{i_1,\dots, i_q} (x) \dd x^{i_1} \wedge \dots \wedge \dd x^{i_q} = \sum_{i_1, \dots, i_q} a_{i_1,\dots, i_q} (\phi(t)) \dd x^{i_1} \wedge \dots \wedge \dd x^{i_q}.
\end{equation}

Then $f^*_q(\alpha_q)$ is
\begin{align} \label{eq:lefschetzfact}
f^*_q (\alpha_q) &= \sum_{i_1, \dots, i_q} a_{i_1,\dots, i_q} (f(x))\dd \left(  f^{i_1}(x)\right) \wedge \dots \wedge \dd\left( f^{i_q}(x) \right) \nn\\
&= \sum_{i_1, \dots, i_q} a_{i_1,\dots, i_q} (f\circ\phi)(t) \dd \left(  f^{i_1}(x)\right) \wedge \dots \wedge \dd\left( f^{i_q}(x) \right).
\end{align}

Recalling the identification in \ref{eq:identification}, we see that $f^*_q$ acts by sending:
\begin{align}
\phi &\rightarrow f\circ \phi \nn\\
\psibar &\rightarrow  Df \circ \psibar \nn\\
\psi &\rightarrow \psi
\end{align}

so that the boundary conditions in equation \ref{eq:lefschetzpath} are:
\begin{align} \label{eq:lefschetzbc}
\phi(\beta) &= f (\phi(0)) \nn\\
\psibar(\beta) &= Df (\psibar(0)) \nn\\
\psi(\beta) &= \psi(0).
\end{align}

By the same argument as in our proof of Gauss-Bonnet-Chern, the path integral localizes to the constant maps. Due to the boundary conditions \ref{eq:lefschetzbc}, these are just the constant maps $\phi$ to fixed points $p$ of $f$, since $p := \phi(\beta) = \phi(0)$ implies in \ref{eq:lefschetzbc} that $f(p) = p$. We perform a Fourier expansion in local coordinates around $p$, respecting the boundary conditions:
\begin{align} \label{eq:lefschetzfourier}
\phi^I(t) &= \frac{t}{\beta} f^I(\sqrt{\beta} x_0) + \left(1-\frac{t}{\beta}\right) \sqrt{\beta} x^I_0 + \sqrt{\beta} \sum_{k\neq0} a_k^I \exp\left(\frac{2\pi i k }{\beta}t\right)\nn\\
\psibar^I (t) &= \left[\frac{t}{\beta} Df^I(\psibar_0) + \left(1-\frac{t}{\beta}\right) \psibar^I_0 \right]+ \sum_{k\neq0} \psibar_k^I \exp\left(\frac{2\pi i k}{\beta}t\right)\nn\\
\psi^I (t) &= \psi^I_0 + \sum_{k\neq0} \psi_k^I \exp\left(\frac{2\pi i k}{\beta}t\right)
\end{align}

where the factors of $\beta$ have been included to ensure the quadratic terms in the action are $\beta$-independent and that the path measure is $\beta$-independent. %Note that linearity of $D f$ is crucial for $\psibar$ to respect the boundary condition.

We Taylor expand $f$ to get:
\begin{align} \label{eq:lefschetztaylorexp}
\frac{\dd\phi^I}{\dd t} &= \frac{1}{\beta} \sqrt{\beta} \left((Df)^I_J x_0^J - x_0^I\right) + \mathcal{O}(1) + \frac{1}{\sqrt{\beta}} \sum_{k\neq0} 2\pi i k a_k^I \exp\left(\frac{2\pi i k }{\beta}t\right)\nn\\
&= \frac{1}{\sqrt{\beta}}\left( \left((Df)^I_J - \delta^I_J\right) x_0^J + \sum_{k\neq0} 2\pi i k a_k^I \exp\left(\frac{2\pi i k }{\beta}t\right)  \right) +\mathcal{O}(1)
\end{align}
where we used $f(0) =0$, since we are expanding around a fixed point.

Furthermore, if we again use Riemann normal coordinates, then
\begin{align}
D_t \psibar^I = \frac{\dd\psibar^I}{\dd t} &= \frac{1}{\beta} \left( (Df)^I_K \psibar^K_0 - \psibar_0^I\right)+ \frac{1}{\beta} \sum_{k\neq0} 2 \pi i k \psibar_k^I \exp\left(\frac{2\pi i k}{\beta}t\right) \nn\\
&= \frac{1}{\beta} \left( (Df)^I_K  - \delta^I_K\right)\psibar^K_0 + \frac{1}{\beta} \sum_{k\neq0} 2 \pi i k \psibar_k^I\exp\left(\frac{2\pi i k}{\beta}t\right)
\end{align}

Similarly as in proving Gauss-Bonnet, we then find Euclidean action (with Riemann normal coordinates)
\begin{align}
S_E = \int_0^{\beta} \dd t\,& \left[ \frac{1}{2} \delta_{IJ} \frac{1}{\beta} \left((Df)^I_K - \delta^I_K\right) x_0^K \left((Df)^J_L - \delta^J_L\right) x_0^L  -  \frac{1}{\beta} \sum_{k\neq0} \frac{1}{2}|a_k^I|^2 (2\pi k i)^2\right. \nn\\
&  - \frac{1}{\beta}\sum_{k\neq0} 2\pi k i \psibar^I_k\psi^I_k- \delta_{IJ} \frac{1}{\beta} \left( (Df)^I_K  - \delta^I_K\right)\psibar^K_0 \psi^J_0 + \nn\\
& +  \sum_{k\neq 0}(\psibar_0,\psi_k)\text{-cross-terms } + \left.  \sum_{k\neq 0}(\psi_0,\psibar_k)\text{-cross-terms } + \mathcal{O}(1) \right].
\end{align}
%\delta_{IJ} \frac{1}{\beta}  \sum_{k\neq0} 2 \pi i k \psibar_k^I\psi_k^J \exp\left(\frac{2\pi i k}{\beta}t\right) +

The $(x_0,a_k)_{k\neq 0}$ and $(\psibar_l,\psi_{k-l})_{k\neq 0}$ cross-terms are not present as they are multiples of $\int_0^{\beta} \dt e^{2\pi i k t / \beta} = 0$. Furthermore, the $(\psibar_0,\psi_k)_{k\neq 0}$ and $(\psi_k,\psibar_0)_{k\neq 0}$ cross-terms will not survive the Grassmann integration in the path integral. 

When taking the limit $\beta\rightarrow 0$, which is allowed as the path integral is independent of $\beta$, we can ignore the last term, which is $\mathcal{O}(\beta)$. Also, the remaining integrals over non-zero modes $|a_k|^2, \psibar_k\psi_k$ cancel each other as in the proof of Gauss-Bonnet. Hence we can use the ``effective" action:
\begin{align}
\hat{S}_E &= \frac{1}{2}\delta_{IJ} x_0^K \left((Df)^I_K - \delta^I_K\right) \left((Df)^J_L - \delta^J_L\right) x_0^L - \delta_{IJ} \left( (Df)^I_K  - \delta^I_K\right)\psibar^K_0 \psi^J_0 \nn\\
&= \frac{1}{2} \left( (Df - \mathds{1}) \boldsymbol{x_0} \right)^T \left( (Df - \mathds{1}) \boldsymbol{x_0}\right)  - \left( (Df - \mathds{1}) \boldsymbol{\psibar_0} \right)^T \boldsymbol{\psi_0} \nn\\
&= \frac{1}{2} \boldsymbol{x_0}^T (Df - \mathds{1})^T(Df - \mathds{1}) \boldsymbol{x_0} - \boldsymbol{\psibar_0}^T \left( Df - \mathds{1}\right)^T \boldsymbol{\psi_0}
\end{align}
where we regard $(Df - \mathds{1})$ as a matrix and  $\boldsymbol{x_0}, \boldsymbol{\psibar_0}, \boldsymbol{\psi_0}$ as vectors.\\
\\
The path measures are exactly as before, to give as contribution around $p$: 
\begin{align} \label{eq:lefschetzint}
\int_{p,BC} \mathcal{D}\phi\mathcal{D}\psibar\mathcal{D}\psi e^{-\hat{S}_E} &= \int \frac{\dd^n \boldsymbol{x_0}}{(2\pi)^{n/2}} \exp \left( -\frac{1}{2} \boldsymbol{x_0}^T (Df - \mathds{1})^T(Df - \mathds{1}) \boldsymbol{x_0} \right) \times \nn\\
& \qquad \times\int \dd^n \boldsymbol{\psibar_0} \dd^n \boldsymbol{\psi_0} \exp\left(\boldsymbol{\psibar_0}^T \left( Df - \mathds{1}\right)^T \boldsymbol{\psi_0} \right).
\end{align}

Now use the Gaussian integral formulas from appendix \ref{sec:fresnel}:
\begin{align}
\int \frac{\dd^n \boldsymbol{x_0}}{(2\pi)^{n/2}} &\exp \left( -\frac{1}{2} \boldsymbol{x_0}^T (Df - \mathds{1})^T(Df - \mathds{1}) \boldsymbol{x_0} \right)\nn\\
 &= \frac{1}{\sqrt{ \det \left( (Df - \mathds{1})^T  (Df - \mathds{1}) \right)}} \nn\\
&= \frac{1}{\left| \det(Df - \mathds{1}) \right|}
\end{align}

and its Grassmannian version:
\begin{equation}
\int \dd^n \boldsymbol{\psibar_0} \dd^n \boldsymbol{\psi_0} \exp\left(\boldsymbol{\psibar_0}^T \left( Df - \mathds{1}\right)^T \boldsymbol{\psi_0} \right) = \det \left( Df - \mathds{1} \right)^T = \det \left( Df - \mathds{1} \right).
\end{equation}

Combining these results in equation \ref{eq:lefschetzint}, we get the contribution from a fixed point $p$:
\begin{align}
\int_{p,BC} \mathcal{D}\phi\mathcal{D}\psibar\mathcal{D}\psi e^{-\hat{S}_E} &= \frac{\det \left( Df - \mathds{1} \right)}{\left| \det(Df - \mathds{1}) \right|} \nn\\
&= \text{sgn} \left( \det \left( Df - \mathds{1} \right)\right) \nn\\
&= i(f,p).
\end{align}

Summing the contributions from all fixed points $p$, we arrive at the final result:
\begin{equation}
\boxed{\Lambda_f = \int_{BC} \mathcal{D}\phi\mathcal{D}\psibar\mathcal{D}\psi e^{-S_E} = \sum_{\mathrm{fixed}\,\mathrm{points}\, p} i(f,p).}
\end{equation}

\end{proof}

As a corollary, we arrive at a version of Brouwer's fixed-point theorem \footnote{Brouwer's fixed-point theorem only assumes continuity of $f$; note we assume additionally that $f$ is smooth.}:

\begin{cor}\textbf{(Brouwer's fixed-point theorem)}
Let $f: D^n \rightarrow D^n$ be a smooth map from the unit disk to itself. Then $f$ has a fixed point.
\end{cor}
\begin{proof}
Recall the De Rham cohomologies of $D^n$:
\begin{equation}
H^p_{DR}(D^n) = \begin{cases} \mathbb{R} &\text{if } p = 0 \\ 0 & \text{else} \end{cases}
\end{equation}
In fact, we saw that $H^0_{DR}(D^n)$ consists simply of constant maps. From equation \ref{eq:lefschetzfact}, we conclude that any map $f$ induces the identity on $H^0_{DR}(D^n)$. Hence 
\begin{equation}
\Lambda_f = 1.
\end{equation}
So the Lefschetz fixed-point theorem implies $f$ has at least one fixed point. 
\end{proof}

Compare this with the \emph{contraction mapping theorem} from topology. Let $f$ be a contraction mapping $f: M \rightarrow M$ on a metric space $M$, i.e. a mapping such that there exists a real number $0 \leq K < 1$ such that for all $x,y \in M$:
\begin{equation}
\dd(f(x),f(y)) \leq K\dd(x,y).
\end{equation}
Then the contraction mapping theorem states that $f$ has a unique fixed point. This proves the statement that if you are in Oxford and pull out a map of England, there will be exactly one spot on the map that is physically in the place it points to.

The contraction mapping condition is stronger than continuity (any contraction mapping is Lipschitz-continuous by definition), but this ensures uniqueness of the fixed point. Furthermore Brouwer's fixed-point theorem is non-constructive, whereas the contraction mapping theorem is - its proof involves taking an arbitrary point $x_0 \in M$ and defining a sequence $x_{n+1} = f(x_n)$; then this sequence converges to the unique fixed point $p$.  \\
\\
We verify the Lefschetz and Brouwer fixed-point theorems for a simple rotation.
\begin{ex}
Consider the map $f: D^2\rightarrow D^2$ on the unit disk that rotates through $\theta$:
\begin{equation}
\begin{pmatrix} x \\ y \end{pmatrix} \mapsto \begin{pmatrix} \cos\theta & -\sin\theta \\ \sin\theta & \cos\theta\end{pmatrix} \begin{pmatrix} x \\ y \end{pmatrix}.
\end{equation}
Then $f$ has a single fixed point at the origin, where it has derivative:
\begin{equation}
D_0 f = \begin{pmatrix} \cos\theta & -\sin\theta \\ \sin\theta & \cos\theta\end{pmatrix}
\end{equation}
as it is linear. Hence
\begin{equation}
\det \left( D_0 f - \mathds{1} \right) = \begin{vmatrix} -1 + \cos\theta & -\sin\theta \\ \sin\theta & -1 + \cos\theta\end{vmatrix} = (-1 + \cos\theta)^2 + \sin^2\theta > 0
\end{equation}
so that $\mathrm{sgn}\left(\det \left( D_0 f - \mathds{1} \right)\right) = 1$ and indeed $\Lambda_f = 1$.
\end{ex}

Finally, we state a connection between Lefschetz and Euler numbers.

\begin{cor}
Consider a compact, oriented, Riemannian manifold $M$ and $\{ f_t \}$, a $1$-parameter group of maps $f_t: M \rightarrow M$ continuously connected to the identity. Then
\begin{equation}
\Lambda_{f_t} = \chi(M)
\end{equation}
for any $f_t$.
\end{cor}
\begin{proof}
The identity map $I$ trivially induces the identity on all cohomologies. Thus, from the definition of $\Lambda_f$:
\begin{equation}
\Lambda_I = \sum_{q=0}^n (-1)^q \mathrm{Tr} \left(I | H^q_{DR} (M) \right) = \sum_{q=0}^n (-1)^q \mathrm{dim} \left(H^q_{DR} (M)\right) = \chi(M)
\end{equation}
as the trace of the identity gives the dimension.

Each $f_t$ has a Lefschetz number $\Lambda_{f_t}$ depending continuously on $t$. Furthermore, as an integer, it is constant on the connected component of $f_t$, which contains the identity. Therefore
\begin{equation}
\Lambda_{f_t} = \Lambda_I = \chi(M).
\end{equation} 
\end{proof}

This shows the main power of the Lefschetz fixed-point theorem: $\Lambda_f$ is invariant under continuous deformations of $f$, and we can often reduce calculations to simple ones.

\begin{ex} \textbf{(Sphere and torus)} 
\begin{itemize}
\item Consider rotations $R(\theta)$ of $S^2$, through angle $\theta$, around an axis through the North and South poles. These are connected to the identity, as $I = R(\theta = 0).$ For $\theta \notin 2\pi \mathbb{Z}$, the poles are its only fixed points, around which it locally looks like the $2$D rotation in our previous example. Then indeed:
\begin{equation}
\Lambda_{R(\theta)} = i(R(\theta), NP) + i(R(\theta), SP) = 1 + 1 = 2 = \chi(S^2).
\end{equation}
\item Similarly consider rotations $F(\theta)$ of the $2$-torus $T^2$ around a vertical axis through its ``hole". For $\theta \notin 2\pi \mathbb{Z}$, it has no fixed points. Thus
\begin{equation}
\chi(T^2) = \Lambda_{F(\theta)} = 0
\end{equation}
as expected.
\end{itemize}

\end{ex}

\chapter*{Conclusion}
In this dissertation, we introduced the path integral as a formulation of quantum mechanics and analysed some of its physical and mathematical properties. We introduced the modern idea of supersymmetry and showed how the path integral was naturally suited to supersymmetric calculations. This culminated in ``physics proofs" of the Gauss-Bonnet-Chern and Lefschetz fixed-point theorems. 

The expert reader might recognize that these are part of a wider class of index theorems that follow from the more general Atiyah-Singer index theorem, which can also be proven by a path integral in an appropriate supersymmetric model (see e.g. \cite{Nakahara}). Originally a proof of this theorem was intended, but due to the extra background knowledge required for this, only two special cases have been included. We hope that this provides a useful introduction to students interested in the field and serves as a good starting point for further study.

\appendix

\chapter{Mathematical results} \label{chap:appendix}

\section{Gaussian and Fresnel integrals} \label{sec:fresnel}
In this appendix, we prove the Gaussian and Fresnel integrals used throughout the dissertation.

\begin{thm} \textbf{(Gaussian integral)}
Let $A$ be an $n\times n$ symmetric, positive-definite matrix. Then
\begin{equation}
\int \mathrm{d}^n\boldsymbol{x}\exp \left(- \boldsymbol{x}^T A \boldsymbol{x}\right) = \frac{\pi^{n/2}}{\sqrt{\det A}} = \frac{1}{\sqrt{\det (A/\pi)}}.
\end{equation}
\begin{proof}
We shall use, without proof, the standard result:
\begin{equation} \label{eq:1dgaussian}
\int \dx \exp \left(- ax^2\right) = \sqrt{\frac{\pi}{a}}.
\end{equation}
(This can be proven by squaring the integral and evaluating it in polar coordinates.)\\
\\
As $A$ is symmetric and positive definite, then by the spectral theorem from linear algebra, there is a matrix $O$ such that $O^TO = \mathds{1} $ and $O^TAO = D$ is diagonal: $D = \mathrm{diag}(\lambda_1,\dots,\lambda_n)$, where the $\lambda_n$ are the eigenvalues of $A$. Specifically then $\det D = \det A$. \\
Introduce coordinates $\boldsymbol{y} = O \boldsymbol{x}$. Then as $O$ is orthogonal: $\boldsymbol{x}^T A \boldsymbol{x} = \boldsymbol{y}^T D \boldsymbol{y} = \sum_{i=1}^n \lambda_i y_i^2$. Further, $\det O = \det O^T = 1$, so $\mathrm{d}^n\boldsymbol{x} = \mathrm{d}^n\boldsymbol{y}$. Hence:
\begin{equation}
\int \mathrm{d}^n\boldsymbol{x}\exp \left(- \boldsymbol{x}^T A \boldsymbol{x}\right) = \int \mathrm{d}^n\boldsymbol{y}\exp \left(- \sum_{i=1}^n \lambda_i y_i^2\right)
= \prod_{i=1}^{n} \sqrt{\frac{\pi}{\lambda_i}} = \frac{\pi^{n/2}}{\sqrt{\det A}}
\end{equation}
using equation \ref{eq:1dgaussian} and that $\det A = \prod_{i=1}^n \lambda_i$.
\end{proof}
\end{thm}

\begin{thm} \textbf{(Grassmann Gaussian integral)}
Let $A$ be an $n\times n$ symmetric, positive-definite matrix and $\boldsymbol{\psi}, \boldsymbol{\tilde{\psi}}$ be vectors of Grassmann variables. Then
\begin{equation}
\int \mathrm{d}^n\boldsymbol{\tilde{\psi}}\mathrm{d}^n\boldsymbol{\psi}\exp \left( \boldsymbol{\tilde{\psi}}^T A \boldsymbol{\psi		}\right) = \det A.
\end{equation}
\end{thm}
\begin{proof}
We use the one-dimensional relation for $a \in \mathbb{R}$:
\begin{equation}
\int \dd\tilde{\psi} \,\dd \psi\,  \exp(a \tilde{\psi}\psi) = a
\end{equation}
which follows from a simple Taylor expansion and the Grassmann integration rules.

The rest of the proof is identical to the real-variable case.
\end{proof}

We now prove the important (real-variable) Fresnel integral, which is a similar integral but with imaginary exponent.
\begin{thm} \textbf{(Fresnel integral)} 
Let $A$ be an $n\times n$ symmetric, positive-definite matrix. Then
\begin{equation}
\int \mathrm{d}^n\boldsymbol{x}\exp \left(i \boldsymbol{x}^T A \boldsymbol{x}\right) =\frac{(\pi i)^{n/2}}{ \sqrt{\det A}} = \frac{\pi^{n/2}}{ \sqrt{\det( A/i)}} = \frac{1}{ \sqrt{ \det\left( \frac{A}{\pi i}\right)}}.
\end{equation}
\end{thm}
\begin{proof}
We shall prove that for $a>0$:
\begin{equation} \label{eq:1dfresnel}
\int \dx \exp \left(iax^2\right) = \sqrt{\frac{\pi i }{a}}.
\end{equation}
The full result then follows similarly to the proof for the Gaussian.\\
\\
We evaluate a contour integral of the holomorphic function $f(z) = \exp(iaz^2)$. The contour is a circular sector of radius $R$: $\Gamma = \Gamma_1 \cup \Gamma_2 \cup \Gamma_3$: 

%figure starts here
	
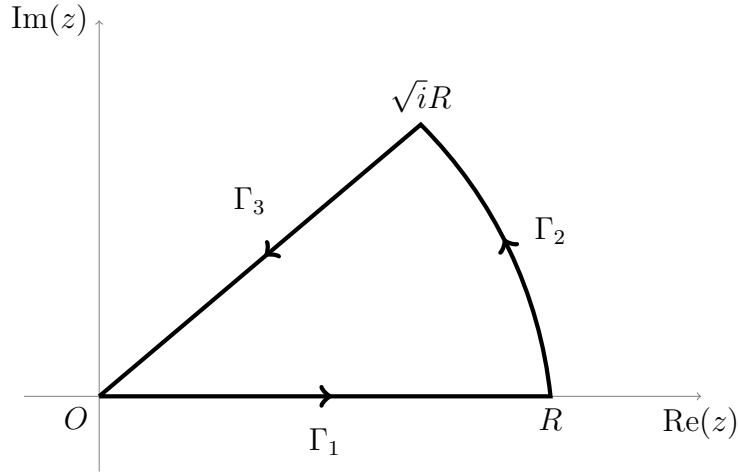
\begin{center}
	\begin{tikzpicture}[decoration={markings,
	mark=at position 3.1cm with {\arrow[line width=2pt]{>}},
	mark=at position 8.2cm with {\arrow[line width=2pt]{>}},
	mark=at position 12.8cm with {\arrow[line width=2pt]{>}},
	}
	]
	% The axes
	\draw[help lines,->] (-1,0) -- (8,0) coordinate (xaxis);
	\draw[help lines,->] (0,-1) -- (0,5) coordinate (yaxis);

	% The path
	\path[draw,line width=1.6pt,postaction=decorate] (0,0) -- (6,0) node[below] {$R$} arc (6:45:6) node[above] {$\sqrt{i}R$} -- (0,0);

	% The labels
	\node[below] at (xaxis) {$\text{Re}(z)$};
	\node[left] at (yaxis) {$\text{Im}(z)$};
	\node[below left] {$O$};
	\node at (2,2.6) {$\Gamma_3$};
	\node at (6,2.2) {$\Gamma_2$};
	\node at (3,-0.6) {$\Gamma_1$};
	\end{tikzpicture}
	\end{center}
%figure ends here

As $f$ is holomorphic, then by Cauchy's theorem:
\begin{equation}
\oint_{\Gamma} f(z) dz = \int_{\Gamma_1} f(z) dz + \int_{\Gamma_2} f(z) dz + \int_{\Gamma_3} f(z) dz = 0.
\end{equation}
Now use the explicit parameterisation of these contour parts to get that
\begin{equation}
\int_{\Gamma_1} f(z) dz = \int_0^R \dx\exp(iax^2)
\end{equation}
and
\begin{equation}
\int_{\Gamma_3} f(z) dz = \sqrt{i} \int_R^0 \mathrm{d}r \,\exp(ia(\sqrt{i}r)^2)= - \sqrt{i} \int_0^R\dx \exp(-ax^2).
\end{equation}
I will prove that 
\begin{equation}
\lim_{R\rightarrow \infty} \int_{\Gamma_2} f(z) dz = 0.
\end{equation}
Then by the previous equations:
\begin{equation}
\int_0^{\infty}\dx \exp(iax^2) =\int_{\Gamma_1} f(z) dz = -\int_{\Gamma_3} f(z) dz = \sqrt{i} \int_0^{\infty}\dx \exp(-ax^2) = \frac{1}{2}\sqrt{\frac{\pi i}{a}}.
\end{equation}
As the integrand is even, the result \ref{eq:1dfresnel} follows.\\
\\
We parameterise $\Gamma_2$ as $z = R \exp(it)$ with $t \in [0,\pi/4]$. Then $\mathrm{d}z = iR \exp(it) \dt$, so
\begin{equation}
\int_{\Gamma_2} f(z) dz = \int_0^{\pi/4} \dt iR \exp(it)\exp\left(iaR^2e^{2it}\right).
\end{equation}
We show the modulus of this integral is $\mathcal{O}(1/R)$, from which the result follows. To do this, we will use \emph{Jordan's lemma} from complex analysis, which states that $\frac{\sin t}{t} > \frac{2}{\pi}$ for $t \in [0,\pi/2]$. Therefore:
\begin{align}
\left|\int_0^{\pi/4} \dt iR \exp(it)\exp\left(iaR^2e^{2it}\right)\right| &\leq R \int_0^{\pi/4}\dt \left|\exp(it)\exp\left(iaR^2e^{2it}\right)\right| \nn\\
&= R \int_0^{\pi/4}\dt \left|\exp\left(iaR^2(\cos(2t) + i \sin(2t)\right)\right| \nn\\
&= R \int_0^{\pi/4}\dt \left|\exp\left(-aR^2\sin(2t)\right)\right| \left|\exp\left(i aR^2\cos(2t)\right)\right| 
\end{align}
where we used Euler's formula $e^{iz} = \cos z + i \sin z$. Now make the substitution $u = 2t$ and use that $\exp\left(-aR^2\sin(2t)\right) >0$ to get:
\begin{align}
\left| \int_{\Gamma_2} f(z) dz\right| &\leq \frac{R}{2} \int_0^{\pi/2} \mathrm{d}u \,\exp\left(-aR^2\sin u\right)\nn\\
&\leq \frac{R}{2} \int_0^{\pi/2}\mathrm{d}u\, \exp\left(-\frac{2aR^2}{\pi} u\right) \nn\\
&=  \frac{\pi}{4aR}( 1- e^{-aR^2})\leq \frac{\pi}{4aR} = \mathcal{O}(1/R)
\end{align}
using Jordan's lemma in the second-to-last line.
\end{proof}

\section{Campbell-Baker-Haussdorf formula} \label{sec:CBH}
To prove that 
\begin{equation}
\exp\left(-\frac{i\epsilon}{\hbar}(\hat{T} + \hat{V})\right) = \exp\left(-\frac{i\epsilon}{\hbar}\hat{T}\right)\exp\left(-\frac{i\epsilon}{\hbar}\hat{V}\right) + \mathcal{O}(\epsilon^2),
\end{equation}
we shall use the Campbell-Baker-Haussdorf formula \cite{Mackenzie,Blau}:

\begin{lemma} \label{lemma:CBHformula}
\textbf{(Campbell-Baker-Haussdorf formula)} Let $X,Y$ be two linear operators and let $\left[ X,Y \right]$ denote their commutator. Define $Z$ by $e^Z = e^Xe^Y$. Then $Z$ satisfies 
\begin{equation} \label{eq:CBHformula}
Z = X + Y + \frac{1}{2} [X,Y] + \frac{1}{12} \left([X,[X,Y]] + [Y,[X,Y]] \right) + \dots
\end{equation}
with $\dots$ denoting terms of order $4$ or higher in $X$ and $Y$.
\end{lemma}

We use this lemma to prove the following:
\begin{propn}
Let $\hat{T}$ and $\hat{V}$ be two operators and let $\hat{A}$ be defined by
\begin{equation} \label{eq:Xdefn}
\exp\left(-\frac{i\epsilon}{\hbar}(\hat{T} + \hat{V})\right) = \exp\left(-\frac{i\epsilon}{\hbar}\hat{T}\right)\exp\left(-\frac{i\epsilon}{\hbar}\hat{V}\right)\exp(\hat{A}).
\end{equation}
Then
\begin{equation} \label{eq:exponentialcorrection}
\hat{A} = \left(\frac{\epsilon}{\hbar}\right)^2\left( \frac{1}{2} [\hat{T}, \hat{V}] + \mathcal{O}(\epsilon^2) \right).
\end{equation}

\begin{proof}
Define $\hat{A}' = - \left(\frac{\hbar}{\epsilon}\right)^2 \hat{A}$ and rearrange equation \ref{eq:Xdefn} to get
\begin{equation}
\exp\left(-\frac{i\epsilon}{\hbar}(\hat{T} + \hat{V})\right)\exp\left( \left(\frac{\epsilon}{\hbar}\right)^2 \hat{A}' \right) = \exp\left(-\frac{i\epsilon}{\hbar}\hat{T}\right)\exp\left(-\frac{i\epsilon}{\hbar}\hat{V}\right).
\end{equation}

Applying the Campbell-Baker-Haussdorf formula to both sides yields
\begin{equation}
\exp\left( -\frac{i\epsilon}{\hbar}(\hat{T} + \hat{V}) + \left(\frac{\epsilon}{\hbar}\right)^2\hat{A}' + \mathcal{O}(\epsilon^3) \right) = \exp\left( -\frac{i\epsilon}{\hbar}(\hat{T} + \hat{V}) - \frac{1}{2}\left(\frac{\epsilon}{\hbar}\right)^2 [\hat{T},\hat{V}] + \mathcal{O}(\epsilon^3) \right).
\end{equation}

Expand both exponentials and equate terms of order $\epsilon^2$ to get:
\begin{equation}
\hat{A}' = -\frac{1}{2}[\hat{T},\hat{V}] + \mathcal{O}(\epsilon)
\end{equation}
which gives the stated result.
\end{proof}

\end{propn}

Thus indeed
\begin{equation}
\exp\left(-\frac{i\epsilon}{\hbar}(\hat{T} + \hat{V})\right) = \exp\left(-\frac{i\epsilon}{\hbar}\hat{T}\right)\exp\left(-\frac{i\epsilon}{\hbar}\hat{V}\right) + \mathcal{O}(\epsilon^2).
\end{equation}

\section{Zeta-regularized constant infinite product}
In this section, we prove the following lemma:

\begin{lemma}
Under zeta-regularization, the following formula holds for any constant $b\in\mathbb{C}\setminus \{ 0\}$:
\begin{equation}
\prod_{n\geq 1} b = b^{-1/2}.
\end{equation}
\end{lemma}
\begin{proof}
This is the determinant of the operator $A_b = b \mathds{1}$ acting on a separable Hilbert space.

It has spectral zeta function:
\begin{equation}
\zeta_{A_b}(s) = \sum_{n\geq 1} b^{-s} = b^{-s} \sum_{n\geq 1} 1 = b^{-s} \zeta(0) = -\frac{1}{2} b^{-s}
\end{equation}
where we used the identity: $\zeta(0) = -\frac{1}{2}$. Hence
\begin{equation}
\zeta'_{A_b}(s) = \frac{1}{2} b^{-s} \log(b) 
\end{equation}

Therefore by definition \ref{defn:zetadet}
\begin{equation}
\prod_{n\geq 1} b = \det A_b = \exp\left(-\zeta'_{A_b}(0)\right) = \exp \left(-\frac{1}{2} \log b \right) = b^{-1/2}
\end{equation}
as required.
\end{proof}

\end{document}